\documentclass[11pt]{article}

\usepackage{mdframed}
\usepackage{tikz}

\usetikzlibrary{plotmarks}
\usetikzlibrary{shadows}
\usetikzlibrary{decorations.pathreplacing}
\usetikzlibrary{positioning}
\usetikzlibrary{scopes}
\usetikzlibrary{arrows}
\usetikzlibrary{calc}
\usetikzlibrary{decorations.pathmorphing}
\usetikzlibrary{decorations.pathreplacing,decorations.markings}
\usetikzlibrary{patterns}

\tikzset{snake it/.style={decorate, decoration={snake,segment length=7pt, amplitude=1pt}}}

\tikzstyle{vertex}=[circle, draw,fill=gray!30, inner sep=0pt, minimum size=16pt]
\tikzstyle{svertex}=[circle, draw,fill=gray!30, inner sep=0pt, minimum size=10pt]
\tikzstyle{sgvertex}=[circle, draw,fill=gray!15, inner sep=0pt, minimum size=10pt]
\tikzstyle{ssgvertex}=[circle, draw,fill=gray!15, inner sep=0pt, minimum size=6pt]
\tikzstyle{sssgvertex}=[circle, draw,fill=gray!15, inner sep=0pt, minimum size=4pt]
\tikzstyle{sssagvertex}=[circle, draw,fill=gray!15, inner sep=0pt, minimum size=3.2pt]
\tikzstyle{ssssgvertex}=[circle, draw,fill=gray!15, inner sep=0pt, minimum size=2pt]
\tikzstyle{cutb} = [fill=blue!30]
\tikzset{
>= stealth'}

\ifnum\pdfshellescape=1 \usetikzlibrary{external}
\fi

\newcommand{\bl}[1]{ #1}

\newcommand{\bll}[1]{ #1}

\def\showauthornotes{1}

\def\showkeys{0}
\def\showdraftbox{0}
\def\showcolorlinks{0}
\def\usemicrotype{1}
\def\showfixme{0}

\usepackage{caption}
\usepackage{subfig}
\usepackage{etex}

\usepackage[utf8]{inputenc}

\usepackage{xspace,enumerate}
\usepackage[shortlabels]{enumitem}

\usepackage[]{xcolor}

\usepackage[T1]{fontenc}

\usepackage[american]{babel}

\usepackage{mathtools}

 \usepackage{bm}

\usepackage{amsthm}

\newtheorem{theorem}{Theorem}[section]
\newtheorem*{theorem*}{Theorem}

\newtheorem*{proposition*}{Proposition}
\newtheorem{lemma}[theorem]{Lemma}
\newtheorem*{lemma*}{Lemma}
\newtheorem{corollary}[theorem]{Corollary}
\newtheorem*{conjecture*}{Conjecture}
\newtheorem{fact}[theorem]{Fact}
\newtheorem*{fact*}{Fact}

\newtheorem*{hypothesis*}{Hypothesis}

\theoremstyle{definition}
\newtheorem{definition}[theorem]{Definition}

\newtheorem{example}[theorem]{Example}

\newtheorem{openquestion}[theorem]{Open Question}

\theoremstyle{remark}
\newtheorem{claim}[theorem]{Claim}
\newtheorem*{claim*}{Claim}
\newtheorem{remark}[theorem]{Remark}
\newtheorem*{remark*}{Remark}

\newtheorem*{observation*}{Observation}

\usepackage[letterpaper,
	top=1.2in,
	bottom=1.2in,
	left=1.4in,
	right=1.4in]{geometry}

\usepackage[varg]{pxfonts}

\ifnum\showkeys=1
\usepackage[color]{showkeys}
\fi

\definecolor{OliveGreen}{rgb}{0,0.6,0}
\ifnum\showcolorlinks=1
\usepackage[
pagebackref,
colorlinks=true,
urlcolor=blue,
linkcolor=blue,
citecolor=OliveGreen,
]{hyperref}
\fi

\ifnum\showcolorlinks=0
\fi

\usepackage{prettyref}

\newcommand{\savehyperref}[2]{\texorpdfstring{\hyperref[#1]{#2}}{#2}}

\newrefformat{eq}{\savehyperref{#1}{\textup{(\ref*{#1})}}}
\newrefformat{lem}{\savehyperref{#1}{Lemma~\ref*{#1}}}
\newrefformat{def}{\savehyperref{#1}{Definition~\ref*{#1}}}
\newrefformat{thm}{\savehyperref{#1}{Theorem~\ref*{#1}}}
\newrefformat{cor}{\savehyperref{#1}{Corollary~\ref*{#1}}}
\newrefformat{cha}{\savehyperref{#1}{Chapter~\ref*{#1}}}
\newrefformat{sec}{\savehyperref{#1}{Section~\ref*{#1}}}
\newrefformat{app}{\savehyperref{#1}{Appendix~\ref*{#1}}}
\newrefformat{tab}{\savehyperref{#1}{Table~\ref*{#1}}}
\newrefformat{fig}{\savehyperref{#1}{Figure~\ref*{#1}}}
\newrefformat{hyp}{\savehyperref{#1}{Hypothesis~\ref*{#1}}}
\newrefformat{alg}{\savehyperref{#1}{Algorithm~\ref*{#1}}}
\newrefformat{rem}{\savehyperref{#1}{Remark~\ref*{#1}}}
\newrefformat{item}{\savehyperref{#1}{Item~\ref*{#1}}}
\newrefformat{step}{\savehyperref{#1}{step~\ref*{#1}}}
\newrefformat{conj}{\savehyperref{#1}{Conjecture~\ref*{#1}}}
\newrefformat{fact}{\savehyperref{#1}{Fact~\ref*{#1}}}
\newrefformat{prop}{\savehyperref{#1}{Proposition~\ref*{#1}}}
\newrefformat{prob}{\savehyperref{#1}{Problem~\ref*{#1}}}
\newrefformat{claim}{\savehyperref{#1}{Claim~\ref*{#1}}}
\newrefformat{relax}{\savehyperref{#1}{Relaxation~\ref*{#1}}}
\newrefformat{red}{\savehyperref{#1}{Reduction~\ref*{#1}}}
\newrefformat{part}{\savehyperref{#1}{Part~\ref*{#1}}}

\newcommand{\Sref}[1]{\hyperref[#1]{\S\ref*{#1}}}

\usepackage{nicefrac}

\ifnum\usemicrotype=1
\usepackage{microtype}
\fi

\ifnum\showauthornotes=1
\newcommand{\Authornote}[2]{{\sffamily\small\color{red}{[#1: #2]}}}
\newcommand{\Authornotecolored}[3]{{\sffamily\small\color{#1}{[#2: #3]}}}
\newcommand{\Authorcomment}[2]{{\sffamily\small\color{gray}{[#1: #2]}}}
\newcommand{\Authorstartcomment}[1]{\sffamily\small\color{gray}[#1: }

\newcommand{\Authorfnote}[2]{\footnote{\color{red}{#1: #2}}}
\newcommand{\Authorfixme}[1]{\Authornote{#1}{\textbf{??}}}
\newcommand{\Authormarginmark}[1]{\marginpar{\textcolor{red}{\fbox{\Large #1:!}}}}
\else
\newcommand{\Authornote}[2]{}
\newcommand{\Authornotecolored}[3]{}
\newcommand{\Authorcomment}[2]{}
\newcommand{\Authorstartcomment}[1]{}

\newcommand{\Authorfnote}[2]{}
\newcommand{\Authorfixme}[1]{}
\newcommand{\Authormarginmark}[1]{}
\fi

\ifnum\showfixme=0

\fi

\usepackage{boxedminipage}

\newcommand{\Brac}[1]{\left[#1\right]}

 \usepackage{dsfont}
\usepackage{mathrsfs}

\newcommand{\textparen}[1]{\text{(#1)}}

\ifx\because\undefined
\newcommand{\because}[1]{\textparen{because #1}}
\else
\renewcommand{\because}[1]{\textparen{because #1}}
\fi

\renewcommand{\vec}[1]{{\bm{#1}}}

\newcommand\bdot\bullet

\newcommand{\cost}[0]{w}

\DeclareMathOperator{\lb}{lb}

\DeclareMathOperator{\LP}{LP}
\DeclareMathOperator{\DUAL}{DUAL}

\DeclareMathOperator{\supp}{supp}

\newcommand{\Z}{\mathbb Z}

\newcommand{\R}{\mathbb R}

\newcommand{\cA}{\mathcal A}

\newcommand{\cC}{\mathcal C}

\newcommand{\cI}{\mathcal I}

\newcommand{\cL}{\mathcal L}

\renewcommand{\leq}{\leqslant}
\renewcommand{\le}{\leqslant}
\renewcommand{\geq}{\geqslant}
\renewcommand{\ge}{\geqslant}

\ifnum\showdraftbox=1
\newcommand{\draftbox}{\begin{center}
  \fbox{\begin{minipage}{2in}\begin{center}\Large\textsc{Working Draft}\\Please do not distribute\end{center}\end{minipage}}\end{center}
\vspace{0.2cm}}
\else
\newcommand{\draftbox}{}
\fi

\let\epsilon=\varepsilon

\numberwithin{equation}{section}

\newcommand{\MYstore}[2]{\global\expandafter \def \csname MYMEMORY #1 \endcsname{#2}}

\newcommand{\MYload}[1]{\csname MYMEMORY #1 \endcsname }

\newcommand{\MYnewlabel}[1]{\newcommand\MYcurrentlabel{#1}\MYoldlabel{#1}}

\newcommand{\MYdummylabel}[1]{}

\newcommand{\torestate}[1]{\let\MYoldlabel\label \let\label\MYnewlabel #1\MYstore{\MYcurrentlabel}{#1}\let\label\MYoldlabel }

\newcommand{\restatetheorem}[1]{\let\MYoldlabel\label
  \let\label\MYdummylabel
  \begin{theorem*}[Restatement of \prettyref{#1}]
    \MYload{#1}
  \end{theorem*}
  \let\label\MYoldlabel
}

\newcommand{\restatelemma}[1]{\let\MYoldlabel\label
  \let\label\MYdummylabel
  \begin{lemma*}[Restatement of \prettyref{#1}]
    \MYload{#1}
  \end{lemma*}
  \let\label\MYoldlabel
}

\newcommand{\restateprop}[1]{\let\MYoldlabel\label
  \let\label\MYdummylabel
  \begin{proposition*}[Restatement of \prettyref{#1}]
    \MYload{#1}
  \end{proposition*}
  \let\label\MYoldlabel
}

\newcommand{\restatefact}[1]{\let\MYoldlabel\label
  \let\label\MYdummylabel
  \begin{fact*}[Restatement of \prettyref{#1}]
    \MYload{#1}
  \end{fact*}
  \let\label\MYoldlabel
}

\newcommand{\restate}[1]{\let\MYoldlabel\label
  \let\label\MYdummylabel
  \MYload{#1}
  \let\label\MYoldlabel
}

\let\origparagraph\paragraph
\renewcommand{\paragraph}[1]{\origparagraph{#1.}}

\allowdisplaybreaks

\usepackage{arydshln}
\usepackage{bbm}
\newcommand{\one}{\mathbbm{1}}

\DeclareMathOperator*{\con}{\cC}

\DeclareMathOperator*{\low}{\mathrm{ind}}

\DeclareMathOperator{\valu}{value}
\DeclareMathOperator{\lbs}{\overline{\lb}}

\newcommand{\EPC}{Subtour Partition Cover}

\newcommand{\ceil}[1]{\left\lceil #1 \right\rceil}

\newcommand{\xs}{x}

\newcommand{\cLmax}{\cL_\textrm{max}}

\newcommand{\level}[0]{\textrm{level}}
\newcommand{\deltaval}{0.78}
\newcommand{\finalval}{506}
\newcommand{\finalvalFS}{1012}
\newcommand{\finalintegralitygap}{319}
\newcommand{\finalintegralitygapKTV}{1273}
\newcommand{\finalvalKTV}{2021}
\newcommand{\hide}[1]{}
\newcommand{\nw}[0]{\ensuremath{\alpha_{\textrm{\tiny S}}}\xspace}
\newcommand{\wD}{\beta}

\usepackage[capitalize]{cleveref}

\title{\bf A Constant-Factor Approximation Algorithm for the Asymmetric Traveling Salesman Problem\footnote{
	This paper is the joint journal version of the conference
        publications \cite{Svensson15} and \cite{SvenssonTV18}.}}

	\newcommand{\fgsfdsspace}{\hspace{-0.8em}}
	\author{\fgsfdsspace Ola Svensson\thanks{\'Ecole Polytechnique
            F\'ed\'erale de Lausanne. Supported by 
the European Research Council (ERC) under the European Union's Horizon 2020 research and innovation programme (grant agreement 335288--OptApprox) and the Swiss National Science Foundation project
200021-184656 ``Randomness in Problem Instances and Randomized Algorithms''.}\\
  \texttt{ola.svensson@epfl.ch} \fgsfdsspace\\
	\and
	\fgsfdsspace Jakub Tarnawski\thanks{\'Ecole Polytechnique
            F\'ed\'erale de Lausanne. Supported by 
the European Research Council (ERC) under the European Union's Horizon 2020 research and innovation programme (grant agreement 335288--OptApprox).} \\ \texttt{jakub.tarnawski@gmail.com}\fgsfdsspace \\
	\and
	\fgsfdsspace L\'aszl\'o A. V\'egh\thanks{London School of
          Economics and Political Science.  Supported by EPSRC First Grant EP/M02797X/1 and 
the European Research Council (ERC) under the European Union's Horizon 2020 research and innovation programme (grant agreement 757481--ScaleOpt).}\\ \texttt{l.vegh@lse.ac.uk} 	\fgsfdsspace \\
}

\begin{document}

	\maketitle
	\draftbox
	\thispagestyle{empty}

\begin{abstract}
  We give a constant-factor approximation algorithm for the asymmetric
  traveling salesman problem (ATSP). Our approximation guarantee is analyzed with
  respect to the standard LP relaxation, and thus our result
  confirms the conjectured  constant integrality gap of that relaxation.

The main idea of our approach is a reduction to Subtour Partition
Cover, an easier problem obtained by significantly relaxing the general connectivity
requirements into local connectivity conditions. We first show that
any algorithm for Subtour Partition Cover
 can be turned into an algorithm 
for ATSP while only losing a small constant factor in the performance guarantee.
Next, we present a reduction from general ATSP
instances to structured instances, on which we then solve  Subtour Partition Cover, yielding our constant-factor approximation algorithm for ATSP. 
\end{abstract}

	\medskip
	\noindent
	{\small \textbf{Keywords:}
	approximation algorithms, asymmetric traveling salesman problem, combinatorial optimization, linear programming
	}

	\clearpage

\tableofcontents

\section{Introduction}
The traveling salesman problem---to find the shortest tour visiting
$n$ given cities---is  one of the best-known NP-hard optimization problems.

Without any assumptions on the distances, a simple reduction from the problem
of deciding whether a graph is Hamiltonian shows that it is NP-hard to
approximate the shortest tour to within any factor.   
Therefore it is common to relax the problem by allowing the tour to visit
cities more than once. This is equivalent to assuming that the distances satisfy
the triangle inequality:  the distance from  city $i$ to $k$ is no larger than
the distance from $i$ to $j$ plus the distance from $j$ to $k$. 
All results mentioned and proved in this paper refer to this
setting.

If we also assume the distances to be symmetric, then Christofides' classic
algorithm from 1976~\cite{Ch76},
also discovered independently by Serdyukov~\cite{Serdyukov78,Bevern2020},
is guaranteed  to find a tour of length at most
$\nicefrac{3}{2}$ times the optimum. Improving this approximation guarantee is
a notorious open question in approximation algorithms. There has been a flurry of recent
progress in the special case of unweighted graphs~\cite{GharanSS11,MomkeS16,Mucha12,SeboV14}. However, even though the standard linear programming (LP) relaxation
is conjectured to approximate the optimum within a factor of $\nicefrac{4}{3}$, it remains an elusive problem to improve upon the Christofides--Serdyukov algorithm. 

If we do not restrict ourselves to symmetric distances (undirected graphs), we obtain the more
general \emph{asymmetric} traveling salesman problem (ATSP).
Compared to the symmetric setting, the gap in our understanding is much larger, and the current algorithmic techniques have failed to give \emph{any} constant
approximation guarantee. This is intriguing especially since
the standard LP relaxation, also known as the \emph{Held--Karp lower bound},  is conjectured to approximate the optimum to within a small
constant. In fact, it is only known that its integrality gap\footnote{Recall that the integrality gap is defined as  the maximum ratio between
the optimum values of the exact (integer) formulation and of its relaxation.} is at least $2$~\cite{CharikarGK06}. We also note that the best known inapproximability bound for ATSP is $75/74$~\cite{KarpinskiLS13}.

One can easily show that a multiset of edges forms a feasible solution
to ATSP if and only if it is connected and Eulerian.
The first approximation algorithm for ATSP was given
by Frieze, Galbiati and
Maffioli~\cite{FriezeGM82},  achieving an approximation guarantee of  $\log_2(n)$. 
Their elegant ``repeated cycle cover'' approach maintains an Eulerian
edge multiset throughout, but it initially relaxes connectivity, and
enforces connectivity gradually.
This approach
was refined in several papers~\cite{Blaser08,KaplanLS05,FeigeS07},  but there was no
\emph{superconstant} improvement in the approximation guarantee  until the more recent 2010 $O(\log
n \ / \log\log n)$-approximation algorithm by Asadpour et al.~\cite{AsadpourGMGS10}. They
introduced a new and influential approach to ATSP based on relaxing
the Eulerian degree constraints but maintaining connectivity
throughout. The key idea is establishing a connection to the graph-theoretic concept of thin spanning trees. This  has further led to improved algorithms
for special cases of ATSP, such as graphs of bounded genus~\cite{GharanS11}. Moreover,
Anari and Oveis Gharan recently exploited this connection to significantly improve the
best known upper bound on the integrality gap of the standard LP relaxation to
$O(\textrm{poly}\log\log n)$~\cite{AnariG15}. This implies an efficient
algorithm for estimating the optimal value of a tour within a factor
$O(\textrm{poly}\log\log n)$ but,  as their arguments are non-constructive,
no approximation algorithm for finding a tour of matching guarantee. 

\medskip

In this paper, we follow an approach more akin to the one by Frieze, Galbiati and
Maffioli~\cite{FriezeGM82}. Namely, we maintain the Eulerian degree
constraints but relax the connectivity requirements, by introducing a
problem called {\em Subset Partition Cover}. We prove our main
theorem using this auxiliary problem.
\begin{theorem}
  There is a polynomial-time algorithm for ATSP that returns a tour of
  value at most $\finalval$ times the Held--Karp lower bound.
  \label{thm:constantATSP}
\end{theorem}
This paper is a joint version of the two conference
publications \cite{Svensson15} and \cite{SvenssonTV18}.
The paper \cite{Svensson15} introduced the relaxed problem
{\em Local-Connectivity ATSP} and used it to obtain a constant-factor
approximation algorithm for unweighted digraphs, and more generally
for \emph{node-weighted graphs},
i.e., graphs whose weight function
can be written as $w(u,v) = f(u) + f(v)$ for some $f : V \to \R_+$.\footnote{
	In \cite{Svensson15}, the definition is slightly different: $w(u,v)=g(u)$ for every $(u,v)\in E$ for a
        function $g:V\to \R_+$. The two definitions are equivalent: by
        assigning $g(u)=2 f(u)$, the weight of any tour is equal for the
        weights $w(u,v)=g(u)$ and for the weights $w(u,v)=f(u)+f(v)$.}
 The Subset Partition Cover problem described in this paper is a more refined version of
 Local-Connectivity ATSP; the terminology has been changed in order to
 emphasize the differences between the problems. 

The publication \cite{SvenssonTV18} proved
Theorem~\ref{thm:constantATSP} with an approximation ratio of 5500. The
significant improvement in the guarantee presented here is due to
using the more refined Subset Partition Cover problem along with some
more efficient reductions.

We remark that we can obtain a tighter upper bound of $\finalintegralitygap$ for the integrality gap of the Held--Karp relaxation,
and that our results also imply a constant-factor approximation algorithm for the Asymmetric Traveling Salesman \emph{Path} Problem
via black-box reductions, given by Feige and Singh for the
approximation guarantee \cite{FeigeS07} and recently by K{\"{o}}hne,
Traub and Vygen \cite{KohneTV18} for the integrality gap---see \cref{sec:completepuzzle}.

In a recent development, Traub and Vygen \cite{Traub2020} have improved the approximation guarantee to $22+\varepsilon$, and the integrality gap to $22$. Their results are attained by improving and simplifying the techniques in this paper. They use a simpler chain of reductions by avoiding the use of irreducible instances inherent in our argument, as well as refining other parts of the reduction.
See the conclusions (Section~\ref{sec:conclusion}) for further discussion and open problems.

\subsection{Brief overview of approach and outline of paper}

It will be convenient to define ATSP in terms of its \emph{graphic formulation}:
\begin{definition}\label{def:ATSP}
  The input for ATSP is a pair $(G,w)$, where $G$ is a strongly connected directed graph (digraph) and $w$ is a nonnegative weight function defined on the edges. The
  objective is to find a closed walk of minimum weight that visits every vertex
  at least once. 
\end{definition} 
Another standard definition of ATSP asks for a minimum-weight Hamiltonian cycle, that is, a closed walk that visits every vertex exactly once. If the weight function satisifies the triangle inequality, these two forms are equivalent: a tour that visits every vertex at least once can be shortcut to a tour visiting every vertex exactly once, without increasing the weight of the tour. Throughout the paper, we use the graphic formulation as in Definition~\ref{def:ATSP}.

Without loss of generality, one could assume that
$G$ is a complete digraph.
However, for our reductions,
it will be important that $G$ may not be complete. We also remark that a closed walk that visits every vertex at least once is equivalent to an Eulerian multiset of edges that connects the graph.  (An edge set of a digraph is Eulerian if the in-degree of each vertex equals its out-degree.) 

The first main step of our argument is introducing the problem
\emph{\EPC{}} in \cref{sec:LCdef}.
The main technical contribution of Part~I is a reduction
which shows that
if one can solve \EPC{} on some class of graphs,
then one can obtain a constant-factor approximation for ATSP on that class of graphs
(cf.~\cref{thm:LocalToGlobal}).
In \cite{Svensson15}, a constant-factor
approximation algorithm was obtained for ATSP on node-weighted graphs
by solving \EPC{} (more precisely, the earlier variant
 Local-Connectivity ATSP) for this class.
Subsequently, \cite{SvenssonTV16} solved Local-Connectivity ATSP for
graphs with two different edge-weights.
This required a rather difficult technical argument, and it appears
to be very challenging to solve \EPC{} directly on general graphs.

In
this paper, we follow a different approach.
Before applying the reduction (\cref{thm:LocalToGlobal}),
we remain in the realm of ATSP
and use a series of natural reductions to gradually simplify the structure of instances that we are dealing with.
The first of these reductions (in \cref{sec:lam-weight}) crucially uses the laminar structure
arising from the Held--Karp relaxation and its dual linear program (see \cref{lem:lam-support,thm:laminar}).
All further reductions are described in Part~II.
The most structured instances, for which we apply the reduction to \EPC{}
and on which we then solve \EPC{}
in Part~III,
are called \emph{vertebrate pairs}.

\medskip

The outline of the paper is as follows.
In \cref{sec:prelim}
we present preliminaries and introduce notation used throughout the paper.
We introduce the standard Held--Karp relaxation
in \cref{sec:relaxation}.
\cref{sec:lam-weight} is devoted to our first reduction.
There, we show that we can focus
on \emph{laminarly-weighted ATSP instances}: there is a laminar family $\cL$ of vertex
sets and a nonnegative vector $(y_S)_{S\in \cL}$ such that any edge $e$ has  $w(e) = \sum_{S\in \cL: \, e\in
\delta(S)} y_S$.
See the left part of Figure~\ref{fig:intro} for an example.
Note that the special case
when the laminar family consists only of singletons
roughly corresponds to node-weighted instances.
We call laminarly-weighted ATSP instances where $\cL \subseteq \{ \{ v \} : v \in V \}$
\emph{singleton instances}.\footnote{Every singleton instance is a node-weighted instance, but not vice versa; see the discussion following \cref{def:singleton}.}

Next, in Part~I we define the \EPC{} problem,
where the connectivity requirements are relaxed in comparison to ATSP,
and reduce the task of solving ATSP (with a constant-factor approximation) to that problem. We also solve \EPC{} on singleton instances (\cref{thm:lcapprox}),
thus illustrating the power of the reduction
as well as
developing a tool necessary later in Part~II.

In Part~II we turn our attention back to ATSP and,
starting from laminarly-weighted instances,
show that we can obtain very structured ATSP instances
called vertebrate pairs
by only increasing the approximation guarantee by a constant factor. A vertebrate pair consists of a laminar instance and a subtour $B$, called the backbone,  that crosses every non-singleton set of $\cL$.  An example is depicted on the right part of Figure~\ref{fig:intro}.

\begin{figure}[t]
  \centering
  \begin{tikzpicture}
\tikzset{arrow data/.style 2 args={decoration={markings,
         mark=at position #1 with \arrow{#2}},
         postaction=decorate}
      }

  \begin{scope}[scale=0.8]
    \draw[fill=gray!15!white, draw=gray!80!black] (-2.5, -0.5) ellipse (1.5cm and 2cm); \begin{scope}
      \draw[fill=gray!40!white, draw=gray!80!black] (-2.5, 0) ellipse (0.75cm and 1cm); \end{scope}
    \begin{scope}[xshift=2.0cm,rotate=30]
      \draw[fill=gray!15!white, draw=gray!80!black] (0, 0) ellipse (1cm and 1.5cm); \draw[fill=gray!40!white, draw=gray!80!black,rotate=-5] (0, -.75) ellipse (0.5cm and 0.6cm); \end{scope}
      \draw[fill=gray!15!white, draw=gray!80!black,rotate=-5] (0.5, -1.70) ellipse (0.5cm and 0.6cm); \draw[fill=gray!40!white, draw=gray!80!black,rotate around={-55:(0.45, -1.65)}] (0.45, -1.65) ellipse (0.4cm and 0.25cm); \node[ssssgvertex, fill=black] (u) at (-2.4, 0.5) {};
    \node[ssssgvertex, fill=black] (a) at (-2.7, 0.0) {};
    \node[ssssgvertex, fill=black] (b) at (-2.4, -0.5) {};
    \node[ssssgvertex, fill=black] (c) at (-3.0, -1.5) {};
    \node[ssssgvertex, fill=black] (d) at (-2.0, -1.8) {};

    \node[ssssgvertex, fill=black] (e) at (0.2, -2.0) {};
    \node[ssssgvertex, fill=black] (f) at (0.35, -1.5) {};
    \node[ssssgvertex, fill=black] (g) at (0.55, -1.8) {};

    \node[ssssgvertex, fill=black] (h) at (1.6, 0.8) {};
    \node[ssssgvertex, fill=black] (i) at (1.4, 0.4) {};

    \node[ssssgvertex, fill=black] (j) at (2.3, 0.3) {};
    \node[ssssgvertex, fill=black] (k) at (2.3, -0.4) {};
    \node[ssssgvertex, fill=black] (l) at (2.5, -1.0) {};

    \draw (u) edge[thick,->] node[above] {\scriptsize $e$} (h); \end{scope}
  \begin{scope}[xshift=8cm, scale=0.8]
    \draw[fill=gray!15!white, draw=gray!80!black] (-2.5, -0.5) ellipse (1.5cm and 2cm); \begin{scope}
      \draw[fill=gray!40!white, draw=gray!80!black] (-2.5, 0) ellipse (0.75cm and 1cm); \end{scope}
    \begin{scope}[xshift=2.0cm,rotate=30]
      \draw[fill=gray!15!white, draw=gray!80!black] (0, 0) ellipse (1cm and 1.5cm); \draw[fill=gray!40!white, draw=gray!80!black,rotate=-5] (0, -.75) ellipse (0.5cm and 0.6cm); \end{scope}
      \draw[fill=gray!15!white, draw=gray!80!black,rotate=-5] (0.5, -1.70) ellipse (0.5cm and 0.6cm); \draw[fill=gray!40!white, draw=gray!80!black,rotate around={-55:(0.45, -1.65)}] (0.45, -1.65) ellipse (0.4cm and 0.25cm); \node[ssssgvertex, fill=black] (u) at (-2.4, 0.5) {};
    \node[ssssgvertex, fill=black] (a) at (-2.7, 0.0) {};
    \node[ssssgvertex, fill=black] (b) at (-2.4, -0.5) {};
    \node[ssssgvertex, fill=black] (c) at (-3.0, -1.5) {};
    \node[ssssgvertex, fill=black] (d) at (-2.0, -1.8) {};

    \node[ssssgvertex, fill=black] (e) at (0.2, -2.0) {};
    \node[ssssgvertex, fill=black] (f) at (0.35, -1.5) {};
    \node[ssssgvertex, fill=black] (g) at (0.55, -1.8) {};

    \node[ssssgvertex, fill=black] (h) at (1.6, 0.8) {};
    \node[ssssgvertex, fill=black] (i) at (1.4, 0.4) {};

    \node[ssssgvertex, fill=black] (j) at (2.3, 0.3) {};
    \node[ssssgvertex, fill=black] (k) at (2.3, -0.4) {};
    \node[ssssgvertex, fill=black] (l) at (2.5, -1.0) {};

    \draw [black, thick] plot [smooth,tension=1.0] coordinates { (h) (u) (d)  (f) (l) (j) };
    \draw (j) edge[bend right=7, thick] (h);
\end{scope}
\end{tikzpicture}
   \caption{On the left we give an example of a laminarly-weighted  ATSP instance. The sets of the laminar family are shown in gray. We depict a single edge $e$ that crosses three sets in the laminar family, say  $S_1$, $S_2$, $S_3$, and so $w(e) = y_{S_1} + y_{S_2} + y_{S_3}$.
On the right, we give an example of a vertebrate pair. Notice that the backbone (depicted as the cycle) crosses all non-singleton sets of the laminar family, though it may not visit all the vertices.}
  \label{fig:intro}
\end{figure}

Finally,
in Part~III
we give an algorithm for \EPC{}
on such instances. By the aforementioned reductions, solving \EPC{} for vertebrate pairs is sufficient for obtaining a constant-factor approximation algorithm for general ATSP. We combine all the ingredients and calculate the obtained ratio in \cref{sec:completepuzzle}. 
We discuss future research directions in  \cref{sec:conclusion}.

\section{Preliminaries}
\label{sec:prelim}

For a directed graph $G$, we let $V(G)$ and $E(G)$ denote the set of
vertices and edges, respectively. All graphs in the paper will be directed; we will use the term `graph' to refer to a directed graph.
We  use simply $V$ and $E$ whenever the graph is
clear from the context. By an edge set $F\subseteq E$, we always mean
an \emph{edge multiset}: the same edge can be present in multiple
copies. We will refer to the union of two edge (multi)sets as a
multiset (adding up the multiplicities of every edge).

For  vertex sets $S,T\subseteq V$ we let $\delta(S,T)=\{(u,v)\in E:
u\in S\setminus T, v\in T\setminus S\}$. For a set $S\subseteq V$ we let
$\delta^+(S)=\delta(S,V\setminus S)$ denote the set
of outgoing edges, and we let $\delta^-(S)=\delta(V\setminus S,S)$
denote the set of incoming edges. Further, let $\delta(S)=\delta^-(S)\cup \delta^+(S)$
be all boundary edges
and $E(S) = \{ (u,v) \in E : u,v \in S \}$
be interior edges
of a vertex set $S$.
For a vertex $v\in V$ we define
$\delta^+(v)=\delta^+(\{v\})$ and $\delta^-(v)=\delta^-(\{v\})$.
For an edge (multi)set $F\subseteq E$, we use $\delta_F(S,T)=\delta(S,T)\cap F$,
$\delta^+_F(S)=\delta^+(S)\cap F$, etc.
We let $V(F)$ denote the set of vertices
incident to at least one edge in $F$,
and
$\one_F$
denote the indicator vector of $F$, which has a coordinate for each
edge $e$ with value equal to the multiplicity of $e$ in $F$.

For a  vertex set $U \subsetneq V$, we let $G[U] = (U,E(U))$
denote the subgraph induced by $U$.  That is, $G[U]$ is the  subgraph of $G$
whose vertex set is $U$ and whose edge set consists of all edges in
$E(G)$ with both endpoints in $U$.
We also let $G/U$ denote the graph obtained by contracting the vertex
set $U$, i.e., by replacing the vertices in $U$ by
a single new vertex $u$ and redirecting every edge with one endpoint
in $U$ to point from/to the new vertex $u$.  
This may create parallel edges in $G/U$. We keep all parallel copies;
thus, every edge in $G/U$ will have a unique preimage in $G$. 

For a set $S\subsetneq V$ we let $S_{\mathrm{in}}$  and $S_{\mathrm{out}}$ be those vertices of $S$ that have an incoming edge from outside of $S$ and those that have an outgoing edge to outside of $S$, respectively. That is,
\begin{align*}
  S_\mathrm{in} = \{ v \in S : \delta^-(S) \cap \delta^-(v) \ne \emptyset \} \quad \mbox{and} \quad S_\mathrm{out} = \{ v \in S : \delta^+(S) \cap \delta^+(v) \ne \emptyset \}\,.
\end{align*}

We let $\R_+$ denote the set of nonnegative real numbers.
The \emph{support} of a function/vector $f:X\to \R_+$ is the subset $\{x\in
X: f(x)>0\}$. For a subset $Y\subseteq X$, we also use
$f(Y)=\sum_{x\in Y} f(x)$.

When talking about graphs, we shall
slightly abuse notation and sometimes write $w(G)$ instead of $w(E)$ and $f(G)$
instead of $f(V)$ when it is clear from the context that $w$ and $f$ are
functions on the edges and vertices.

Finally, a closed walk will be called a subtour:

\begin{definition} \label{def:subtour}
We call $F \subseteq E$
a {\em subtour}
if $F$ is Eulerian
(we have $|\delta^+_F(v)| = |\delta^-_F(v)|$ for every $v \in V$)
and
the graph $(V(F), F)$ is connected.
By convention, $F = \emptyset$ is a subtour. A {\em tour} is a subtour
$F$ with $V(F)=V$.
\end{definition}
Therefore a subtour is an Eulerian multiset of edges that form a single connected component
(or an empty set), and a tour is an Eulerian multiset of edges that connects the graph.

For any Eulerian multiset $F$ of edges,
we refer to the connected components of $(V(F), F)$
as \emph{subtours in $F$}.
We often refer to $F$ as a collection of subtours.
Note that if $T$ is a subtour in $F$, then $T \ne \emptyset$.

We say that a subtour $T$ intersects another subtour $T'$
if we have $V(T) \cap V(T') \ne \emptyset$.

\subsection{Held--Karp Relaxation}
\label{sec:relaxation}
Given an edge-weighted digraph $(G,w)$, the Held--Karp
relaxation has a variable $x(e) \geq 0$ for every edge $e\in E$. The intended solution is that $x(e)$
should equal the number of times $e$ is used in the solution. The  linear programming relaxation
$\LP(G,w)$  is now defined as follows:
\begin{equation}
\begin{aligned}
\arraycolsep=1.4pt\def\arraystretch{1.2}
\begin{array}{lrlr}
\mbox{minimize} \qquad & \displaystyle \sum_{e\in E}  w(e)x(e) \\[7mm]
\mbox{subject to} \qquad  & \displaystyle x(\delta^+(v)) = & \displaystyle x(\delta^-(v)) & \text{ for } v\in V, \\  
& \displaystyle x(\delta(S)) \geq & 2 & \text{ for } \emptyset \neq S \subsetneq V, \\
& x \geq & 0.
\end{array}
  \end{aligned}
\tag{$\LP(G,w)$}
\end{equation}
The optimum value of this LP  is called the \emph{Held--Karp lower bound}.
The first set of constraints says that the in-degree should equal the
out-degree for each vertex, i.e., the solution should be Eulerian. We
call a non-negative vector $x$ satisfying these constraints a \emph{circulation}. The
second set of constraints enforces that the solution is connected. 
These are sometimes referred to as subtour elimination
constraints. Notice that the Eulerian
property implies  $x(\delta^-(S))=x(\delta^+(S))$ for every set
$S \subseteq V$, and therefore these
constraints are equivalent to $x(\delta^+(S))
\geq 1$ for all $\emptyset \neq S \subsetneq V$, which appear more
frequently in the literature. We use the above
formulation as it enables some simplifications in the presentation.

We say that a set $S\subseteq V$ is \emph{tight} with respect to a
solution $x$ to $\LP(G,w)$ if $x(\delta(S))=2$, that is,  $x(\delta^-(S))=x(\delta^+(S))=1$. 

Let us now  formulate the dual linear program
$\DUAL(G,w)$.   We associate variables $(\alpha_v)_{v\in V}$  and $(y_S)_{\emptyset \neq
  S \subset V}$ with the first and second set of constraints of $\LP(G,w)$, respectively.

\begin{equation}
\begin{aligned}
  \arraycolsep=1.4pt\def\arraystretch{1.2}
  \begin{array}{lrlr}
  \mbox{maximize} \qquad & \displaystyle \sum_{\emptyset \neq S \subsetneq V} 2 \cdot y_S \qquad \qquad  \\[7mm]
  \mbox{subject to} \qquad  & \displaystyle \sum_{S: (u,v) \in \delta(S)} y_S + \alpha_u - \alpha_v \leq & \displaystyle w(u,v) & \quad \text{ for } (u,v)\in E, \\  
  & y \geq & 0.
  \end{array}
  \end{aligned}
\tag{$\DUAL(G,w)$}
\end{equation}

For  singleton sets $\{u\}$, we will also use the notation
$y_u=y_{\{u\}}$.

The Held--Karp relaxation has exponentially many constraints, but it can be solved in polynomial time using the ellipsoid method with a separation oracle.
 Moreover, an optimal solution to $\DUAL(G,w)$ can also be found in polynomial
time. We now briefly explain how the ellipsoid method can be applied.
Let $P$ be the feasible region. Given a point $x\in \R^E$, we can decide whether $x\in P$ by checking the Eulerian constraints and computing the minimum cut value. \bll{If $x$ is not feasible, this yields a separating inequality that is one of the inequalities of the system $\LP(G,w)$}.  We have a \emph{well-described polyhedron} as in Definition 6.2.2 in \cite{gls}: $P$ is an $m$-dimensional polyhedron defined by inequalities of encoding length at most $2m+3$. 
Theorems 6.4.9 and 6.5.14-15 in \cite{gls} show that optimal primal and dual solutions can be found using a polynomial number of oracle calls. Since the encoding length is $O(m)$, and separation can be done in strongly polynomial time, the overall running time is strongly polynomial. We note that an optimal primal solution can also be found by formulating an equivalent compact (polynomial-size) linear program~\cite{Arthanari82,Carr96}; but such a reduction would not directly yield a dual optimal solution.

The ellipsoid method provides an optimal solution 
to $\DUAL(G,w)$ with a polynomial-size support. We will need a dual optimal solution with a ``nice'' support, as stated next. Recall that a family ${\cL}\subseteq
2^V$ of
vertex subsets is \emph{laminar} if for any $A,B\in{\cL}$ we have either
$A\subseteq B$ or $B\subseteq A$ or $A\cap B=\emptyset$.

\begin{lemma}\label{lem:lam-support}
For every edge-weighted digraph $(G,w)$ there exists an optimal solution $(\alpha,y)$ to $\DUAL(G,w)$ such that
the support of $y$ is a laminar family of vertex subsets. Moreover,
such a solution can be computed in polynomial time.
\end{lemma}
We note that the existence of a laminar solution was also used previously 
in~\cite{Vempala1999}. \bll{For finding one in polynomial time we invoke a result by Karzanov~\cite{Karzanov}.}
\begin{proof}
We start by showing the existence of a laminar optimal solution
using a standard uncrossing argument (see e.g.~\cite{Cornuejols1985} for an early application of this technique to the Held--Karp relaxation of the symmetric traveling salesman problem).
 Select $(\alpha, y)$ to be an optimal solution to $\DUAL(G,w)$
 minimizing $\sum_{S} |S|y_S$. That is, among all dual solutions that maximize the dual objective $2 \sum_S y_S$, we select one that minimizes $\sum_S |S| y_S$.  We claim that the support $\cL = \{S: y_S > 0 \}$ is a laminar family. Suppose not, i.e., that there are sets $A, B \in \cL$ such that $A \cap B, A \setminus B, B\setminus A \neq \emptyset$. Then we can obtain a new dual solution $(\alpha, \hat y)$, where  $\hat y$ is defined, for $\epsilon = \min (y_A, y_B) > 0$, as
  \begin{align*}
    \hat y_S = \begin{cases}
      y_S - \epsilon & \mbox{ if $S = A$ or $S=B$,} \\
      y_S + \epsilon & \mbox{ if $S = A\setminus B$ or $S = B \setminus A$,} \\
      y_S & \mbox{ otherwise.}
    \end{cases}
  \end{align*}
  That $(\alpha, \hat y)$ remains a feasible solution follows since
  $\hat y$ remains non-negative (by the selection of $\epsilon$) and since
  for any edge $e$  we have 
  \begin{align*}
    \one_{e\in \delta(A)} + \one_{e\in \delta(B)} \geq \one_{e\in \delta(A\setminus B)} + \one_{e\in \delta(B\setminus A)}\,.
  \end{align*}
Therefore $\sum_{S: e\in \delta(S)} \hat y_S \leq  \sum_{S: e\in \delta(S)}  y_S$ and so the constraint corresponding to edge $e$ remains satisfied. Further, we clearly have $2\sum_S \hat y_S = 2 \sum_S y_S$. In other words, $(\alpha, \hat y)$ is an optimal dual solution. However,
  \begin{align*}
    \sum_S |S| ( y_S - \hat y_S ) = (|A| + |B| - |A\setminus B| - |B \setminus A|) \epsilon > 0 \,,
  \end{align*}
  which contradicts that $(\alpha, y)$ was selected to be an optimal dual solution minimizing $\sum_S |S| y_S$. Therefore, there can be no such sets $A$ and $B$ in  $\cL$, and so it is a laminar family.

To find a laminar optimal solution in polynomial time, we start with
an arbitrary dual optimal solution. As noted above, one can be
computed in polynomial time. Now we apply
the above uncrossing operation to obtain a laminar optimal
solution. 
A result by Karzanov \cite[Theorem 2]{Karzanov}
shows that if we carefully select the sequence of pairs $A, B$ to uncross, this
can be performed in polynomial time (although for an arbitrary
sequence, the number of uncrossing steps may not be polynomially bounded).
\end{proof}

\subsection{Laminarly-Weighted ATSP and Singleton Instances} \label{sec:lam-weight}

In this section we show that
without loss of generality
(i.e., without any loss in the approximation factor)
we can focus on instances whose weights come from a sparse
and highly structured family of sets.
Below we define the crucial notion of laminarly-weighted instances
that we will work with throughout the paper.

\begin{definition}\label{def:lam}A tuple $\cI=(G,\cL,x,y)$ is called a
  \emph{laminarly-weighted ATSP instance} if $G$ is a
  strongly connected digraph, $\cL$ is a laminar family of
vertex subsets, $x$ is a feasible solution
to $\LP(G,\mathbf{0})$, and $y: \cL\to \R_+$. We further require that 
$x_e>0$ for every $e\in E$ and that every set $S \in \cL$ be tight with respect to $x$,
i.e., that $x(\delta^+(S)) = x(\delta^-(S)) = 1$.
We define the \emph{induced weight function} $w_{\cI} : E \to \R_+$ as
\[
w_{\cI}(e) = \sum_{S\in \cL: \ e\in\delta(S)} y_S \qquad \mbox{for every
}e\in E.
\]
\end{definition}
Here, $\mathbf{0}$ denotes the zero weight function.

Given an instance $\cI$ as in the definition, the vectors $x$ and
$y$ have the following important property. Define a dual solution $(\bar \alpha, \bar y)$ by setting $\bar
\alpha_u=0$ for all $u\in V$, and $\bar y_S=y_S$ if $S\in \cL$ and
$\bar y_S=0$ otherwise. Then complementary slackness implies that  for
the induced weight function $w_{\cI}$, the vector $x$
is an optimal solution to $\LP(G,w_{\cI})$ and $(\bar\alpha,\bar y)$ is an optimal solution to
$\DUAL(G,w_{\cI})$.

Our first main insight is that ATSP with arbitrary weights can be
reduced to the laminarly-weighted ATSP problem.

\begin{theorem}\label{thm:laminar}
Assume we have a polynomial-time algorithm that
finds a solution of weight
at most $\alpha$ times the Held--Karp lower bound for every laminarly-weighted ATSP instance. Then there is a polynomial-time
algorithm for the general ATSP problem  that finds a solution of weight
at most $\alpha$ times the Held--Karp lower bound.
\end{theorem}
\begin{proof}
Consider an arbitrary edge-weighted strongly connected digraph $(G,w)$. Let $x$ be an
optimal solution to $\LP(G,w)$ and let $(\alpha,y)$ be an optimal solution to $\DUAL(G,w)$ as guaranteed by
Lemma~\ref{lem:lam-support}, that is, $y$ has a laminar support $\cL$. We now define a pair $(G',w')$ as 
\[V(G') = V(G)\,, \quad E(G')  = \{e\in E(G): x(e)> 0\}\,, \quad \mbox{and}\quad  w'(u,v)  = w(u,v)
-\alpha_u + \alpha_v\,.\]

We claim that  $\cI=(G',\cL,x,y)$ is a laminarly-weighted ATSP
instance whose induced weight function $w_\cI$ equals $w'$.
To see this, recall that $x$ is a primal optimal solution and that $(\alpha,y)$ is a dual optimal solution (for $(G,w)$). Therefore complementary slackness implies that every
set in $\cL$ is tight with respect to $x$ and that for every
edge $(u,v)\in E(G')$, the weight $w'(u,v) = w(u,v) - \alpha_u + \alpha_v $ equals the sum of $y_S$-values for
the sets $S$ crossed by $(u,v)$. Finally, we have $x_e > 0$ for every $e\in E(G')$ by definition.  So $\cI$ satisfies all the properties of Definition~\ref{def:lam}, i.e., it is a laminarly-weighted instance.

We now argue that an $\alpha$-approximate solution for $\cI$ with respect to the Held--Karp relaxation $\LP(G',w')$ implies an $\alpha$-approximate solution for the original instance $(G,w)$ with respect to $\LP(G,w)$. To this end, we make the following observation:
\begin{claim*}\label{claim:cost-shift}
For any circulation $x \in \R_+^{E(G')}$,
we have $\sum_{e \in E(G')} w(e)x(e) = \sum_{e \in E(G')} 
w'(e)x(e)$.
\end{claim*}
Therefore
the Held--Karp lower bound is the same for $(G,w)$ and for $(G',w')$, and 
any solution for $(G',w')$ is a solution of the same weight for
$(G,w)$.
\end{proof}

In the rest of the paper we work exclusively with laminarly-weighted ATSP
instances $\cI = (G, \cL, x, y)$. We will refer to them as simply
\emph{instances}.
\begin{definition} \label{def:singleton}
We say that an instance
$\cI = (G, \cL, x, y)$
is a \emph{singleton instance}
if 
all sets in $\cL$ are singletons.
\end{definition}
Such instances will play an
important role in our algorithm. In particular, note that for singleton
instances, the weight
function $w_{\cI}$ is induced by nodes
($w(u,v) = y_u + y_v$ for all $(u,v) \in E$)
just like in a node-weighted instance.
The difference between singleton and node-weighted instances
is that singleton instances are those
\emph{laminarly-weighted} instances
whose weight function is induced by nodes
after having performed the reduction of \cref{thm:laminar}.
A node-weighted instance does not necessarily give rise to a singleton instance:
for singleton instances we will also require that $x(\delta^+(v))=x(\delta^-(v))=1$ for every node $v$ with $y_v>0$.

Recall that $\cost_\cI(F)$ is the induced weight of an
edge multiset $F\subseteq E$ in the
instance $\cI$. We will omit the subscript and use simply $\cost(F)$
whenever $\cI$ is clear from the context.
\begin{definition}
For 
an instance
$\cI = (G, \cL, x, y)$
and
a set $S \subseteq V$ we define
\[ \valu_\cI(S) = 2 \cdot \sum_{R \in \cL : \ R \subsetneq S} y_R \]
to be the fractional dual value associated with the sets strictly inside $S$.
\end{definition}
Again, we will omit  the subscript whenever clear from the context.
We also use
$\valu(\cI) = \valu_\cI(V)$; note that this equals the Held--Karp lower
bound of the instance. Indeed, as noted above, $y$ can be extended to
an optimal dual solution to $\DUAL(G,w)$, and hence the optimum value for
$\DUAL(G,w)$ equals $2\cdot \sum_{S\in \cL} y_S$, which is equal to the
primal optimum value  $\sum_{e\in E} w(e) x(e)$ for $\LP(G,w)$ by
strong duality.

\part{Reducing ATSP to \EPC}
\label{part:epc}
In this part we define the \EPC{} problem and reduce the task of solving  ATSP to that problem.  This reduction will be used to solve general instances in Part~\ref{part:solvingvert}. Here, we illustrate its power by giving a constant-factor approximation algorithm for singleton instances.

Let us begin with some intuition. It is illustrative to consider the following
``naive'' algorithm:
\begin{enumerate}
  \item Select a random cycle cover $C$ using the Held--Karp relaxation.

    It is well known that one can  sample such a cycle cover $C$ of   expected
    weight equal to the optimal value $\valu(\cI)$  of the
    Held--Karp relaxation.
  \item While there exist more than one component, add the lightest cycle (i.e., the cycle of smallest weight) that
    decreases the number of components. 
\end{enumerate}
It is clear that the above algorithm always returns a solution to ATSP: we
start with an Eulerian graph, and the graph stays Eulerian during the execution
of the while-loop, which does not terminate until the graph is connected. This
gives a tour.
However, what is its weight? First, as remarked above, we
have that the expected weight of the cycle cover is $\valu(\cI)$. So if $C$
contains $k =|C|$ cycles, we would expect that a cycle in $C$ has weight
$\valu(\cI)/k$ (at least on average).  Moreover,  the number of cycles added in
Step~$2$ is at most $k-1$ since each cycle decreases the number of components
by at least one.  Thus, if each cycle in Step~$2$ has weight at most  the average
weight $\valu(\cI)/k$ of a cycle in $C$,  we obtain a $2$-approximate tour
of weight at most $\valu(\cI) + \frac{k-1}{k} \valu(\cI) \leq 2 \valu(\cI)$.

Unfortunately, it seems hard to find a cycle cover $C$ such that we can always
connect it with light cycles. Instead, what we can do is to first select
a cycle cover $C$, then add light cycles that decrease the number of components
as long as possible.  When there are no more light cycles to add, the
vertices are partitioned into connected components $V_1, \ldots, V_k$.
In order to make progress from this point, we would like to find a ``light''
Eulerian set $F$ of edges that crosses  the cuts $\{(V_i, \bar V_i)\mid i=1,2,
\ldots, k\}$. We could then hope to add $F$ to our solution and continue from
there. It turns out that the meaning of ``light'' in this
context is crucial.  For our arguments to work, we need that $F$ is selected so that  the
edges in each component have weight at most $\alpha$ times what
the linear programming solution ``pays'' for the vertices in that component. This is the intuition behind
the definitions of \EPC{} (formerly Local-Connectivity ATSP) and
``light'' algorithms for that problem. We also need to be very careful as
to how we add edges from light cycles and how to use the light
algorithm for \EPC{}. In \cref{sec:LocalToGlobal}, our
algorithm will iteratively solve \EPC{} and, in each
iteration, it will  add a carefully chosen subset of the found edges, together with
light cycles.

We remark that in \EPC{} we have relaxed the global
connectivity properties of ATSP into local connectivity conditions that only
say that we need to find an Eulerian set of edges that crosses at most $n=|V|$
cuts defined  by a partitioning of the vertices. In spite of that, we are able
to leverage the intuition above to obtain our main technical result
of this part (\cref{thm:LocalToGlobal}).

Its proof is based on
generalizing  and, as alluded to above, deviating from the above intuition in several ways. First, we  start with
a carefully chosen collection of subtours which  generalizes the role of the
cycle cover $C$ in Step~1 above. Second, both the iterative use of the light
algorithm for \EPC{} and the way we add light cycles  are 
done  in a careful and dependent manner so as to be able to bound the total
weight of the returned solution.

\section{\EPC}
\label{sec:LCdef}

In this section we define the {\EPC{}} problem,
which is obtained from ATSP by relaxing the connectivity
requirements.
Consider a laminarly-weighted instance $\cI = (G, \cL, x, y)$. 
For notational convenience we extend the vector $y$ to all singletons
so that $y_v = 0$ if $\{v\}\not \in \cL$.  Let $\lb_{\cI}: V \rightarrow \R$
be the \emph{lower bound function} defined by $\lb_{\cI}(v) =2
y_v$. We simplify notation and write $\lb$ instead of $\lb_{\cI}$ if
$\cI$ is clear from the context. Note that $\lb(V)$ is at most the
Held--Karp lower bound $\valu(\cI)$, with equality only for singleton
instances.
For an edge set $F$, we use the simplified notation $\lb(F)=\lb(V(F))$ to denote the total lower bound of the vertices
incident to $F$.

Perhaps the main difficulty of ATSP is to satisfy the connectivity requirement,
i.e., to select an Eulerian subset $F$ of edges that satisfies all
subtour elimination constraints.
In \emph{\EPC}, this condition is relaxed and we only require that
the subtour elimination constraints be satisfied for some disjoint sets.

\begin{mdframed}
\begin{center} \textbf{\EPC} \end{center}
\begin{description}
  \item[\textnormal{\emph{Given:}}]
An instance $\cI = (G, \cL, x, y)$,
a subtour $B$ in $G$,
and a partition $(V_1, V_2, \ldots, V_k)$
	of $V\setminus V(B)$ such that the graph induced by $V_i$ is strongly
	connected for $i=1,\ldots, k$.  
\item[\textnormal{\emph{Find:}}] A collection $F$ of subtours  of $E$   such
that $|\delta^+_{F}(V_i)| \geq 1$ for  $i=1, 2, \ldots, k$.
\end{description}
\end{mdframed}
\vspace{2mm}

An example of an instance and a solution to \EPC{} can be found in Figure~\ref{fig:algo} on page~\pageref{fig:algo}. In that figure, the vertices $V \setminus V(B)$ are partitioned into six sets depicted in gray (the right-most  set is $V(B)$) and the blue (solid) subtours show a solution. Note that the subtours are not required to connect the whole graph; they are only required to cross the boundaries defined by the partitioning of $V \setminus V(B)$.

\begin{definition} \label{def:light_algorithm}
We say that an algorithm for \EPC{} is 
\emph{$(\alpha,\wD)$-light}  for an instance $\cI$ and subtour $B$ if, for any input partition of strongly connected
subsets, the collection $F$ of subtours satisfies
\begin{itemize}
\item $w_{\cI}(T)\le \alpha \lb(T)$ for every subtour $T$ in $F$ with   $V(T)\cap
  V(B)=\emptyset$, and 
\item $w_{\cI}(F_B)\le \wD$, where $F_B\subseteq F$ is the collection of subtours in $F$ that intersect $B$.\footnote{Recall that we say that a subtour $T$ intersects another subtour $B$ if they visit a common vertex, i.e., $V(T) \cap V(B) \neq \emptyset$. Hence, $F_B = \{T \mbox{ subtour in F}: V(T) \cap V(B) \neq \emptyset\}$.}
\end{itemize}
\end{definition}

We use the $(\alpha,\wD)$-light
terminology to avoid any ambiguities with the concept of approximation
algorithms. 

If we let $c$ be the scaling factor such that $c\lb(V) = \valu(\cI)$, then   
an $\alpha$-approximation algorithm for ATSP with respect to the Held--Karp
relaxation is trivially an $(\alpha\cdot c, 0)$-light algorithm for \EPC{} with $B=\emptyset$: output the same tour $F$ as the
algorithm for ATSP. However, \EPC{} seems like a significantly easier problem than ATSP, as the set of subtours $F$ only needs to cross $k$ cuts formed  by a partitioning of the vertices $V\setminus V(B)$.
We substantiate this intuition by proving,  in Section~\ref{sec:lcATSPalg},
that there exists a simple $(2,0)$-light algorithm for
\EPC{} on singleton instances with $B = \emptyset$.  Perhaps more
surprisingly, in Section~\ref{sec:LocalToGlobal} we show that an
$(\alpha,\wD)$-light algorithm for \EPC{}  for an instance $\cI$ with
subtour $B$ can be turned into an approximation
algorithm for ATSP that always returns a tour of cost at most
$(9+\varepsilon)\alpha\lb(V\setminus V(B))+\wD+w(B)$ for
any $\varepsilon>0$.

The main difference between the \EPC{} problem and the Local-Connectivity ATSP problem introduced in~\cite{Svensson15} is the introduction of the subtour $B$ and  the more general definition of lightness. While this flexibility is unnecessary for singleton instances with $B=\emptyset$ (which are closely related to the node-weighted instances considered in that paper), it will be useful in the general case: in Section~\ref{sec:solvingvertebrate} we give an algorithm for \EPC{} that, in turn, implies the constant-factor approximation algorithm for general instances.

\begin{remark}
  \label{rem:lb}Our generic reduction from ATSP to \EPC{} (Theorem~\ref{thm:LoctoGlo}) is robust with respect to
  the definition of $\lb$ and there are many possibilities to define such
  a lower bound. Another natural example is $\lb(v) = \sum_{e\in
    \delta^+(v)} x_e w(e)$. In \cite{Svensson15}, 
Local-Connectivity
ATSP was defined with this $\lb$ function, and with
$B=\emptyset$; in this case we can set $\wD=0$.
In fact, in order to get a constant bound on
  the integrality gap of the Held--Karp relaxation, our results say that it is
  enough to find an $(O(1),0)$-light algorithm for
  \EPC{} with respect to  some nonnegative $\lb$ that only needs to
  satisfy that $\lb(V)$ is at most the value $\valu(\cI)$ of the optimal solution
  to the LP. 
Even more generally,  if $\lb(V)$ is at most the value of an optimal tour
  (rather than the LP value)
  then our methods would give a similar approximation guarantee (but not with
  respect to the Held--Karp relaxation). 

A variant of \EPC{} was used in \cite{SvenssonTV16} to
  obtain a constant factor approximation guarantee for ATSP with two
  different edge weights; a key idea of that paper is the careful
  choice of the $\lb$ function.
\end{remark}

 \section{\EPC{} for Singleton Instances}
\label{sec:lcATSPalg}

We give a simple $(2,0)$-light algorithm for \EPC{} for
the special case of singleton instances, that is, when $\cL$ is a
singleton family, and for $B=\emptyset$.

The proof is based on finding  an integral circulation that sends flow across
the cuts $\{(V_i, \bar V_i): i=1, 2, \ldots, k\}$ and, in addition, satisfies
that the outgoing flow of each vertex $v\in V$ with $y_v > 0$ is at most $2$, which in turn, by the assumptions on the instance,  implies
a $(2,0)$-light algorithm.
\begin{theorem}
  \label{thm:lcapprox}
  There exists a  polynomial-time algorithm for
  \EPC{} that is $(2,0)$-light for singleton instances with $B=\emptyset$.
\end{theorem}
\begin{proof}
  Let $\cI=(G,\cL,x,y), B, (V_1, V_2, \ldots, V_k)$ be an instance of \EPC{} where $\cI$ is a  singleton instance and $B = \emptyset$. Let also $w=w_{\cI}$ denote the induced weight
function.
 We prove the theorem by giving a polynomial-time algorithm that finds a collection $F$ of subtours satisfying 
\begin{align}
  \label{eq:cycprop}
|\delta^+_{F}(V_i)|  \geq 1  \mbox{ for } i =1, \ldots, k \quad \mbox{and} \quad
|\delta^+_{F}(v)| & \leq 2 \mbox{ for $v \in V$ with $y_v > 0$}.
\end{align}
The first condition means that $F$ crosses every cut $(V_i, \bar
V_i)$, thus the algorithm indeed solves the \EPC{} problem.
We show that the second condition implies $(2,0)$-lightness.
Since $B = \emptyset$, we need to show that $w(T) \le 2 \lb(T)$ for every subtour $T$ in $F$.
Since $\cI$ is a singleton instance, $w(u,v)=y_u+y_v$ for all
$(u,v)\in E$ (recall the convention $y_u=0$ if $\{u\}\notin
\cL$). Consider now any subtour $T$ in $F$. We have
\begin{align*}
w(T) = \sum_{e \in T} w(e) =
  \sum_{v\in V(T)}|\delta_F(v)| y_v =
  2 \sum_{v\in V(T)}|\delta^+_F(v)| y_v \le
  4 \sum_{v\in V(T)} y_v =
  2\lb(T).
\end{align*}

We proceed by describing a polynomial-time algorithm for finding an Eulerian
set $F$ satisfying~\eqref{eq:cycprop}.
The set $F$ will be obtained by rounding the circulation $x$ to integrality
while maintaining that it crosses each cut $(V_i, \bar V_i)$.
For each cut $(V_i, \bar V_i)$
we introduce a new auxiliary vertex $a_i$
to represent it.
In lieu of requiring a flow of at least $1$ through $V_i$,
we will require such a flow through $a_i$.
To show that such a (fractional) circulation exists, we modify $x$
by redirecting an arbitrary flow of value $1$ that passes through $V_i$
to instead pass through $a_i$.

First, we transform $G$ into a new graph $G'$
and $x$ into a new circulation $x'$
by performing the following for each $i = 1, \ldots, k$ (see also Figure~\ref{fig:Gaux}):
\begin{itemize}
	\item
		Select a subset of incoming edges $X_i^- \subseteq \delta^-(V_i)$ with $x(X_i^-) = 1$.
		This is possible since $x(\delta^-(V_i))\ge 1$.\footnote{To obtain exactly $1$, we might need to
			break an edge up into two copies,
			dividing its $x$-value between them appropriately,
			and include one copy in $X_i^-$ but not the other;
			we omit this for simplicity of notation, and
			assume there is such an edge set with exactly $x(X_i^-)= 1$.
		}
	\item
		Consider a cycle decomposition of $x$ and follow the incoming edges in $X_i^-$ in the decomposition.
		Select $X_i^+ \subseteq \delta^+(V_i)$
		to be the set of outgoing edges on which these cycles first leave $V_i$
		after entering on an edge in $X_i^-$,
		such that $x(X_i^+) = 1$.\footnote{Again, to obtain exactly $1$, we proceed as in the above footnote.}
		We define a flow $x_i$ to be the $x$-flow on these cycle
		segments connecting  the heads of edges in $X_i^-$ and the tails
		of edges in $X_i^+$. 
	\item
		We introduce a new auxiliary vertex $a_i$
		and redirect all edges in $X_i^-$ to point to $a_i$,
		and those in $X_i^+$ to point from $a_i$.
		We subtract the flow $x_i$ from $x$.
\end{itemize}

\begin{figure}[t!]
\centering

\begin{tikzpicture}[scale=0.65]

\begin{scope}
	\draw[fill=gray!10!white, draw=gray!80!black] (0, 0) ellipse (3.2cm and 1.7cm) node[above right=0.5cm and 1.6cm] {$V_i$};
	\node[sssgvertex, fill=black] (u1) at (-2,-0) {};
	\node[sssgvertex, fill=black] (u2) at (0,0.5) {};
	\node[sssgvertex, fill=black] (u3) at (2,-0) {};
	\draw (u1) edge[->, bend left] node[above] {$e_{12}$} (u2);
	\draw (u2) edge[->, bend left] node[above] {$e_{23}$} (u3);
	\draw (u3) edge[->, bend left=15, decorate,decoration={snake,amplitude=.2mm,segment length=2pt,post length=1mm}] node[below, pos=0.3] {$e_{31}$} (u1);
	\node[sssgvertex, fill=black] (v1) at (-2.5,3) {};
	\node[sssgvertex, fill=black] (v2) at (-0.5,3) {};
	\node[sssgvertex, fill=black] (v3) at (1.5,3) {};
	\node[sssgvertex, fill=black] (w1) at (-2.5,-3) {};
	\node[sssgvertex, fill=black] (w2) at (-0.5,-3) {};
	\node[sssgvertex, fill=black] (w3) at (1.5,-3) {};
	\draw (v1) edge[->, bend left=10] node[left, pos=0.3] {$e^-_1$} (u1);
	\draw (v2) edge[->, bend left=10, very thick] node[right, pos=0.2] {$e^-_2$} (u2);
	\draw (v3) edge[->, bend left=10, very thick] node[right, pos=0.25] {$e^-_3$} (u3);
	\draw (u1) edge[->, bend left=10, very thick] node[left, pos=0.7] {$e^+_1$} (w1);
	\draw (u2) edge[->, bend left=10, very thick] node[right, pos=0.82] {$e^+_2$} (w2);
	\draw (u3) edge[->, bend left=10] node[right, pos=0.75] {$e^+_3$} (w3);
	
	\node at (0,-4) {\textbf{(a)} $x$: $x(e) = \nicefrac 12$ for all $e$};
\end{scope}

\begin{scope}[xshift=11cm]
	\draw[fill=gray!10!white, draw=gray!80!black] (0, 0) ellipse (3.2cm and 1.7cm) node[above right=0.5cm and 1.6cm] {$V_i$};
	\node[sssgvertex, fill=black] (u1) at (-2,-0) {};
	\node[sssgvertex, fill=black] (u2) at (0,0.5) {};
	\node[sssgvertex, fill=black] (u3) at (2,-0) {};
	\draw (u1) edge[->, bend left] node[above left=-0.1] {} (u2);
	\draw (u2) edge[->, bend left] node[above right=-0.1] {} (u3);
\node[sssgvertex, fill=black] (v1) at (-2.5,3) {};
	\node[sssgvertex, fill=black] (v2) at (-0.5,3) {};
	\node[sssgvertex, fill=black] (v3) at (1.5,3) {};
	\node[sssgvertex, fill=black] (w1) at (-2.5,-3) {};
	\node[sssgvertex, fill=black] (w2) at (-0.5,-3) {};
	\node[sssgvertex, fill=black] (w3) at (1.5,-3) {};
	\draw (v1) edge[->, bend left=10] node[left, pos=0.3] {} (u1);
\draw (u3) edge[->, bend left=10] node[right, pos=0.75] {} (w3);
	
	\node[sgvertex,minimum size=14pt] (ai) at (-4,0) {$a_i$};
	\draw (v2) edge[->, bend right] (ai);
	\draw (v3) edge[->, bend right=20] (ai);
	\draw (ai) edge[->, bend right] (w1);
	\draw (ai) edge[->, bend right=20] (w2);
	
	\node at (0,-4) {\textbf{(b)} $x'$: $x'(e) = \nicefrac 12$ for all $e$};
\end{scope}

\begin{scope}[yshift=-9cm]
	\draw[fill=gray!10!white, draw=gray!80!black] (0, 0) ellipse (3.2cm and 1.7cm) node[above right=0.5cm and 1.6cm] {$V_i$};
	\node[sssgvertex, fill=black] (u1) at (-2,-0) {};
	\node[sssgvertex, fill=black] (u2) at (0,0.5) {};
	\node[sssgvertex, fill=black] (u3) at (2,-0) {};
	\draw (u1) edge[->, bend left] node[above left=-0.1] {} (u2);
	\draw (u2) edge[->, bend left] node[above right=-0.1] {} (u3);
\node[sssgvertex, fill=black] (v1) at (-2.5,3) {};
	\node[sssgvertex, fill=black] (v2) at (-0.5,3) {};
	\node[sssgvertex, fill=black] (v3) at (1.5,3) {};
	\node[sssgvertex, fill=black] (w1) at (-2.5,-3) {};
	\node[sssgvertex, fill=black] (w2) at (-0.5,-3) {};
	\node[sssgvertex, fill=black] (w3) at (1.5,-3) {};
	\draw (v1) edge[->, bend left=10] node[left, pos=0.3] {} (u1);
\draw (u3) edge[->, bend left=10] node[right, pos=0.75] {} (w3);
	
	\node[sgvertex,minimum size=14pt] (ai) at (-4,0) {$a_i$};
	\draw (v2) edge[->, bend right] (ai);
\draw (ai) edge[->, bend right] (w1);

	\node at (0,-4) {\textbf{(c)} $z'$ (integral)};
\end{scope}

\begin{scope}[xshift=11cm, yshift=-9cm]
	\draw[fill=gray!10!white, draw=gray!80!black] (0, 0) ellipse (3.2cm and 1.7cm) node[above right=0.5cm and 1.6cm] {$V_i$};
	\node[sssgvertex, fill=black] (u1) at (-2,-0) {};
	\node[sssgvertex, fill=black] (u2) at (0,0.5) {};
	\node[sssgvertex, fill=black] (u3) at (2,-0) {};
	\draw (u1) edge[->, bend left] node[above left=-0.1] {} (u2);
	\draw (u2) edge[->, bend left] node[above right=-0.1] {} (u3);
	\draw (u2) edge[->, bend left=75, dashed] node[above right=-0.1] {} (u3);
	\draw (u3) edge[->, bend left=15, dashed] node[below, pos=0.3] {} (u1);
	\node[sssgvertex, fill=black] (v1) at (-2.5,3) {};
	\node[sssgvertex, fill=black] (v2) at (-0.5,3) {};
	\node[sssgvertex, fill=black] (v3) at (1.5,3) {};
	\node[sssgvertex, fill=black] (w1) at (-2.5,-3) {};
	\node[sssgvertex, fill=black] (w2) at (-0.5,-3) {};
	\node[sssgvertex, fill=black] (w3) at (1.5,-3) {};
	\draw (v1) edge[->, bend left=10] node[left, pos=0.3] {} (u1);
	\draw (v2) edge[->, bend left=10, very thick] node[right, pos=0.2] {} (u2);
\draw (u1) edge[->, bend left=10, very thick] node[left, pos=0.7] {} (w1);
\draw (u3) edge[->, bend left=10] node[right, pos=0.75] {} (w3);
	
	\node at (0.1,0.1) {$u_i$};
	\node at (-2.4,0) {$v_i$};

	\node at (0,-4) {\textbf{(d)} $F$};
\end{scope}

\end{tikzpicture}

\caption{A depiction of the
proof of Theorem~\ref{thm:lcapprox}.
  The neighborhood of a component $V_i$ is shown.
  \\
  \textbf{(a)}~shows $x$, with $x(e) = \nicefrac 12$ on every shown edge.
  We select $X_i^- = \{e_2^-, e_3^-\}$ (thick incoming edges).
  Suppose that the cycle decomposition of $x$ has
  a cycle containing $e_1^-$, $e_{12}$, $e_{23}$, $e_3^+$,
  a cycle containing $e_2^-$, $e_2^+$,
  and
  a cycle containing $e_3^-$, $e_{31}$, $e_1^+$.
  Thus we have $X_i^+ = \{e_1^+, e_2^+\}$ (thick outgoing edges),
  and the flow $x_i$ (wiggly) puts value $\nicefrac 12$ on $e_{31}$.
  \\
  \textbf{(b)}~shows $x'$, with $x'(e) = \nicefrac 12$ on every shown edge.
  We redirect $e_2^-$, $e_3^-$ to point to $a_i$
  and $e_1^+$, $e_2^+$ to point from $a_i$.
We also subtract $x_i$, removing $e_{31}$.
  \\
  \textbf{(c)}~shows $z'$, which is integral.
  Note that $z'(\delta^-(a_i)) = z'(\delta^+(a_i)) = 1$.
  \\
  \textbf{(d)}~shows the final solution $F$.
  The thick edges,
  which are redirected from the edges incident to $a_i$ in $z'$,
  guarantee that $F$ crosses $V_i$.
  The path $P_i$ is dashed.
  }
\label{fig:Gaux}
\end{figure}
 
Note that $x'$ is a circulation,
and that it satisfies the following conditions:
\begin{itemize}
	\item $x'(\delta^+(v)) \le 1$ for all $v \in V$ with $y_v > 0$,
	\item $x'(\delta^+(a_i)) = 1$ for all $i = 1, \ldots, k$.
\end{itemize}
Here, the first condition holds since all sets in $\cL$ are tight and thus $x(\delta^+(v)) = 1$ whenever $y_v > 0$;
the rest is by construction.
As the vertex-degree bounds are integral, 
we can also, in polynomial time, find an \emph{integral} circulation $z'$ that satisfies these two conditions (see e.g.
  Chapter~11~in~\cite{Schrijver03}).

Next, we map $z'$ from $G'$ to a flow $z$ in $G$ in the natural way:
by 
reversing the redirection of
 the edges incident to the auxiliary vertices $a_i$
while retaining their flow.
Now, the flow $z$ so obtained may not be a circulation. Specifically,
since the in- and out-degree of $a_i$ were exactly $1$ in $z'$,
in each component $V_i$ there is a pair of vertices
$u_i$, $v_i$ that are the head and tail, respectively,
of the mapped-back edges adjacent to $a_i$.
These are the only vertices whose in-degree in $z$
may differ from their out-degree.
(They differ unless $u_i = v_i$.)
To repair this,
for each $i = 1, \ldots, k$
we route a path $P_i$ from $u_i$ to $v_i$ in $V_i$;
this is always possible as we assumed that $V_i$ is strongly connected (by the definition of \EPC{}).

We obtain our final solution $F$ from $z$
by taking every edge $e \in E$
with multiplicity $z_e$
and adding the paths $P_i$.
Note that $F$ is Eulerian, i.e, a collection of subtours.
To see that $F$ satisfies~\eqref{eq:cycprop},
note that $F$ crosses every cut $(V_i, \bar V_i)$
(since the edges redirected from $a_i$
are boundary edges of $V_i$).
Moreover, by the first property above
and the fact that paths $P_i$ are vertex-disjoint
(being inside disjoint subsets $V_i$),
we have $|\delta_F^+(v)| \le 2$
for each $v \in V$ with $y_v > 0$.
This concludes the proof of Theorem~\ref{thm:lcapprox}.

\end{proof}

 \section{From Local to Global Connectivity}
\label{sec:LocalToGlobal}
In this section, we reduce the task of approximating ATSP to that of solving \EPC{}. 
To simplify the notation, for a subtour $B$, we
let 
\begin{equation}\label{eq:lb-B}
\lb_\cI(\bar B)=\lb_\cI(V\setminus V(B))=2\sum_{v\in V\setminus V(B)} y_v.
\end{equation}
The main theorem can be stated as follows.

\begin{theorem} \label{thm:LocalToGlobal}
Let $\cA$ be an algorithm for \EPC{}.
For any instance $\cI=(G,\cL,x,y)$ and subtour $B$,
if $\cA$ is $(\alpha,\wD)$-light for $\cI$ and $B$,
then there exists a tour of $G$ of weight at most
$5\alpha \lb_\cI(\bar B)+\wD +w_{\cI}(B)$.
Moreover, for any $\varepsilon >0$, a tour of weight at most
  $9(1+\varepsilon)\alpha \lb_\cI(\bar B) +\wD +w_{\cI}(B)$ can be found in time polynomial in the number
  $n=|V|$ of vertices, in $1/\varepsilon$,  and in the running time of $\cA$.
  \label{thm:LoctoGlo}
\end{theorem}
Using Theorem~\ref{thm:lcapprox}, we immediately obtain a constant-factor approximation for ATSP on singleton instances.
\begin{corollary}\label{cor:singleton}
For any $\varepsilon>0$, there exists a polynomial-time
$(18+\varepsilon)$-approximation algorithm for ATSP on singleton instances.
\end{corollary}
In the sequel, we will use $\nw = 18 + \varepsilon$ for the approximation ratio for ATSP
on singleton instances to make the dependence on this factor
transparent.

Throughout this section we let $\cI=(G,\cL,x,y)$, $B$ and $\cA$ be fixed as in the
statement of the theorem; we let $w=w_{\cI}$ throughout. The proof of the theorem is by giving an algorithm
that uses $\cA$ as a subroutine. We first give the non-polynomial-time algorithm (with the better guarantee) in
Section~\ref{sec:ExistDescAlg}, followed by Section~\ref{sec:polyalg} where we
modify the arguments so that we also efficiently find a tour (with a slightly worse guarantee).

\subsection{Existence of a Good Tour}
\label{sec:ExistDescAlg}
The idea of the  algorithm is to start with a collection of subtours and then
iteratively merge/connect them into a single tour that visits all vertices  by adding additional  (cheap) subtours. We remark that since we
will only add Eulerian subsets of edges, the algorithm always maintains a collection of subtours. So  
the {state of the algorithm} is described by a collection of subtours $T^*$.

\paragraph{Initialization}
For the rest of the section, we assume that $V(B)\neq V$. Otherwise, $B$ itself is a tour and the
algorithm simply returns $B$.
The algorithm starts by selecting non-empty subtours $T^*_1, T^*_2, \ldots, T^*_k$ such that
\begin{enumerate}[label=I\arabic*:, ref=I\arabic* ]
  \item $B, T^*_1, T^*_2, \ldots, T^*_k$ are disjoint subtours; \label{en:i1}
  \item $w(T^*_i) \leq 2\alpha \lb(T^*_i)$ for $i=1,2, \ldots, k$; \label{en:i2}
  \item the tuple  $\langle \lb(T^*_1), \lb(T^*_2), \ldots, \lb(T^*_k) \rangle$ is lexicographically maximal.  \label{eq:lexord} \label{en:i3}
\end{enumerate}
As the lexicographic order is maximized, the subtours are ordered so
that $\lb(T^*_1) \geq \lb(T^*_2) \geq \cdots \geq \lb(T^*_k)$. 
The set $T^*$ is initialized 
as $T^* = T^*_0\cup T^*_1 \cup T^*_2 \cup \cdots \cup T^*_k$, where we let $T^*_0 = B$.

During the
execution of the algorithm we will also use the following concept. 
For a
subtour $T$ of $G$,  let $\low(T)$ be the smallest index of a subtour in $T^*_0, \ldots, T^*_k$ that it intersects (or $\infty$ if it intersects none). That is,
\begin{align*}
  \low(T)=
\min\{i: V(T^*_i) \cap V(T) \neq \emptyset\}
      .
\end{align*}
Moreover, an important quantity will be $\lb(T^*_{\low(T)})$, with the convention that $\lb(T^*_\infty)= 0$.

\begin{remark}
  \label{rem:init}
  The main difference in the polynomial-time algorithm is the initialization,
  as we do not know how to find an initialization satisfying \ref{en:i1}-\ref{en:i3}  in polynomial time. 
  Indeed, it is consistent with our knowledge that $2\alpha$ (even $2$)  is an upper bound on the
  integrality gap and, in that case,  such an algorithm would  always find a
  tour for singleton instances with  $B=\emptyset$. \end{remark}
\begin{remark}
  For intuition, let us mention that the reason to maximize the lexicographic order  (subject to \ref{en:i1}-\ref{en:i2}) is that we
  will use the following properties to bound the weight of the final
  tour:
  \begin{enumerate}
    \item Let $T$ be a subtour with $\low(T) = i > 0$ and $w(T)\leq
      2\alpha \lb(T)$. Then  $\lb(T) \leq
      \lb(T^*_i)$.
    \item For any disjoint subtours $T_1, T_2, \ldots, T_\ell$ with $\low(T_j) = i > 0$ and $w(T_j) \leq \alpha \lb(T_j)$ for $j=1, \ldots, \ell$, we have
  \begin{align*}
	\sum_{j=1}^\ell \lb(T_j) \le 2 \lb(T^*_i).
  \end{align*}
  \end{enumerate}
  These claims will be used to bound the weight of the subtours added in the merge procedure (see below). Their proofs are easy and can be found in the analysis (see the proofs of Claim~\ref{claim:boundlex1} and Claim~\ref{claim:boundlex2}).
\end{remark}

\paragraph{Merge procedure}
After the initialization, $T^*$ contains a collection of subtours that do not necessarily form a tour. The goal of the ``merge procedure'' is to form a tour of the entire graph, connecting these subtours by adding additional (cheap) subtours. We will do so while maintaining the invariant that $T^*_0 = B$ is a disjoint subtour  in $T^*$ until the very last step when a tour is formed. We emphasize that even though $T^*$ changes throughout the procedure, the index $\low(T)$ will always be defined with respect to the original subtours $T^*_0,\ldots, T^*_k$.

Specifically, the procedure repeats the following until $T^*$ is a tour (that visits all the vertices). 
Let $T_0^*, T_1, \ldots, T_\ell$ be the collection of subtours in $T^*$ (recall that $T_0^*=B$). As they are disjoint, $T_1, T_2, \ldots, T_\ell$ naturally partition the vertex set $V\setminus V(B)$ into $V(T_1), V(T_2), \ldots, V(T_\ell)$ plus  singleton sets for the remaining vertices in $V \setminus \left(V(B) \cup V(T_1) \cup \ldots \cup V(T_\ell)\right)$. Let $\con$ be this partitioning of $V \setminus V(B)$. By construction, each nonsingleton set in $\con$ corresponds to a subtour and thus, for each $V' \in \con$, the subgraph induced by $V'$ is strongly connected.   We can therefore use $\cA$ to find a collection $F$ of subtours 
such that   
\begin{enumerate}[(i)]
  \item $|\delta^+_{F}(V')|  \geq 1$ for all $ V' \in \con$, \item  $w_{\cI}(T)\le \alpha \lb(T)$ for every subtour $T$ in $F$ disjoint from $B$, and \label{cond:light}
\item $w(F_B)\leq \wD$, where $F_B\subseteq F$ is the collection of subtours in $F$ that intersect $B$. \label{cond:beta}
\end{enumerate}

Note that $\cA$ is guaranteed to find such a collection $F$ since it is assumed to be an $(\alpha,\wD)$-light algorithm for \EPC{} on $(\cI, B)$. 
Furthermore, we
may assume that a subtour $T$ in $F$ does not only visit a subset $V(T)$ of the vertices $V(T')$ visited by a subtour $T'$ in $T^*$.  Indeed, such a subtour can
safely be removed from $F$, yielding a new (smaller) collection of subtours that satisfies the
above conditions.
Having selected $F$, we now proceed to explain the ``update phase''.
\begin{enumerate}[label=U\arabic*:, ref=U\arabic* ]\item \label{en:u1}  Let $X = \emptyset$.  
	\item \label{en:u2} Select a subtour $T$ in $T^* \cup F \cup X$
          that \emph{maximizes}	$\low(T)$. Let $j=\low(T)$.
  \item \label{en:u3}If  there exists a cycle $C$ 
    of weight $w(C) \leq \alpha \lb(T^*_j)$
	that connects $T$ to other vertices, i.e., $V(T) \cap V(C) \neq \emptyset$ and $V(C) \not \subseteq V(T)$,   then add $C$ to $X$ and
	repeat from Step~\ref{en:u2}.
\item \label{en:u4}Otherwise, update $T^*$ by adding the ``new'' edges in $T$, i.e.,  $T^* \leftarrow T^* \cup (T \cap F) \cup (T \cap X)$.

\end{enumerate}

Some comments about the update of $T^*$ are in order. We emphasize that we do \emph{not}
add all edges of $F \cup X$ to $T^*$. Instead, we only add those new edges that belong to  
 the component $T$ selected in the final iteration of the update phase. Among other things, this ensures the invariant that $B$ is a subtour in $T^*$ until the very end. Indeed, the iteration where we add a subtour intersecting $B = T^*_0$ must be the last iteration. This is because, in that case, the selected $T$ that maximizes $\low(T)$ must satisfy $\low(T)=0$, which in turn implies that $T^* \cup F \cup X$ is a single tour $T$ that visits all vertices.
 
 Finally, let us remark that the update maintains that $T^*$  is a collection of subtours (i.e., an Eulerian multiset of edges).  As $T$ is a subtour in $T^* \cup F \cup X$, and $F$ and $X$ themselves are collections of subtours, we have that $T^*$
remains a collection of subtours  after the update. This finishes the description of the merging procedure and the algorithm (see also the example below). 

\begin{figure}[h]
\centering
\newcommand{\base}{
\draw[fill=gray!20] plot [smooth cycle,tension=0.7] coordinates { (-1,-1) 
  (1,-1) (1.4,1) (-1,1)};
  \draw[gray] (-1,-1) edge[-, bend right = 10 ] (0,0);
  \draw[gray] (1,-1) edge[-, bend left = 20] (0,0);
  \draw[gray] (1.4,1) edge[-, bend right = 15 ] (0,0);
  \draw[gray] (-1,1) edge[-, bend left = 10] (0,0);
\node at (0.1, 0.8) {$T^*_{10}$};
  \node at (-0.8, 0.0) {$T^*_9$};
  \node at (0.1, -0.8) {$T^*_7$};
  \node at (0.9, 0.0) {$T^*_6$};

  \begin{scope}[xshift=4.5cm]
	\draw[fill=gray!20] plot [smooth cycle,tension=0.7] coordinates { (-0.7,-0.7) 
  		(0.5,-0.7) (1, 0) (0.8,0.7) (-0.8,0.8) (-1.2, 0)};
	\node at (0.0, 0.0) {$T^*_3$};
  \end{scope}

  \begin{scope}[xshift=8cm]
	\draw[fill=gray!20] plot [smooth cycle,tension=0.7] coordinates { (-0.7,-0.7) 
  		(0.5,-0.7) (1, 0) (0.8,0.7) (-0.8,0.8) (-0.6, 0)};
  	\draw[gray] (-0.6,0) edge[-, bend left =20] (1,0);
	\node at (0.2, -0.35) {$T^*_8$};
	\node at (0.2, 0.5) {$T^*_5$};
  \end{scope}

  \begin{scope}[yshift = -3.5cm]
  \begin{scope}[xshift=0cm]
	\draw[fill=gray!20] plot [smooth cycle,tension=0.7] coordinates { (-0.7,-0.7) 
  		(0.5,-0.7) (1, 0) (0.8,0.7) (0, 1) (-0.8,0.8) (-1.2, 0)};
	\node at (0.0, 0.1) {$T^*_4$};
  \end{scope}
  \begin{scope}[xshift=4.5cm]
	\draw[fill=gray!20] plot [smooth cycle,tension=0.7] coordinates { (-0.7,-0.7) 
  		(0.5,-0.7)  (0.8,0.7) (0, 1.2) (-0.8,0.8) (-1.2, 0)};
	\node at (-0.1, 0.1) {$T^*_2$};
  \end{scope}
  \begin{scope}[xshift=8cm, yshift=0.3cm]
	\draw[fill=gray!20] plot [smooth cycle,tension=0.7] coordinates { (-1,-1) 
  		(0.5,-1)  (0.8,0.7) (0, 1.2) (-1,1) (-1.4, 0) };
	\node at (-0.2, -0.0) {$T^*_1$};
  \end{scope}
  \end{scope}
  \begin{scope}[xshift=11.5cm, yshift=-1.3cm]
	\draw[fill=gray!20] plot [smooth cycle,tension=0.7] coordinates { (-1,-0.6) 
  		(0.5,-0.5)  (0.8,0.4) (0, 0.7) (-1,0.6) (-1.4, 0) };
	\node at (-0.2, -0.0) {$T^*_0 = B$};
  \end{scope}
}
\begin{tikzpicture}
\base
\draw[thick, blue] plot [smooth cycle, tension = 1] coordinates {(-0.4,-0.5) (-0.4, -2.5) (0.7, -2.6) (1, -0.5) };
\draw[ultra thick, blue] plot [smooth cycle, tension = 1] coordinates {(5.6,0.0) (5.6, 0.7) (7.5, 0.7) (7.5, -0.2) };
\draw[thick, blue] plot [smooth cycle, tension = 1] coordinates {(5.3, -3.0) (5.5,-4.2) (7.5, -3.9) (10.1, -3.2) (10.1, -1.0)  (7.5, -2.5)};
\draw[thick, red, dashed] plot [smooth cycle, tension = 1] coordinates {(4.2, -2.5) (3.5,-2.5)  (0.2, -0.2) (1.5, 0.1) };
\end{tikzpicture}
 \caption{An illustration of the merge procedure. The gray areas depict the subtours in $T^*$. Blue  (solid) cycles depict $F$
and the red (dashed) cycle depicts $X$ after one iteration of the update
phase. The thick cycle represents the edges that this merge procedure would add to
$T^*$.}
\label{fig:algo}
\end{figure}

\begin{example}
  In Figure~\ref{fig:algo}, we have that, at the start of a merging step, $T^*$ consists of $7$ subtours containing
$\{T^*_6, T^*_7, T^*_9, T^*_{10}\}, \{T^*_3\}, \{T^*_5, T^*_8\}, \{T^*_4\}, \{T^*_2\},$ and  $\{T^*_1\}$.
The blue (solid) cycles depict the subtours  of $F$. First, we set $X =\emptyset$ and the algorithm selects the
subtour $T$ in $T^* \cup F \cup X$ that maximizes $\low(T)$. In
this example, it would be the leftmost of the three subtours in $T^* \cup
F$, with $\low(T) = 4$. The algorithm now tries to connect this component
to another component by adding a cycle with weight at most
$\alpha \lb(T^*_4)$. The red (dashed) cycle corresponds to such a cycle and its edge set  is
added to $X$. In the next iteration, the algorithm considers the two
subtours in $T^* \cup F \cup X$.  The one that maximizes $\low(T)$ contains $T^*_3,
T^*_5$, and $T^*_8$. Now suppose that there is no cycle of weight at
most $\alpha \lb(T^*_3)$ that connects this component to another component.
Then the set $T^*$ is updated by adding those subtours (edges) of $F\cup X$ that
belong to this component (depicted by the thick cycle).
\end{example}

\subsubsection{Analysis}

\paragraph{Termination} We show that the algorithm terminates by arguing that
the update phase decreases the number of connected components and the merge procedure is therefore 
 repeated at most $k\leq n$ times. 

\begin{lemma}
  \label{lem:termination}
The update phase terminates in polynomial time and decreases the number of connected components in $(V,T^*)$.
 \end{lemma}
 \begin{proof}
   First, observe that each single step of the update phase can be implemented
   in polynomial time. The only nontrivial part is Step~\ref{en:u3}, which can
   be implemented as follows: for each edge $(u,v)\in
   \delta^+(V(T))$ consider the cycle consisting of $(u,v)$ and 
   a shortest path from $v$ to $u$. Moreover, the entire update phase terminates
   in polynomial time, because  each time the if-condition of Step~\ref{en:u3}
   is satisfied, we add a cycle to $X$ that decreases the number of connected
   components in $(V,T^* \cup F \cup X)$.  The if-condition of
   Step~\ref{en:u3} can therefore be satisfied at most $ k \leq n$ times. 
 
We proceed by proving that, at termination, the update phase decreases the number
of connected components  in $(V,T^*)$. Recall that, once the algorithm
   reaches Step~\ref{en:u4}, it has selected a subtour $T$ in $T^* \cup F \cup X$. We claim that $T$ visits vertices in at least two components of $(V, T^*)$. This is because the subtours in $F$ satisfy $|\delta^+_F(V')| \geq 1$ for all $V'\in \con$, where $\con$ denotes the connected components of $(V\setminus V(B), T^* \setminus B)$. Moreover, as already noted, $T$ intersects $B$ only in the last iteration, when $T$ forms a tour.  Therefore, when
   the algorithm updates $T^*$ by adding the edges $(F \cup X) \cap
   T$, it decreases the number of components in $(V, T^*)$ by at least
   one. 
 \end{proof}

\newcommand{\Fatstep}[1]{\tilde{F}_{#1}}
\newcommand{\Xatstep}[1]{\tilde{X}_{#1}}
\newcommand{\Fatstepxing}[2]{\tilde{F}_{#1}^{#2}}
 \paragraph{Performance Guarantee}
We split our analysis of the performance guarantee into two
 parts. Note that when one execution of the merge procedure terminates (Step~\ref{en:u4}), we add the edge set
 $(F\cap T) \cup (X \cap T)$ to our solution. We will analyze the
 contribution of these two sets ($F \cap T$ and $X \cap T$) separately.
 More formally, suppose that the algorithm performs $R$ repetitions of
 the merge procedure. Let $ T_1, T_2, \ldots,  T_R$, $F_1, F_2, \ldots, F_R$, and
 $X_1, X_2, \ldots, X_R$ denote the selected subtour $T$, the edge set
 $F$, and the edge set $X$, respectively, at the end of each repetition.
  To simplify notation, we denote the edges added to $T^*$ in the $r$-th
 repetition by $\Fatstep{r} = F_r \cap T_r$ and $\Xatstep{r} =
 X_r \cap T_r$.
 
 With this notation, we proceed to bound the total weight of the solution by
 \begin{align*}
   \underbrace{w\left(\bigcup_{r=1}^R  \Fatstep{r}\right)}_{\leq 2\alpha\lb(\bar B)+\wD\mbox{\scriptsize\  by Lemma~\ref{lem:bound2}}} 
   + \underbrace{w\left(\bigcup_{r=1}^R  \Xatstep{r}\right)}_{\leq \alpha \lb(\bar B)\mbox{\scriptsize\ by Lemma~\ref{lem:bound1}}} 
   + \ w(B)
+{\sum_{i=1}^k w(T^*_i)}\\
   \leq 5\alpha\lb(\bar B)+\wD+w(B), \end{align*}
as claimed in Theorem~\ref{thm:LoctoGlo}.
Here we used that $\sum_{i=1}^k w(T^*_i) \leq 2\alpha \lb(\bar B)$, which is implied by the selection of $T_1^*, T_2^*, \ldots, T^*_k$ (\ref{en:i1}-\ref{en:i2}). It remains to prove Lemmas~\ref{lem:bound1} and~\ref{lem:bound2}.

 \begin{lemma} 
   \label{lem:bound1}We have 
   $w\left(\bigcup_{r=1}^R \Xatstep{r} \right)  \leq \alpha \lb(\bar B)$.
\end{lemma}

\begin{proof}
  Note that $\Xatstep{r}$ consists of a subset of the cycles added to $X_r$
  in Step~\ref{en:u3} of the update phase: specifically, of those cycles
   that were contained  in the
  subtour $T_r$ selected in Step~\ref{en:u2} in the last iteration
  of the update phase during the  $r$-th repetition of the merge procedure. We can therefore decompose
  $\bigcup_{r=1}^R  \Xatstep{r}$ into cycles $C_1, C_2, \ldots, C_c$,  indexed in the order they were added by the
  algorithm.  We assume that all these cycles have strictly positive weight, as $0$-weight cycles do not affect $w\left(\bigcup_{r=1}^R \Xatstep{r} \right)$.
  At the time a cycle $C_i$  (with $w(C_i) > 0$) was selected in Step~\ref{en:u3} of the update phase, it satisfied the following two properties:
  \begin{itemize}
	\item[(i)] it connected the subtour $T$
          with $\low(T)=j>0$
          selected in Step~\ref{en:u2} with at least one other
        subtour $T'$ such that $\low(T') < \low(T)$; and
	\item[(ii)] it had weight $w(C_i)\le \alpha\lb(T^*_{j})$. 
  \end{itemize}
  In this case, we say that $C_i$ is marked by $j$. Note that $1\leq j \leq k$, since $\alpha \lb(T^*_j) \ge w(C_i) >0$ and by convention $\lb(T^*_\infty) = 0$.
  
  We claim that at most one cycle in $C_1, C_2, \ldots, C_c$  is
 marked by each of the numbers $\{1,2,\ldots,k\}$. 
To see this, consider the
  first cycle $C_i$ marked by $j$ (if any). By {\em (i)} above,
  when $C_i$ was added, it connected two subtours $T$ and
  $T'$ such that $\low(T')< \low(T)=j$.   As the algorithm
  only adds edges, $T$ and $T'$ will remain connected throughout the
  execution of the algorithm. Therefore, by the definition of $\low$ and by
  the fact that $\low(T')<\low(T)$, we have that a subtour $T''$ selected in Step~\ref{en:u2} later in the algorithm  always
  has $\low(T'') \neq j$. Hence, no other cycle will be marked
  by $j$.

 The bound now follows since at most one  cycle with positive weight is marked by $j$, and such a cycle has weight at most $\alpha
 \lb(T^*_{j})$.
 Moreover, we have $\sum_{j=1}^k \alpha \lb(T^*_j) \leq \alpha \lb(\bar B)$, which is again implied by the selection of $T_1^*, T_2^*, \ldots, T^*_k$ (\ref{en:i1}-\ref{en:i2}).
\end{proof}
We complete the analysis of the performance guarantee with the following lemma. We remark that this is the only part of the proof that relies on the initial subtours $T_1^*, \ldots, T_k^*$ maximizing the lexicographic order (\ref{en:i3}). 
\begin{lemma}
  \label{lem:bound2}
  We have
  $w\left(\bigcup_{r=1}^R \Fatstep{r} \right)  \leq  2 \alpha\lb(\bar B)
  +\wD$. 
\end{lemma}

\begin{proof}
  Consider the $r$-th repetition of the merge procedure. Partition the collection of subtours $\Fatstep{r}$ into 
\begin{align*}
    \Fatstepxing{r}{i} = \{T\mbox{ subtour in $\Fatstep{r}$}: \low(T) = i\}\quad \mbox{for $i\in \{0,1 \ldots, k, \infty\}$.}
  \end{align*}
  That is, $\Fatstepxing{r}{i}$  contains those  subtours in
  $\Fatstep{r}$ that intersect $T^*_i$  and do not intersect any of the subtours
  $T^*_0 = B, T^*_1, T^*_2, \ldots, T^*_{i-1}$ (or intersect none if $i=\infty$).
  The total weight $w(\Fatstep{r})$ of $\Fatstep{r}$ thus equals
  \begin{align*}
    w(\Fatstepxing{r}{0}) + w(\Fatstepxing{r}{\infty}) + \sum_{i=1}^k w(\Fatstepxing{r}{i})\,.
  \end{align*}
   We bound the weight of $\Fatstep{r}$
   by considering these terms separately. Let us start with $\Fatstepxing{r}{0}$.

\begin{claim}\label{cl:bound0}
  The set $\Fatstepxing{r}{0}$ can be
  non-empty only for $r=R$, and $w(\Fatstepxing{R}{0})\le \wD$.
\end{claim}
\begin{proof}
  The first claim follows by the invariant that $B$ is a subtour in $T^*$ until the very last iteration of the merge procedure. Indeed,  if $\Fatstepxing{r}{0}\neq\emptyset$,
then the algorithm must terminate after  the
$r$-th merge procedure: the subtour $T$ selected in
Step~\ref{en:u2} must visit all vertices. For the second part, we have that every
$T\in \Fatstepxing{R}{0}$ must intersect $T^*_0 = B$. Therefore,
property~\ref{cond:beta} of the edge set $\Fatstep{R}$ returned by $\cA$ asserts that
$w(\Fatstepxing{R}{0})\le \wD$.
\end{proof}

For $i>0$, we start by two simple claims that follow since each subtour $T$ in $\Fatstep{r}$
  satisfies $w(T) \leq \alpha \lb(T)$ (by property~\ref{cond:light} of $\cA$) and the choice of $T^*_1, \ldots, T^*_k$ to maximize the
  lexicographic order~\ref{en:i3} subject to~\ref{en:i1}-\ref{en:i2}. 
  \begin{claim}
	\label{claim:boundlex1}
  For $i>0$ and $T\in \Fatstepxing{r}{i}$ we have $\lb(T) \leq \lb(T_i^*)$.
  \end{claim}
  \begin{proof}
	The inequality $\lb(T) > \lb(T_i^*)$ together with the fact that $w(T) \leq \alpha \lb(T) \leq 2\alpha \lb(T)$ would contradict that $T^*_1, \ldots, T^*_k$
	was chosen to maximize the lexicographic order~\ref{en:i3}. Indeed, in that case, a
  an initialization satisfying~\ref{en:i1}-\ref{en:i2} of higher lexicographic order
	would  be $T^*_1, \ldots, T^*_{i-1}, T$. 	\end{proof}

Together with $\lb(T^*_\infty)  = 0$, this claim implies that $w(\Fatstepxing{r}{\infty})=0$. 
We now present a more general claim that also applies to $\Fatstepxing{r}{i}$ with $1\leq i \leq k$. 

  \begin{claim}
	\label{claim:boundlex2}
  If $i>0$, then we have $\lb(\Fatstepxing{r}{i}) \leq  2\lb(T^*_i)$.
  \end{claim}
  \begin{proof}
    Suppose towards a contradiction that $\lb(\Fatstepxing{r}{i}) >  2\lb(T^*_i)$.  Let
    $T_1, T_2, \ldots, T_\ell$ be the subtours in $\Fatstepxing{r}{i}$ and define $T$ to be the subtour obtained by taking the union of the subtours $T_i^*$ and $T_1, \ldots,
    T_\ell$.  Consider the initialization $T^*_1, \ldots,T^*_{i-1}, T$. By construction these subtours are disjoint and disjoint from $B$ since $i>0$.  Thus~\ref{en:i1} is satisfied. Moreover, 
     we have $\lb(T) > \lb(T_i^*)$
    and therefore this initialization is lexicographically larger
    than $T^*_1, \ldots, T^*_k$.  This is
    a contradiction if $T$ also satisfies~\ref{en:i2}, i.e.,
    if ${w(T)}\leq 2\alpha \lb(T)$. 
    
    Therefore, we must have $w(T) > 2\alpha \lb(T)$.  By the facts
    that $w(T_j) \leq \alpha \lb(T_j)$ (by property~\ref{cond:light} of $\cA$)
    and that $w(T^*_i) \leq 2\alpha \lb(T^*_i)$ (by~\ref{en:i2}),
    $$ w(T)= w(T^*_i) + \sum_{j=1}^\ell w(T_j) \leq 2\alpha
    \lb(T^*_i) + \sum_{j=1}^\ell \alpha \lb(T_j)\quad\mbox{ and }\quad
    \lb(T)\geq \sum_{j=1}^\ell \lb(T_j).  $$ 
    Together with $w(T)
    >  2\alpha \lb(T)$, we see that
    $$  2\alpha \sum_{j=1}^\ell \lb(T_j)< 2\alpha
    \lb(T^*_i) +\alpha \sum_{j=1}^\ell \lb(T_j).$$
    From here, we can conclude that  $\lb(\Fatstepxing{r}{i}) = \sum_{j=1}^\ell \lb(T_j)
    \leq 2 \lb(T^*_i)$.
  \end{proof}
  
  Using the above claims, we can write $w\left(\bigcup_{r=1}^R  \Fatstep{r} \right)$ as
  \begin{align*}
    \sum_{r=1}^R\left( w(\Fatstepxing{r}{0}) + w(\Fatstepxing{r}{\infty}) + \sum_{i=1}^k w(\Fatstepxing{r}{i}) \right) & \leq  \beta + \sum_{r=1}^R\sum_{i=1}^k w(\Fatstepxing{r}{i})\\
    & \leq \beta + \alpha \sum_{r=1}^R \sum_{i=1}^k  \lb(\Fatstepxing{r}{i}) \\
    & = \beta + \alpha \sum_{i=1}^k \sum_{r: \ \Fatstepxing{r}{i} \neq \emptyset} \lb(\Fatstepxing{r}{i}) \\
    &\leq \beta +  2\alpha \sum_{i=1}^k \sum_{r: \ \Fatstepxing{r}{i} \neq \emptyset} \lb(T^*_i). 
  \end{align*}
  We complete the proof of the lemma by using Claim~\ref{claim:boundlex1} to prove that
  $\Fatstepxing{r}{i}$ is non-empty for at most one repetition $r$ of the merge procedure.
  Suppose towards a contradiction that  there exist $1\leq r_0< r_1 \leq R$ such
  that both $\Fatstepxing{r_0}{i} \neq \emptyset$ and $\Fatstepxing{r_1}{i} \neq \emptyset$. In
  the $r_0$-th repetition of the merge procedure, $T^*_i$ was contained in the
  subtour $T_{r_0}$ (selected in Step~\ref{en:u2}) since otherwise no edges incident to $T^*_i$ would
  have been added to $T^*$. Therefore $j=\low(T_{r_0}) \leq i$.
  Now consider a subtour $T$ in $\Fatstepxing{r_1}{i}$.
  First, recall that we have assumed that $T$, being a subtour in $F_{r_1}$, does not only visit
  a subset $V(T)$ of vertices $V(T')$ visited by a subtour $T'$ in $T^*$.
  In particular, since $T_{r_0}$ is a subset of some subtour $T'$ in $T^*$ during the $r_1$-th repetition,
  we have $V(T) \not \subseteq V(T_{r_0})$.
  Second, by Claim~\ref{claim:boundlex1},
  we have  $w(T) \leq \alpha \lb(T) \leq \alpha \lb(T_i^*)$.

  In short, $T$ is a subtour that connects $T_{r_0}$
  to another component and it has weight at most $\alpha \lb(T_i^*)\le
  \alpha \lb(T_j^*)$, where $j=\low(T_{r_0}) \leq i$. As $T$ is Eulerian, it can be decomposed into cycles. One
  of these cycles, say $C$,  connects $T_{r_0}$ to another component and 
  \begin{align}
	w(C) \leq w(T) \leq \alpha \lb(T^*_j).
	\label{eq:cycleweight}
  \end{align}
  In other words, there exists a cycle $C$ that, in the $r_0$-th repetition of
  the merge procedure, satisfied the if-condition of Step~\ref{en:u3}, which
  contradicts the fact that $C$ was not added during the $r_0$-th repetition.
\end{proof}

\subsection{Polynomial-Time Algorithm}
\label{sec:polyalg}
In this section we describe how to modify the arguments in
Section~\ref{sec:ExistDescAlg} to obtain an algorithm that runs in time
polynomial in the number $n$ of vertices, in $1/\varepsilon$, and in the running time of $\cA$.

By Lemma~\ref{lem:termination}, the update phase can be implemented in
time polynomial in $n$. Therefore, the merge procedure described in
Section~\ref{sec:ExistDescAlg} runs in time polynomial in $n$ and in the
running time of $\cA$.  The problem is the initialization: as mentioned in
Remark~\ref{rem:init}, it seems difficult to give a polynomial-time algorithm for
finding subtours  $T_1^*, \ldots, T_k^*$  that satisfy~\ref{en:i1} and \ref{en:i2} together with the third condition~\ref{en:i3}  that we should maximize the lexicographic
order of
\begin{align*}
  \langle \lb(T^*_1), \lb(T^*_2), \ldots, \lb(T^*_k) \rangle.
\end{align*}

We overcome this obstacle by first identifying the properties that we actually
use from selecting the subtours as above. We then show that we can
obtain an initialization that satisfies these properties in polynomial
time. Our initialization will still satisfy~\ref{en:i1} and a relaxed variant of~\ref{en:i2} that we now describe. 
To simplify notation,  define
\begin{align*}
  \lbs(T) = \lb(T) + \varepsilon \cdot \frac{|V(T)|}{n} \cdot \lb(\bar B)
\end{align*}
for a subtour $T$. Note that $\lbs(T)$ is a slightly increased version of $\lb(T)$. This increase is used to lower-bound the progress in  Lemma~\ref{lem:polytime}. Also note that $\lbs$, like $\lb$, is additive over vertex-disjoint subtours. Finally, we reprise the convention that $\lbs(T^*_\infty) = \lb(T^*_\infty) = 0$. 

Our initializations will be collections $T_1^*, \ldots, T_k^*$ of subtours satisfying
\begin{enumerate}[label=I\arabic*':, ref=I\arabic*' ]
  \item[I1:] $B, T^*_1, T^*_2, \ldots, T^*_k$ are disjoint subtours; 
    \setcounter{enumi}{1} 
  \item $w(T^*_i) \leq 3\alpha \lbs(T^*_i)$ for $i=1,2, \ldots, k$. \label{en:ii2}
\end{enumerate}
While we do not maximize the lexicographic order, we assume that the subtours are indexed so that $\lbs(T_1^*) \geq \lbs(T_2^*) \geq \ldots \geq \lbs(T_k^*)$.
Note that the difference between~\ref{en:ii2} and~\ref{en:i2} is that here we require $w(T^*_i) \leq 3\alpha \lbs(T^*_i)$ instead of $w(T^*_i) \leq 2 \alpha \lb(T^*_i)$. The reason why
we use a factor of $3$ instead of $2$ is that it leads to a better constant when balancing the parameters and, as previously mentioned, we use $\lbs$ instead of $\lb$ to lower-bound the progress in Lemma~\ref{lem:polytime}.

The main change to our initialization to achieve polynomial running time is that we do \emph{not} maximize the lexicographic order~\eqref{en:i3}. As mentioned in the analysis in Section~\ref{sec:ExistDescAlg}, the only way
we use that the initialization  maximizes the lexicographic order is for the proof of Lemma~\ref{lem:bound2}. In particular,
this is used in the proofs of Claims~\ref{claim:boundlex1}
and~\ref{claim:boundlex2}. 
Instead of maximizing the lexicographic order, our polynomial-time algorithm will ensure a relaxed variant of those claims (formalized in the lemma below: see Condition~\eqref{eq:relax}). The
claimed polynomial-time algorithm is then obtained by first proving that
a slight modification of the merge procedure returns a tour of value at most
$9(1+\varepsilon)\alpha\lb(\bar B) + \beta + w(B) $ if Condition~\eqref{eq:relax} holds, and then showing that an initialization
satisfying this condition (and~\ref{en:i1},\ref{en:ii2}) can be found in time polynomial in
$n$, \bl{$1/\varepsilon$,} and in the running time of $\cA$. We start by describing the 
modification to the merge procedure. 

\paragraph{Modified merge procedure} The only modification to the merge
procedure in Section~\ref{sec:ExistDescAlg} is that we change the update phase
by  relaxing the condition of the if-statement in Step~\ref{en:u3} from $w(C)
\leq  \alpha \lb(T^*_j)$ to $w(C) \leq 3\alpha \lbs(T^*_j)$, where $j = \low(T)$ and $T$ is the subtour selected in Step~\ref{en:u2}. In other words, Step~\ref{en:u3} is replaced by:
\begin{enumerate}[label=U\arabic*':, ref=U\arabic*']
	 \setcounter{enumi}{2}

  \item \label{en:uu3}If  there exists a cycle $C$ 
    of weight $w(C) \leq 3\alpha \lbs(T^*_j)$
	that connects $T$ to other vertices, i.e., $V(T) \cap V(C) \neq \emptyset$ and $V(C) \not \subseteq V(T)$,   then add $C$ to $X$ and
	repeat from Step~\ref{en:u2}.
\end{enumerate}
Clearly the modified merge procedure still runs in time polynomial in $n$ and
in the running time of $\cA$. Moreover, we show that if Condition~\eqref{eq:relax} holds then the returned tour will have the desired weight. Recall from
Section~\ref{sec:ExistDescAlg}  that $\Fatstep{r}$ denotes the subset of $F$ and $\Xatstep{r}$
  denotes the subset of $X$ that were added in the $r$-th repetition of the
  (modified) merge procedure.  Furthermore, we define (as in the previous section)
  $\Fatstepxing{r}{i} = \{T\mbox{ subtour in $\Fatstep{r}$}: \low(T) = i\}$.
\begin{lemma}
  \label{lem:polybound}
  Suppose that the algorithm is initialized with subtours $T^*_1, T^*_2, \ldots, T^*_k$ satisfying~\ref{en:i1} and~\ref{en:ii2}. If  in each repetition $r$ of the
  modified merge procedure we add a subset $\Fatstep{r}$  such that 
  \begin{align}
    \lb(\Fatstepxing{r}{i}) \leq 3\lbs(T_i^*) \qquad \mbox{for all }i\in \{1,2, \ldots, k, \infty\} \,,
	\label{eq:relax}
  \end{align}
  then the returned tour has weight at most $9(1+\varepsilon)\alpha\lb(\bar B) + \beta + w(B)$.
\end{lemma}
Let us comment on the above statement before giving its proof. The bound~\eqref{eq:relax}  is a relaxation of the bound
of Claim~\ref{claim:boundlex2} from $\lb(\Fatstepxing{r}{i}) \leq 2 \lb(T_i^*)$ to
$\lb(\Fatstepxing{r}{i}) \leq 3 \lbs(T_i^*)$; and it also implies a relaxed version of
Claim~\ref{claim:boundlex1}: from $\lb(T) \leq \lb(T_i^*)$ to $\lb(T) \leq 3\lbs(T_i^*)$
(for every $T$ in $\Fatstepxing{r}{i}$).
It is because of this relaxed bound that we modified the if-condition of the
update phase (by relaxing it by the same amount); this will be apparent in
the proof.
\begin{proof}

  As in the analysis of the performance guarantee in
  Section~\ref{sec:ExistDescAlg}, we can write the weight of the returned tour
  as 
  \begin{align*}
w\left(\bigcup_{r=1}^R \Fatstep{r}\right) + 
    w\left(\bigcup_{r=1}^R \Xatstep{r}\right) + 
    w(B) + 
    \sum_{i=1}^k w(T^*_i).  
  \end{align*} 
 
  To bound $w\left(\bigcup_{r=1}^R \Xatstep{r}\right)$, we  observe that the proof of
  Lemma~\ref{lem:bound1} generalizes verbatim except that the weight of a cycle
  marked by $i$ is now bounded by $3\alpha\lbs(T_i^*)$ instead of by
  $\alpha \lb(T_i^*)$ (because of the relaxation of the bound in the
  if-condition in Step~\ref{en:uu3}). Hence \[w\left(\bigcup_{r=1}^R \Xatstep{r}\right) \leq \sum_{i=1}^k 3\alpha\lbs(T_i^*)\,. \]

  We proceed to bound $w\left(\bigcup_{r=1}^R \Fatstep{r}\right)$. Using the
  same arguments as in the proof of Lemma~\ref{lem:bound2}, we get
  \begin{align*}
    w\left(\bigcup_{r=1}^R \Fatstep{r}\right)   = \sum_{r=1}^R \left( w(\Fatstepxing{r}{0}) + w(\Fatstepxing{r}{\infty}) + \sum_{i=1}^k w(\Fatstepxing{r}{i}) \right) \leq \beta +  \sum_{r=1}^R  \sum_{i=1}^k w(\Fatstepxing{r}{i}) \\
    \leq \beta +   \alpha \sum_{i=1}^k \sum_{r: \, \Fatstepxing{r}{i} \neq \emptyset} \lb(\Fatstepxing{r}{i}) 
    \leq \beta + \alpha \sum_{i=1}^k \sum_{r: \, \Fatstepxing{r}{i} \neq \emptyset} 3\lbs(T_i^*)
  \end{align*}
  where, for the first inequality, we used that Claim~\ref{cl:bound0}: $\sum_{r=1}^R w(\Fatstepxing{r}{0}) \leq \beta$ generalizes verbatim from the non-constructive analysis and we have $w(\Fatstepxing{r}{\infty}) \leq \alpha\lb(\Fatstepxing{r}{\infty}) \leq 3\alpha \lbs(T^*_\infty) =0$ by the assumption of the lemma; similarly,
   the last inequality is by the assumption of the lemma. 
   
   Now using that the subtours are indexed so that $\lbs(T^*_1) \geq \lbs(T^*_2) \geq \ldots \geq \lbs(T^*_k)$  we
  apply exactly the same arguments as in the end of the proof of
  Lemma~\ref{lem:bound2} to prove that $\Fatstepxing{r}{i}$ is non-empty for at most one
  repetition $r$ of the merge procedure. The only difference is 
  that~\eqref{eq:cycleweight} becomes 
  \begin{align*}
    w(C) \leq w(T) \leq  3\alpha\lbs(T_j^*)  \end{align*}
  (because~\eqref{eq:relax} can be seen as a relaxed version of
  Claim~\ref{claim:boundlex1}). However, as we also updated the bound in the
  if-condition, the argument that $C$ would satisfy the if-condition of
  Step~\ref{en:uu3} is still valid.  Hence, we conclude that $\Fatstepxing{r}{i}$ is non-empty in at most one repetition and therefore
  \begin{align*}
     w\left(\bigcup_{r=1}^R \Fatstep{r}\right) 
     \leq \beta + \alpha   \sum_{i=1}^k \sum_{r: \, \Fatstepxing{r}{i} \neq \emptyset} 3\lbs(T_i^*) \leq \beta + \sum_{i=1}^k 3\alpha \lbs(T_i^*)\,.
\end{align*}

  By the above bounds and  since the initialization $T_1^*, T_2^*, \ldots, T_k^*$ satisfies~\ref{en:i1} and~\ref{en:ii2}, the weight of the returned
  tour is
  \begin{align*}
w\left(\bigcup_{r=1}^R \Fatstep{r}\right) + 
    w&\left(\bigcup_{r=1}^R \Xatstep{r}\right) + 
    w(B) + 
    \sum_{i=1}^k w(T^*_i) \\
    &\leq \beta + \sum_{i=1}^k 3\alpha \lbs(T_i^*) + \sum_{i=1}^k 3\alpha \lbs(T_i^*) + w(B) + \sum_{i=1}^k w(T^*_i)\\
    &\leq 9\alpha \sum_{i=1}^k \lbs(T_i^*) + \beta + w(B) \\
    &\leq 9(1+\varepsilon)\alpha\lb(\bar B) + \beta + w(B).
  \end{align*} 
\end{proof}

\paragraph{Finding a good initialization in polynomial time}
By the above Lemma~\ref{lem:polybound}, it is sufficient to find an initialization such that~\ref{en:i1},~\ref{en:ii2} are satisfied and Condition~\eqref{eq:relax} holds during the execution of the modified
merge procedure. However, how can we do it in polynomial time? We proceed as
follows. First, we select the trivial \emph{empty} initialization that consists of no subtours.  Then we run the modified merge
procedure and, in each repetition, we verify that Condition~\eqref{eq:relax}
holds. Note that this condition is easy to verify in time polynomial in $n$. If
it holds until we return  a tour, then we know by Lemma~\ref{lem:polybound}
that the tour has weight at most $9(1+\varepsilon)\alpha\lb(\bar B) + \wD + w(B)$. If it does not hold during
some repetition, then we will restart the algorithm with  a new
initialization that we find using the following lemma. We
continue in this manner until the merge procedure executes without violating
Condition~\eqref{eq:relax} and therefore returns a tour of weight at most
$9(1+\varepsilon)\alpha\lb(\bar B) + \wD + w(B)$.

\begin{lemma}
  \label{lem:polytime}
  Suppose that some repetition of the (modified) merge procedure violates
  Condition~\eqref{eq:relax} when run starting from an initialization $T_1^*, T_2^*, \ldots, T_k^*$ satisfying~\ref{en:i1} and~\ref{en:ii2}. Then we can, in time polynomial in
  $n$, find a new initialization  $ T_1', T_2', \ldots,
   T'_{k'}$ such that~\ref{en:i1},~\ref{en:ii2} are satisfied and
  \begin{align}
    \sum_{j=1}^{ k'} \lbs( T'_j)^2 - \sum_{j=1}^k \lbs(T^*_j)^2 \geq \frac{\varepsilon^2}{3n^2}\lb(\bar B)^2.
	\label{eq:potential}
  \end{align}
\end{lemma}
Note that the above lemma implies that we will reinitialize (in polynomial
time)  at most $3n^2 (1+\varepsilon)^2/\varepsilon^2$ times, because any
initialization $T_1^*, \ldots, T_k^*$ has $\sum_{i=1}^k \lbs(T_i^*)^2 \leq
\left((1+\varepsilon)\lb(\bar B)\right)^2$. As each execution of the merge procedure takes time polynomial in
$n$ and in the running time of $\cA$, we can therefore find a tour of weight at
most $9(1+\varepsilon)\alpha\lb(\bar B) + \wD + w(B)$ in the
time claimed in Theorem~\ref{thm:LoctoGlo}, i.e., polynomial in $n$,
$1/\varepsilon$, and the running time of $\cA$. It remains to prove the
lemma. 
\begin{proof}
  Suppose the $r$-th repetition of the merge procedure violates Condition~\eqref{eq:relax},  that is, there is an $ i \in \{1,2, \ldots, k, \infty\}$ such that
  \begin{align*}
    \lb(\Fatstepxing{r}{i}) > 3\lbs(T_i^*)\,.
  \end{align*}
  Suppose first $i= \infty$. Then $\Fatstepxing{r}{\infty} \neq \emptyset$. Let $T$ be a subtour in $\Fatstepxing{r}{\infty}$. We have $w(T) \leq \alpha \lb(T) \leq \alpha \lbs(T)$ by property~\ref{cond:light} of $\cA$ and $T$ is disjoint from $T^*_1, \ldots, T_k^*$ and $B$  by the definition of $\Fatstepxing{r}{\infty}$. We can therefore compute (in polynomial time) a new initialization  $T'_1 = T_1^*, T'_2 = T^*_2, \ldots, T'_k = T^*_k, T'_{k+1} = T$ with $k' = k+1$ such that~\ref{en:i1},~\ref{en:ii2} are satisfied and
  \begin{align*}
    \sum_{j=1}^{ k'} \lbs( T'_j)^2 - \sum_{j=1}^k \lbs(T^*_j)^2  = \lbs(T'_{k+1})^2 = \lbs(T)^2 \geq  \frac{\varepsilon^2}{n^2}\lb(\bar B)^2\,,
  \end{align*}
  where the last inequality is by the definition of $\lbs$ and  $|V(T)| \geq 1$.

  We now consider the case when $\lb(\Fatstepxing{r}{i}) > 3\lbs(T_i^*)$ for an $i\in \{1,\ldots, k\}$. 
Let $I \subseteq
  \{1, 2, \ldots, k\}$ be the indices of those subtours of
  $T_{1}^*, T_{2}^*, \ldots, T_k^*$ that intersect subtours in $\Fatstepxing{r}{i}$.
  Note that, by definition, we have $i\in I$ and $j\geq i$ for all $j \in I$.
  We construct a new initialization as follows:
  \begin{itemize}
    \item Sort  $I \setminus \{i\} = \{{t_1}, \ldots, {t_{|I|-1}}\}$ so that 
      \begin{align*}
        \frac{\lbs(T^*_{t_j} \setminus \Fatstepxing{r}{i})}{\lbs(T^*_{t_j} \cap \Fatstepxing{r}{i})} \geq 
        \frac{\lbs(T^*_{t_{j+1}} \setminus \Fatstepxing{r}{i})}{ \lbs(T^*_{t_{j+1}}\cap \Fatstepxing{r}{i})}\,,
      \end{align*}
      where for a subtour $T$ we simplify notation by writing  $\lbs(T \setminus \Fatstepxing{r}{i})$ and $\lbs(T \cap \Fatstepxing{r}{i})$ for $\lbs(V(T) \setminus V(\Fatstepxing{r}{i}))$ and $\lbs(V(T) \cap V(\Fatstepxing{r}{i}))$, respectively.
    \item Let $S$ be the minimal (possibly empty) prefix of indices $t_1, t_2, \ldots, t_s$  such that 
      \begin{align*}
        \sum_{j\in S} \lbs(T^*_{j} \setminus  \Fatstepxing{r}{i})  \geq \frac{1}{3} \sum_{j\in I \setminus \{i\}} \lbs(T^*_{j} \setminus  \Fatstepxing{r}{i}) - \lbs(T^*_i)\,.
      \end{align*}
    \item Define $T^*$ to be the subtour obtained by taking $T_i^*$, the union of all the subtours in $\Fatstepxing{r}{i}$, and all the subtours $\{T_j^*\}_{j\in S}$. 
Note that this is a single (i.e., connected) subtour since every subtour in $\Fatstepxing{r}{i}$ intersects $T_i^*$, and every subtour $T_j^*$ with $j\in S\subseteq I$ intersects a subtour in $\Fatstepxing{r}{i}$. 
    \item Reinitialize with subtours $T^*$ and  $\{T^*_j\}_{j\not \in I}$.
  \end{itemize}
  All the above steps can be computed in polynomial time. Moreover,  the new initialization still satisfies~\ref{en:i1} by the definition of $I$ and since $T^*$ does not intersect $B$. We now use the way $S$ was selected to prove that~\ref{en:ii2} still holds and that the ``potential'' function has increased as stated in~\eqref{eq:potential}. As we will calculate below, the increase of the potential function is simply because we required that $\sum_{j\in S} \lbs(T^*_{j} \setminus  \Fatstepxing{r}{i})  \geq \frac{1}{3} \sum_{j\in I \setminus \{i\}} \lbs(T^*_{j} \setminus  \Fatstepxing{r}{i}) - \lbs(T^*_i)$. That~\ref{en:ii2} holds, i.e., that $w(T^*) \leq 3\alpha \lbs(T^*)$, follows since we selected $S$ to be  the minimal prefix with respect to the prescribed ordering, which prefers subtours that have small overlap with $\Fatstepxing{r}{i}$ and therefore contribute significantly to $\lbs(T^*)$.  We now formalize this intuition.
  
  \begin{claim}
    We have $w(T^*) \leq 3\alpha\lbs(T^*)$.
  \end{claim}
  \begin{proof}
    As $S$ is selected to be a minimal prefix and every $j\in I$ satisfies $\lbs(T^*_j) \leq \lbs(T^*_i)$, we claim that
    \begin{align}
      \label{eq:differences}
      \sum_{j\in S} \lbs(T^*_{j} \setminus  \Fatstepxing{r}{i})  \leq \frac{1}{3} \sum_{j\in I \setminus \{i\}} \lbs(T^*_{j} \setminus  \Fatstepxing{r}{i})\,. 
    \end{align}
    This is trivially true if $S = \emptyset$; otherwise, since the prefix $S \setminus \{t_s\}$ was not chosen, we had
    \[ \sum_{j \in S \setminus \{t_s\}}  \lbs(T^*_{j} \setminus  \Fatstepxing{r}{i}) < \frac{1}{3} \sum_{j\in I \setminus \{i\}} \lbs(T^*_{j} \setminus  \Fatstepxing{r}{i}) - \lbs(T^*_i) \]
    and thus indeed
    \[ \sum_{j\in S} \lbs(T^*_{j} \setminus  \Fatstepxing{r}{i})  < \frac{1}{3} \sum_{j\in I \setminus \{i\}} \lbs(T^*_{j} \setminus  \Fatstepxing{r}{i}) \underbrace{- \lbs(T^*_i) + \lbs(T^*_{t_s} \setminus  \Fatstepxing{r}{i})}_{\le 0}
\,. \]
    Moreover, by the sorting of the indices in $I \setminus \{i\}$ we must then also have
    \begin{align}
      \label{eq:intersections}
      \sum_{j\in S} \lbs(T^*_{j} \cap  \Fatstepxing{r}{i})  \le \frac{1}{3} \sum_{j\in I \setminus \{i\}} \lbs(T^*_{j} \cap  \Fatstepxing{r}{i}) \,,
    \end{align}
    which is again trivially true if $S = \emptyset$; otherwise 
    we write
    \begin{align*}
      \frac23  \sum_{j \in S} \lbs(T^*_{j} \cap  \Fatstepxing{r}{i}) \cdot \frac{\lbs(T^*_{t_s} \setminus  \Fatstepxing{r}{i})}{\lbs(T^*_{t_s} \cap  \Fatstepxing{r}{i})}
      &\le
      \frac23  \sum_{j \in S} \lbs(T^*_{j} \cap  \Fatstepxing{r}{i}) \cdot \frac{\lbs(T^*_{j} \setminus  \Fatstepxing{r}{i})}{\lbs(T^*_{j} \cap  \Fatstepxing{r}{i})}
      =
      \frac23  \sum_{j \in S} \lbs(T^*_{j} \setminus  \Fatstepxing{r}{i})
      \\
      &\le 
      \frac{1}{3} \sum_{j\in I \setminus \{i\} \setminus S} \lbs(T^*_{j} \setminus  \Fatstepxing{r}{i})
      =
      \frac{1}{3} \sum_{j\in I \setminus \{i\} \setminus S} \lbs(T^*_{j} \cap  \Fatstepxing{r}{i}) \cdot \frac{\lbs(T^*_{j} \setminus  \Fatstepxing{r}{i})}{\lbs(T^*_{j} \cap  \Fatstepxing{r}{i})}
      \\
      &\le
      \frac13 \sum_{j\in I \setminus \{i\} \setminus S} \lbs(T^*_{j} \cap  \Fatstepxing{r}{i}) \cdot \frac{\lbs(T^*_{t_s} \setminus  \Fatstepxing{r}{i})}{\lbs(T^*_{t_s} \cap  \Fatstepxing{r}{i})}
      \,,
    \end{align*}
    where the middle inequality is by subtracting $\frac13 \sum_{j \in S} \lbs(T^*_{j} \setminus  \Fatstepxing{r}{i})$ from both sides of~\eqref{eq:differences}, and the remaining two are due to our sorting of indices.
    Next, we divide both sides by
$\frac{\lbs(T^*_{t_s} \setminus  \Fatstepxing{r}{i})}{\lbs(T^*_{t_s} \cap  \Fatstepxing{r}{i})}$
    (which is nonzero by minimality of $S$) and add $\frac13 \sum_{j \in S} \lbs(T^*_{j} \cap  \Fatstepxing{r}{i})$ to both sides to obtain~\eqref{eq:intersections}.
    
    By~\eqref{eq:intersections} we have
    \begin{align*}
      \sum_{j\in S} w(T_j^*) \leq 3\alpha \sum_{j\in S} \lbs(T_j^*) &= 3\alpha \sum_{j\in S} \lbs(T_j^* \setminus \Fatstepxing{r}{i}) + 3\alpha \sum_{j\in S} \lbs(T_j^* \cap \Fatstepxing{r}{i}) \\
&      \leq 3\alpha \sum_{j\in S} \lbs(T_j^* \setminus \Fatstepxing{r}{i}) + \alpha \sum_{j\in I \setminus \{i\}} \lbs(T_j^* \cap \Fatstepxing{r}{i}) \\
& \leq 3\alpha \sum_{j\in S} \lbs(T_j^* \setminus \Fatstepxing{r}{i}) + \alpha \lbs(\Fatstepxing{r}{i})\,,
    \end{align*}
    where the first inequality holds because $T^*_1, \ldots, T^*_k$ satisfy~\ref{en:ii2} and the last inequality holds because they are disjoint~(\ref{en:i1}). By property~\ref{cond:light} of $\cA$ and by the assumption that $\lb(\Fatstepxing{r}{i})> 3\lbs(T_i^*)$ we also have respectively that
    \begin{align*}
      w(\Fatstepxing{r}{i}) \leq \alpha \lb(\Fatstepxing{r}{i}) \leq \alpha \lbs(\Fatstepxing{r}{i}) \qquad \mbox{and} \qquad w(T_i^*) \leq 3\alpha \lbs(T_i^*) < \alpha \lbs(\Fatstepxing{r}{i})\,.
    \end{align*}
     These inequalities imply the claim since
    \begin{align*}
      w(T^*) & = w(\Fatstepxing{r}{i}) + w(T_i^*) + \sum_{j\in S } w(T^*_j) \\
      & < \alpha \lbs(\Fatstepxing{r}{i}) + \alpha \lbs(\Fatstepxing{r}{i})  + 3\alpha \sum_{j\in S} \lbs(T_j^* \setminus \Fatstepxing{r}{i}) + \alpha \lbs(\Fatstepxing{r}{i}) \\
      & =  3\alpha \lbs(\Fatstepxing{r}{i}) +  3\alpha \sum_{j\in S} \lbs(T_j^* \setminus \Fatstepxing{r}{i}) \\
      & \leq 3\alpha \lbs(\Fatstepxing{r}{i}) +  3\alpha \sum_{j\in S} \lbs(T_j^* \setminus \Fatstepxing{r}{i})  + 3\alpha \lbs(T_i^* \setminus \Fatstepxing{r}{i}) \\
     & = 3\alpha \lbs(T^*)\,.
    \end{align*}
  \end{proof}

  It remains to verify the increase of the ``potential'' function as stated in~\eqref{eq:potential}. By the definition of the new initialization, the increase is
  \begin{align*}
    \lbs(T^*)^2 - \sum_{j\in I} \lbs(T^*_j)^2\,.
  \end{align*}

  Let us concentrate on the first term:
	\begin{align*}
    \lbs(T^*)^2& = \left( \lbs(\Fatstepxing{r}{i}) + \lbs(T_i^* \setminus \Fatstepxing{r}{i}) + \sum_{j\in S} \lbs(T_j^* \setminus \Fatstepxing{r}{i})\right)^2 \\
    &\geq \lbs(\Fatstepxing{r}{i})\left( \lbs(\Fatstepxing{r}{i}) + \lbs(T_i^* \setminus \Fatstepxing{r}{i}) + \sum_{j\in S} \lbs(T_j^* \setminus \Fatstepxing{r}{i})\right).
  \end{align*}
  By the selection of $S$, the expression inside the parenthesis is at least
  \begin{align*}
    \lbs(\Fatstepxing{r}{i}) + \lbs(T_i^* \setminus \Fatstepxing{r}{i}) \, + & \, \frac{1}{3}\sum_{j\in I\setminus \{i\}} \lbs(T_j^* \setminus \Fatstepxing{r}{i}) - \lbs(T_i^*) \\
   & \geq 
   \lbs(\Fatstepxing{r}{i}) +  \frac{1}{3}\sum_{j\in I} \lbs(T_j^* \setminus \Fatstepxing{r}{i}) - \lbs(T_i^*).
  \end{align*}
  Using $\lbs(T_i^*) < \lbs(\Fatstepxing{r}{i})/3$, we can further lower-bound this expression by
  \begin{align*}
    \frac{1}{3}\lbs(\Fatstepxing{r}{i})  + \frac{1}{3}\left(\lbs(\Fatstepxing{r}{i}) + \sum_{j\in I} \lbs(T_j^* \setminus \Fatstepxing{r}{i})\right) \ge  \frac{1}{3}\lbs(\Fatstepxing{r}{i})  + \frac{1}{3}\sum_{j\in I} \lbs(T_j^*).
  \end{align*}
  Finally, as  $\lbs(T_j^*) \leq \lbs(T_i^*)$ for all $j\in I$, we have  
  \begin{align*}
    \lbs(T^*)^2 - \sum_{j\in I} \lbs(T_j^*)^2 & \geq \lbs(T^*)^2 - \lbs(T_i^*) \sum_{j\in I} \lbs(T_j^*) \\
    & \geq \lbs(\Fatstepxing{r}{i}) \left( \frac{1}{3}\lbs(\Fatstepxing{r}{i})  + \frac{1}{3}\sum_{j\in I} \lbs(T_j^*) \right) - \lbs(T_i^*) \sum_{j\in I} \lbs(T_j^*) \\
    &= \frac13 \lbs(\Fatstepxing{r}{i})^2 + \underbrace{\left(\frac{\lbs(\Fatstepxing{r}{i})}{3} - \lbs(T_i^*)\right)}_{>0} \sum_{j\in I} \lbs(T_j^*) \\
    &\geq  \frac13 \left( \varepsilon \frac{\lb(\bar B)}{n} \right)^2
    =
    \frac{\varepsilon^2}{3n^2} \lb(\bar B)^2 \,,
  \end{align*}
which completes the proof of Lemma~\ref{lem:polytime}.
\end{proof}

\part{Obtaining structured instances}
\label{part:structure}
In the previous part we have reduced the problem of approximating ATSP to that of designing algorithms for \EPC{}. Our approach for dealing with general instances is to first simplify their structure  and then to solve \EPC{} on the resulting structured instances. In this part, we show that we can obtain very structured instances by only increasing the approximation guarantee by a constant factor. In Part~\ref{part:solvingvert} we then solve \EPC{} on those instances.

The outline of this part is as follows. 
We begin by exploring the structure of sets in the laminar family $\cL$: in
Section~\ref{sec:paths} we study paths inside  sets $S\in \cL$ and, in
Section~\ref{sec:contractinduce}, we introduce analogues of the classic
graph-theoretic operations of contracting and inducing on such a set. These
operations naturally give rise to a  recursive algorithm that, intuitively, works
as long as the contraction of some set $S\in \cL$ results in a significant
decrease in the value of the LP relaxation.  In
Section~\ref{sec:reducetoirreducible} we formally analyze this recursive algorithm
and reduce the  task of approximating ATSP to that of
approximating ATSP on \emph{irreducible} instances: those where no set $S\in \cL$
brings about a significant decrease of the LP value if contracted.

Informally, every set $S\in \cL$ in an irreducible instance has two
vertices $u,v\in S$ such that the shortest path from $u$ to $v$ crosses
a large (weighted) fraction  of the sets $R\in \cL : R \subsetneq S$
(otherwise contracting $S$ into a single vertex,
endowed with a node-weight equal to the weight of the shortest path, would lead to a decrease in the LP
value).  This insight, together with the approximation algorithm for
singleton instances in \cref{part:epc} (\cref{cor:singleton}), allows us to construct a
low-weight subtour $B$  that does not necessarily visit every vertex, but crosses every non-singleton
set of $\cL$. See the right part of Figure~\ref{fig:intro} for an example. We refer
to $B$ as a \emph{backbone}, and to the ATSP instance and the backbone
together as a \emph{vertebrate pair}.
This reduction allows us to further assume that our input
is such a vertebrate pair; it
is presented in Section~\ref{sec:reducetovertebrate}.

In each of the above stages, we prove a theorem of the form:
if there is a constant-factor approximation for ATSP on more structured instances,
then there is a constant-factor approximation for ATSP on less structured instances.
For example, an algorithm for irreducible instances
implies an algorithm for laminarly-weighted instances.
One can also think of
making a stronger and stronger assumption on the instance
without loss of generality,
making it increasingly resemble a singleton instance.

\section{Paths in Tight Sets}
\label{sec:paths}
\begin{figure}[t]
  \centering
  \begin{tikzpicture}
\tikzset{arrow data/.style 2 args={decoration={markings,
         mark=at position #1 with \arrow{#2}},
         postaction=decorate}
      }

  \begin{scope}[scale=0.6]
    \begin{scope}
\draw[fill=gray!10!white, draw=gray!80!black] (0, 0) ellipse (4cm and 2cm); \end{scope}
    \begin{scope}
      \clip (0,0) ellipse (4cm and 2cm); 
      \draw (-2.3,-3) -- (-2.3,3);  
      \draw (-0.5,-3) -- (-0.5,3);  

      \draw (2.3,-3) -- (2.3,3);  
\end{scope}
      \node at (-3.1, 0) {\small $S_1$};
      \node at (-1.4, 0) {\small $S_2$};
      \node at (3.1, 0) {\small $S_\ell$};
      \node at (0.9, 0) {\Large $\ldots$};

      \draw[->, line width=0.4mm] (-4.6,0) -- (-3.55,0);
      \draw[->, line width=0.4mm] (-2.8,0) -- (-1.75,0);
      \draw[->, line width=0.4mm] (-0.9,0) -- (0.15,0);
      \draw[->, line width=0.4mm] (1.75,0) -- (2.8,0);
      \draw[->, line width=0.4mm] (3.55,0) -- (4.6,0);
      \node at (0,-3.5) {(a)};
  \end{scope}

  \begin{scope}[xshift=7.5cm,scale=0.8]
    \draw[fill=gray!10!white, draw=gray!80!black] (0, 0) ellipse (4cm and 2cm); \begin{scope}
      \draw[fill=gray!30!white, draw=gray!80!black] (-2.5, 0) ellipse (0.75cm and 1cm); \end{scope}
    \begin{scope}[xshift=1.0cm,rotate=30]
      \draw[fill=gray!30!white, draw=gray!80!black] (0, 0) ellipse (1cm and 1.5cm); \draw[fill=gray!50!white, draw=gray!80!black, rotate=10] (0, .75) ellipse (0.5cm and 0.5cm); \draw[fill=gray!50!white, draw=gray!80!black,rotate=-5] (0, -.75) ellipse (0.5cm and 0.6cm); \end{scope}
    \node[sssgvertex, fill=black] (u) at (-2.5, 0) {};
    \node  at (-2.5, 0.4) {$u$};
    \node[sssgvertex, fill=black] (v) at (2.9, 0.4) {};
    \node  at (2.9, 0.8) {$v$};
    \draw [black,arrow data={0.2}{stealth},
               arrow data={0.4}{stealth},
               arrow data={0.6}{stealth},
               arrow data={0.8}{stealth}] plot [smooth,tension=1] coordinates { (u) (-1.8,-1.2)  (-0.8,-0.4) (0,0) (1,0.5) (1.5, 0) (2,0.5) (2.7,-0.5)  (v)};
\node at (0,-2.7) {(b)};
  \end{scope}
\end{tikzpicture}
   \caption{\textbf{(a)} The structure of a tight set $S$ with strongly connected components $S_1, \ldots, S_\ell$. Every path traversing $S$ enters at a  vertex in $S_{\textrm{in}} \subseteq S_1$, then visits all strongly connected components, which form a ``path'' structure, before it exits from a vertex in $S_{\textrm{out}}\subseteq S_\ell$. \\ \textbf{(b)} The structure of the path $P$ from $u\in S_{\textrm{in}}$ to $v\in S_{\textrm{out}}$ for a tight set $S$ as given by Lemma~\ref{lem:short_path}. The path crosses the set that contains $u$ but not $v$ once and it crosses the sets of $\cL$  that are disjoint from $\{u,v\}$  at most twice. }
  \label{fig:structure}
\end{figure}

An instance $\cI = (G, \cL, \xs, y)$ will be fixed throughout this section. 
We say that a path $P$ \emph{crosses a set $S$ $k$ times} if $|P\cap \delta(S)|=k$.
We say
that a path $P$ \emph{traverses} a set $S$ if both endpoints of
$P$ are in $V\setminus S$ and $P$ crosses $S$ at least \bll{twice (once entering and once leaving).}
We now exhibit properties of paths traversing tight sets. In
particular, we show that
the strongly connected components of a tight set $S$ enjoy a nice path-like structure as depicted in Figure~\ref{fig:structure}.

Recall from Section~\ref{sec:prelim} that a set $S$ of vertices is tight if $x(\delta(S)) = 2$. Moreover, $S_{\mathrm{in}}$  and $S_{\mathrm{out}}$ denote those vertices of $S$ that have an incoming edge from outside of $S$ and those that have an outgoing edge to outside of $S$, respectively.
\begin{lemma}
  For a tight set $S \subsetneq V$ we have the following properties:
  \begin{enumerate}[label=\emph{(\alph*)}]\itemsep0mm
    \item Every path from a vertex $u\in S_{\textrm{in}}$ to a vertex $v \in
      S_{\textrm{out}}$ (and thus every path traversing $S$)  visits every strongly
      connected component of $S$.
    \label{item:path_through_every_scc}
    \item For every $u\in S_{\textrm{in}}$ and $v \in S$ there is a path from $u$ to $v$ inside $S$.
    The same holds for every $u \in S$ and $v \in S_{\textrm{out}}$.
    \label{item:exists_path}
  \end{enumerate}
  \label{lem:path_tight_set}
\end{lemma}
\begin{proof}
  We remark that~\ref{item:path_through_every_scc} can also be seen to  follow
  from the ``$\tau$-narrow cut'' structure  as introduced
  in~\cite{AnKS15} by setting $\tau=0$.  We give
  a different proof that we find simpler for our setting.
  
  Let $S_1,
  S_2, \ldots , S_\ell$ be the vertex sets of the  strongly connected
  components of $S$, indexed using a topological ordering. We thus have that
  each subgraph $G[S_i]$ is strongly connected and that there is no edge from
  a vertex in $S_i$ to a vertex in $S_j$ if $i> j$. 
  
  By the above, we must have $\delta^-(S_1) \subseteq \delta^-(S)$.
  Moreover, since $\xs(\delta^-(S_1)) \ge 1 = \xs(\delta^-(S))$,
  we can further conclude that $\delta^-(S_1) = \delta^-(S)$
  (recall that all edges have positive $\xs$-value)
  and $\xs(\delta^-(S_1)) = \xs(\delta^+(S_1)) = 1$
  (i.e., $S_1$ is a tight set).
  
  Similarly, we can show by induction on $k \ge 2$ that
  \[ \delta^-(S_k) = \delta^+(S_{k-1}) \quad \text{ and } \quad \xs(\delta^-(S_k)) = \xs(\delta^+(S_k)) = 1. \]
  To see this, note that
  $\delta^-(S_k) \subseteq \delta^-(S) \cup \bigcup_{i < k} \delta^+(S_i)$.
  However, $\delta^-(S) = \delta^-(S_1)$ (which is disjoint from $\delta^-(S_k)$),
  and for $i < k-1$, the induction hypothesis gives that $\delta^+(S_i) = \delta^-(S_{i+1})$ (which is also disjoint from $\delta^-(S_k)$).
  The only term left in the union is $i = k-1$ and so $\delta^-(S_k) \subseteq \delta^+(S_{k-1})$.
  Moreover, $1 \le \xs(\delta^-(S_k)) \le \xs(\delta^+(S_{k-1})) = 1$,
  which implies the statement for $k$.
  
  Finally, we have that $\delta^+(S_\ell) = \delta^+(S)$.
  To recap, all incoming edges of $S$ are into $S_1$,
  the set of outgoing edges of every component is the set of incoming edges of the next one,
  and all outgoing edges of $S$ are from $S_\ell$.
  This shows~\ref{item:path_through_every_scc}, i.e., 
  that every path traversing $S$ needs to enter through $S_1$,
  exit through $S_\ell$,
  and pass through every component on the way.

  Finally,~\ref{item:exists_path} follows because $S_{\textrm{in}} \subseteq S_1$ (similarly $S_{\textrm{out}} \subseteq S_\ell$),
  each two consecutive components are connected by an edge,
  and each component is strongly connected.
\end{proof}

\begin{lemma}
    Let $S \subsetneq V$ be a non-empty set such that $\cL \cup \{S\}$ is
    a laminar family.  Suppose $u, v \in S$ are two vertices such that there is
    a path from $u$ to $v$ inside $S$. Then we can in polynomial time find
    a path $P$ from $u$ to $v$ inside $S$ that crosses every set in $\cL$ at
    most twice.  Thus, the path satisfies $\cost(P) \le \sum_{R \in \cL
    : \ R \subsetneq S} 2 \cdot y_R = \valu(S)$.
    
    In addition, if  $u \in S_{\textrm{in}}$ or $v \in S_{\textrm{out}}$, then
    $P$ 
    crosses every  tight set $R\in\cL$, $R\subsetneq S$ at most
    $2 - |R \cap \{u,v\}|$ times.
    Thus, it satisfies
    $\cost(P) \le \sum_{R \in \cL  : \ R \subsetneq S} \left( 2 - |R \cap \{u,v\}| \right) \cdot y_R$.
    
    \label{lem:short_path}
\end{lemma}
\begin{proof}
  Since  $\cL \cup \{S\}$ is a laminar family,
  any path inside $S$ only crosses those sets $R \in \cL$ that have $R \subsetneq S$.
  Now, to prove both statements, 
  it is enough to find, in polynomial time,  a path $P$ inside $S$ that for each $R\in \cL$ with $R \subsetneq S$ satisfies
  \begin{align*}
      |P \cap \delta(R)| \leq \begin{cases}
     2 & \mbox{if $|R\cap\{u,v\}| = 0$,} \\
     1 & \mbox{if $|R\cap\{u,v\}| = 1$,} \\
     2 & \mbox{if $|R\cap\{u,v\}| = 2$,} \\
     0 & \mbox{if $|R\cap\{u,v\}| = 2$ and $u \in S_\textrm{in}$ or $v \in S_\textrm{out}$.} \\
   \end{cases}
  \end{align*}

  The algorithm for finding  $P$ starts with any path $P$ from $u$ to $v$
  inside $S$. Such a path is guaranteed to exist by the assumptions of the
  lemma and can be easily found in polynomial time.  Now, while $P$ does not
  satisfy the above conditions, select a set $R\in \cL$ of \emph{maximum}
  cardinality that violates one of the above conditions. 
  We remark that the selected set $R$ is tight since $R \in \cL$. Therefore Lemma~\ref{lem:path_tight_set}\ref{item:exists_path} implies that there is a path from any $u' \in R$ to any $v' \in R$ inside $R$ if $u' \in R_{\textrm{in}}$ or $v' \in R_{\textrm{out}}$. Using this, the algorithm now modifies $P$  depending on which
  of the above conditions is violated:
  
  \begin{description}
    \item[Case 1: $|R\cap \{u,v\}| = 0.$]  Let $u'$ be the first vertex visited by $P$ in $R$ and let $v'$ be the last.
      Then $u' \in R_{\textrm{in}}$ and $v' \in R_{\textrm{out}}$, which implies by Lemma~\ref{lem:path_tight_set}\ref{item:exists_path}
  that there is a path $Q$ from $u'$ to $v'$ inside $R$. 
  We update $P$ by letting $Q$ replace the segment of $P$ from $u'$ to  $v'$.
  This ensures that the set $R$ is no longer violated, since the path $P$ now only enters and exits $R$ once.

  \item[Case 2: $|R \cap \{u,v\}| = 1.$] This case is similar to the previous one. Suppose that $u\in R $ and $v\not \in R$ (the other case is analogous). 
  Let $v'$ be the last vertex visited by $P$ in $R$.
  Then $v' \in R_{\textrm{out}}$, and again by Lemma~\ref{lem:path_tight_set}\ref{item:exists_path} there is a path  $Q$ from $u$ to $v'$ inside $R$. We update $P$ by letting $Q$ replace the segment of $P$ from $u$ to $v'$.
  This ensures that the set $R$ is no longer violated, since the path $P$ now only exits $R$ once.
  
\item[Case 3: $|R \cap \{u,v\}| = 2$.]
  Let $u'$ be the first vertex visited by $P$ in $R_\textrm{in}$.  By Lemma~\ref{lem:path_tight_set}\ref{item:exists_path}, there is a path  $Q$ from $u'$ to $v$ inside $R$. We
  modify $P$ by letting $Q$ replace the segment of $P$ from $u'$ to $v$.
  This ensures that the set $R$ is no longer violated, since the path $P$ now only enters and exits $R$ at most once.

\item[Case 4: {\normalfont $|R\cap\{u,v\}| = 2$ and $u \in
  S_\textrm{in}$ or $v \in S_\textrm{out}$.}]  Suppose that $u\in
  S_{\textrm{in}}$ (the case $v\in S_{\textrm{out}}$ is analogous). Then,
  as $R \subseteq S$, $R \cap S_\textrm{in} \subseteq R_\textrm{in}$. So, by
  Lemma~\ref{lem:path_tight_set}\ref{item:exists_path}, there is a path $Q$ from $u$ to $v$
  inside $R$. We replace $P$ by $Q$ and the set $R$ is no longer violated. 
\end{description}

At termination, the above algorithm returns a path satisfying all the desired
conditions and thus the lemma. It remains to argue  that the algorithm
terminates in polynomial time. A laminar family contains at most $2n-1$ sets,
so it is easy to efficiently identify a violated set $R$ of maximum
cardinality. The algorithm then, in polynomial time, modifies $P$ by simple path
computations so that the set $R$ is no longer violated. Moreover, since the
modifications are such that new edges are only added within the set $R$, they may only introduce new
violations to sets contained in $R$ -- sets of smaller cardinality. It follows, since
we always select a violated set of maximum cardinality, that any set $R$ in $\cL$  is
selected  in at most one iteration. Hence, the algorithm runs for at most
$|\cL| \leq 2n-1$ iterations (and so it terminates in polynomial time).
\end{proof}

\bll{The next lemma asserts that the strongly connected components of $S\in\cL$ form a laminar family together with $\cL$, a  property that will be needed in Section~\ref{sec:induce} in order to apply Lemma~\ref{lem:short_path}.}
  \begin{lemma}\label{lem:lam-strongly}
For a set $S\in \cL$, let $\mathcal{S}$
be the set of strongly connected components of $S$.
 Then  $\cL \cup \mathcal{S}$ is a laminar family.
  \end{lemma}
  \begin{proof}
 Let $\mathcal{S}=\{S_1,S_2,\ldots,S_\ell\}$,  indexed in a topological order.
    Towards a contradiction, suppose there exists a component $S_i$ and a set 
$R\in \cL$ such that $R\setminus S_i, S_i \setminus R$, and $S\cap R$ are all non-empty. Furthermore, since $\cL$ is a laminar family and $S_i \subseteq S\in \cL$, we must have $R \subsetneq S$. We can thus partition $R$ into the three sets
    \begin{align*}
      R_{<i} = R \cap (S_1 \cup \cdots \cup S_{i-1}), \qquad R_i = R \cap S_i, \qquad  R_{>i} = R \cap (S_{i+1} \cup \cdots \cup S_\ell)\,.
    \end{align*}
    In words, $R_{<i}$ is the part of $R$ that intersects vertices of the strongly connected components that are ordered topologically before $S_i$. Similarly, $R_{>i}$ is the part of $R$  that intersects vertices of the strongly connected components that are ordered topologically after $S_i$. Note that since $R$ is not contained in $S_i$, we have that either $R_{<i}$ or $R_{>i}$ is non-empty. We suppose $R_{<i} \neq \emptyset$ (the case $R_{>i} \neq \emptyset$ is analogous). 

    As $x$ is a feasible solution to $\LP(G,\mathbf{0})$, we have
    $x(\delta^-(R_{<i})) \geq 1$. Moreover, since $\delta(R_i\cup
    R_{>i},R_{<i})=\emptyset$ due to the topological ordering, we have
    $\delta^-(R_{<i}) \subseteq \delta^-(R)$ and thus
    \begin{align*}
      1  = x(\delta^-(R)) \geq x(\delta^-(R_{<i})) + x(\delta(S_i\setminus R_i, R_i))  \geq 1 + x(\delta(S_i\setminus R_i, R_i))\,,
    \end{align*}
    where the first equality follows since $R\in \cL$ is a tight set. However, this is a contradiction because $x(\delta(S_i \setminus R_i, R_i)) >0$; that holds since $S_i \setminus R_i = S_i \setminus R \neq \emptyset$ and $R_i \ne \emptyset$, $S_i$ is a strongly connected component, and $G$ only contains edges with strictly positive $x$-value.
  \end{proof}

 \newcommand{\cheapcontract}[0]{cheap-to-contract\xspace}
\newcommand{\expensivecontract}[0]{expensive-to-contract\xspace}
\newcommand{\contractible}[0]{contractible\xspace}
\newcommand{\reducible}[0]{reducible\xspace}

\section{Contracting and Inducing on a Tight Set}
\label{sec:contractinduce}
  In this section
  we generalize two natural graph-theoretic constructions
  that allow one to decompose the problem of finding a tour with respect to a vertex set $S$.
  The first relies on contracting $S$ (see Definition~\ref{def:contraction} in Section~\ref{sec:contraction})
  and the second relies on inducing on $S$ (see Definition~\ref{def:induced_instance} in Section~\ref{sec:induce}).

  \subsection{Contracting a Tight Set}
  \label{sec:contraction}
   
  Consider an instance $\mathcal{I} = (G, \cL,\xs, y)$.  Before defining the
  contraction of a set $S\in \cL$, we need to define the ``distance''
  functions $d_S$ and $D_S$. For $S\in \cL$ and $u,v\in S$, define
  $d_S(u,v)$ to be the minimum weight of a  path inside $S$ from $u$ to $v$   (if no such path exists, $d_S(u,v) = \infty$). We also let 
  \begin{align*}
    D_{S}(u,v) = \sum_{R\in \cL: \ u\in R \subsetneq S} y_R  + d_S(u,v) + \sum_{R\in \cL: \ v\in R \subsetneq S} y_R\,,
  \end{align*}
  which equals $d_S(u,v) + \sum_{R\in \cL:\   R \subsetneq S}|R\cap \{u,v\}|\cdot y_R$ if $u\neq v$, and $ \sum_{R\in \cL:\   u\in R \subsetneq S} 2y_R$ if $u=v$.
  We remark that $D_S(u,u)$ might be strictly positive.

\ifdefined\ACM
\begin{figure}[t]
  \centering
\begin{tikzpicture}[scale=0.70]
\tikzset{arrow data/.style 2 args={decoration={markings,
         mark=at position #1 with \arrow{#2}},
         postaction=decorate}
      }

  \begin{scope}[scale=0.6]
    \node at (0, 4.6) {\footnotesize Instance $\cI$};
    \draw[fill=gray!10!white, draw=gray!80!black] (0, 0) ellipse (4cm and 2cm) node[above = 1.20cm] {$S$}node[above left = 1.00cm and 0.9cm] {\scriptsize $5$};
    \begin{scope}
      \draw[fill=gray!30!white, draw=gray!80!black] (-2.5, 0.5) ellipse (0.75cm and 0.5cm) node[above  left= 0.20cm and -0.15cm] {\scriptsize $2$};
    \end{scope}
    \begin{scope}[xshift=1.8cm,rotate=30]
      \draw[fill=gray!30!white, draw=gray!80!black] (0, 0) ellipse (1.3cm and 1.5cm) node[above left= 0.65cm and 0.40cm] {\scriptsize $4$};
\draw[fill=gray!50!white, draw=gray!80!black] (0, -.75) ellipse (0.3cm and 0.4cm) node[above left = .12cm and 0.0cm] {\scriptsize $3$};
       \node[ssssgvertex, fill=black] (v1) at (0.1, 0.85) {};
\node[ssssgvertex, fill=black]  at (-0.0, -0.55) {};
       \node[ssssgvertex, fill=black] (v2) at (-0.0, -0.85) {};
       \node[ssssgvertex, fill=black] (t)at (-0.85, -0.15) {};
    \end{scope}
    \node[ssssgvertex, fill=black] (u1) at (-2.7, 0.4) {};
    \node[ssssgvertex, fill=black]  at (-2.2, 0.7) {};
    \node[ssssgvertex, fill=black] (u2) at (-2.5, -0.7) {};
\node[ssssgvertex, fill=black] (v) at (2.6, 0.2) {};
    \draw[fill=gray!30!white, draw=gray!80!black] (-0.7, .7) ellipse (0.2cm and 0.2cm) node[above = 0.08cm] {\scriptsize $1$};
    \node[ssssgvertex, fill=black] (mu)  at (-0.7, 0.7) {};
    \draw[fill=gray!30!white, draw=gray!80!black] (-0.7, -1.2) ellipse (0.2cm and 0.2cm) node[above = 0.08cm] {\scriptsize $2$};
    \node[ssssgvertex, fill=black] (md) at (-0.7, -1.2) {};
    \begin{scope}[xshift=-1.0cm]
      \draw[fill=gray!10!white, draw=gray!80!black, rotate around={55:(3.2, 3.0)}] (3.2, 3.0) ellipse (0.7cm and 1.2cm); \draw[fill=gray!30!white, draw=gray!80!black] (3.5, 2.8) ellipse (0.4cm and 0.4cm); \node[ssssgvertex, fill=black] (b1p) at (3.65, 2.9) {};
      \node[ssssgvertex, fill=black]  (b1) at (3.35, 2.65) {};
      \node[ssssgvertex, fill=black]  at (2.6, 3.4) {};
    \end{scope}
    \begin{scope}
      \draw[fill=gray!10!white, draw=gray!80!black] (-1.5, 3.2) ellipse (0.2cm and 0.2cm); \node[ssssgvertex, fill=black]  at (-1.5, 3.2) {};
      \draw[fill=gray!10!white, draw=gray!80!black] (-3.9, 2.0) ellipse (0.2cm and 0.2cm); \node[ssssgvertex, fill=black] (a1)  at (-3.9, 2.0) {};
    \end{scope}
    \begin{scope}
      \draw[fill=gray!10!white, draw=gray!80!black] (3.6, -2.3) ellipse (0.2cm and 0.2cm); \node[ssssgvertex, fill=black] (b2) at (3.6, -2.3) {};
    \end{scope}
    \begin{scope}
      \draw[fill=gray!10!white, draw=gray!80!black, rotate around={0:(0.0, 0.0)}] (0.0, -3.0) ellipse (0.6cm and 0.5cm); \node[ssssgvertex, fill=black]  at (0.25, -3.1) {};
      \node[ssssgvertex, fill=black]  at (-0.25, -2.9) {};
    \end{scope}
    \begin{scope}
      \draw[fill=gray!10!white, draw=gray!80!black] (-3.6, -2.3) ellipse (0.2cm and 0.2cm); \node[ssssgvertex, fill=black] (a2) at (-3.6, -2.3) {};
    \end{scope}
    \draw (a1) edge[dashed, ->] (u1);
    \draw (a2) edge[dashed, ->] (u2);
    \draw (v1) edge[dashed, ->] (b1);
    \draw (v) edge[dashed, ->, bend right] (b1p);
    \draw (v2) edge[dashed, ->] (b2);
\end{scope}

  \begin{scope}[scale=0.6, xshift=10cm]
    \node at (0, 4.6) {\footnotesize A tour of the instance $\cI/S$ };
    \node[sgvertex] (s) at (0, 0) {$s$};
    \begin{scope}[xshift=-1.0cm]
      \draw[fill=gray!10!white, draw=gray!80!black, rotate around={55:(3.2, 3.0)}] (3.2, 3.0) ellipse (0.7cm and 1.2cm); \draw[fill=gray!30!white, draw=gray!80!black] (3.5, 2.8) ellipse (0.4cm and 0.4cm); \node[ssssgvertex, fill=black] (b1p) at (3.65, 2.9) {};
      \node[ssssgvertex, fill=black]  (b1) at (3.35, 2.65) {};
      \node[ssssgvertex, fill=black](c)  at (2.6, 3.4) {};
    \end{scope}
    \begin{scope}
      \draw[fill=gray!10!white, draw=gray!80!black] (-1.5, 3.2) ellipse (0.2cm and 0.2cm); \node[ssssgvertex, fill=black] (c2) at (-1.5, 3.2) {};
      \draw[fill=gray!10!white, draw=gray!80!black] (-3.9, 2.0) ellipse (0.2cm and 0.2cm); \node[ssssgvertex, fill=black] (a1)  at (-3.9, 2.0) {};
    \end{scope}
    \begin{scope}
      \draw[fill=gray!10!white, draw=gray!80!black] (3.6, -2.3) ellipse (0.2cm and 0.2cm); \node[ssssgvertex, fill=black] (b2) at (3.6, -2.3) {};
    \end{scope}
    \begin{scope}
      \draw[fill=gray!10!white, draw=gray!80!black, rotate around={0:(0.0, 0.0)}] (0.0, -3.0) ellipse (0.6cm and 0.5cm); \node[ssssgvertex, fill=black] (d1) at (0.25, -3.1) {};
      \node[ssssgvertex, fill=black] (d2) at (-0.25, -2.9) {};
    \end{scope}
    \begin{scope}
      \draw[fill=gray!10!white, draw=gray!80!black] (-3.6, -2.3) ellipse (0.2cm and 0.2cm); \node[ssssgvertex, fill=black] (a2) at (-3.6, -2.3) {};
    \end{scope}
    \draw (a1) edge[->] node[below left = -0.1cm and -0.1cm] {\scriptsize $(u^1_{\textrm{in}}, s)$}(s);
    \draw (a2) edge[->] node[above left= -0.1cm and -0.1cm] {\scriptsize $(u^2_{\textrm{in}}, s)$} (s);
\draw (s) edge[->, bend right=60] node[below right= -0.1cm and -0.1cm] {\scriptsize $(s, v^2_{\textrm{out}})$} (b1p);
    \draw (s) edge[->] node[above right= -0.1cm and -0.1cm] {\scriptsize $(s, v^1_{\textrm{out}})$}(b2);
    \draw (b1p) edge[->] (b1);
    \draw (b1) edge[->] (c);
    \draw (c) edge[->] (c2);
    \draw (c2) edge[->] (a1);
    \draw (b2) edge[->] (d1);
    \draw (d1) edge[->] (d2);
    \draw (d2) edge[->] (a2);

\end{scope}

  \begin{scope}[scale=0.6, xshift=22cm]
    \node at (0, 4.6) {\footnotesize The lift of the tour to a subtour of $\cI$ };
    \draw[fill=gray!10!white, draw=gray!80!black] (0, 0) ellipse (4cm and 2cm) node[above = 1.20cm] {$S$} node[above left = 1.00cm and 0.9cm] {\scriptsize $5$};
    \begin{scope}
      \draw[fill=gray!30!white, draw=gray!80!black] (-2.5, 0.5) ellipse (0.75cm and 0.5cm) node[above  left= 0.20cm and -0.15cm] {\scriptsize $2$};
    \end{scope}
    \begin{scope}[xshift=1.8cm,rotate=30]
      \draw[fill=gray!30!white, draw=gray!80!black] (0, 0) ellipse (1.3cm and 1.5cm) node[above left= 0.65cm and 0.40cm] {\scriptsize $4$};
\draw[fill=gray!50!white, draw=gray!80!black] (0, -.75) ellipse (0.3cm and 0.4cm) node[above left = .12cm and 0.0cm] {\scriptsize $3$};
       \node[ssssgvertex, fill=black] (v1) at (0.1, 0.85) {};
\node[ssssgvertex, fill=black]  at (-0.0, -0.55) {};
       \node[ssssgvertex, fill=black] (v2) at (-0.0, -0.85) {};
       \node[ssssgvertex, fill=black] (t)at (-0.85, -0.15) {};
    \end{scope}
    \node[ssssgvertex, fill=black] (u1) at (-2.7, 0.4) {};
    \node[ssssgvertex, fill=black]  at (-2.2, 0.7) {};
    \node[ssssgvertex, fill=black] (u2) at (-2.5, -0.7) {};
\node[ssssgvertex, fill=black] (v) at (2.6, 0.2) {};
    \draw[fill=gray!30!white, draw=gray!80!black] (-0.7, .7) ellipse (0.2cm and 0.2cm) node[above = 0.08cm] {\scriptsize $1$};
    \node[ssssgvertex, fill=black] (mu)  at (-0.7, 0.7) {};
    \draw[fill=gray!30!white, draw=gray!80!black] (-0.7, -1.2) ellipse (0.2cm and 0.2cm) node[above = 0.08cm] {\scriptsize $2$};
    \node[ssssgvertex, fill=black] (md) at (-0.7, -1.2) {};
    \begin{scope}[xshift=-1.0cm]
      \draw[fill=gray!10!white, draw=gray!80!black, rotate around={55:(3.2, 3.0)}] (3.2, 3.0) ellipse (0.7cm and 1.2cm); \draw[fill=gray!30!white, draw=gray!80!black] (3.5, 2.8) ellipse (0.4cm and 0.4cm); \node[ssssgvertex, fill=black] (b1p) at (3.65, 2.9) {};
      \node[ssssgvertex, fill=black]  (b1) at (3.35, 2.65) {};
      \node[ssssgvertex, fill=black](c)  at (2.6, 3.4) {};
    \end{scope}
    \begin{scope}
      \draw[fill=gray!10!white, draw=gray!80!black] (-1.5, 3.2) ellipse (0.2cm and 0.2cm); \node[ssssgvertex, fill=black] (c2) at (-1.5, 3.2) {};
      \draw[fill=gray!10!white, draw=gray!80!black] (-3.9, 2.0) ellipse (0.2cm and 0.2cm); \node[ssssgvertex, fill=black] (a1)  at (-3.9, 2.0) {};
    \end{scope}
    \begin{scope}
      \draw[fill=gray!10!white, draw=gray!80!black] (3.6, -2.3) ellipse (0.2cm and 0.2cm); \node[ssssgvertex, fill=black] (b2) at (3.6, -2.3) {};
    \end{scope}
    \begin{scope}
      \draw[fill=gray!10!white, draw=gray!80!black, rotate around={0:(0.0, 0.0)}] (0.0, -3.0) ellipse (0.6cm and 0.5cm); \node[ssssgvertex, fill=black] (d1) at (0.25, -3.1) {};
      \node[ssssgvertex, fill=black] (d2) at (-0.25, -2.9) {};
    \end{scope}
    \begin{scope}
      \draw[fill=gray!10!white, draw=gray!80!black] (-3.6, -2.3) ellipse (0.2cm and 0.2cm); \node[ssssgvertex, fill=black] (a2) at (-3.6, -2.3) {};
    \end{scope}
    \draw (a1) edge[->]node[left= 0.0cm] {\scriptsize $(u^1_{\textrm{in}}, v^1_{\textrm{in}})$} (u1);
    \draw (a2) edge[->]node[left=  -0.05cm] {\scriptsize $(u^2_{\textrm{in}}, v^2_{\textrm{in}})$} (u2);
\draw (v) edge[->, bend right]node[right= 0.0cm] {\scriptsize $(u^2_{\textrm{out}}, v^2_{\textrm{out}})$} (b1p);
    \draw (v2) edge[->] node[right= 0.0cm] {\scriptsize $(u^1_{\textrm{out}}, v^1_{\textrm{out}})$}(b2);
    \draw (b1p) edge[->] (b1);
    \draw (b1) edge[->] (c);
    \draw (c) edge[->] (c2);
    \draw (c2) edge[->] (a1);
    \draw (b2) edge[->] (d1);
    \draw (d1) edge[->] (d2);
    \draw (d2) edge[->] (a2);
    \draw (u1) edge[->,decorate,decoration={snake,amplitude=.2mm,segment length=2pt,post length=1mm}] (md);
    \draw (u2) edge[->,decorate,decoration={snake,amplitude=.2mm,segment length=2pt,post length=1mm}] (mu);
    \draw (mu) edge[->,decorate,decoration={snake,amplitude=.2mm,segment length=2pt,post length=1mm}] (v1);
    \draw (v1) edge[->,decorate,decoration={snake,amplitude=.2mm,segment length=2pt,post length=1mm}] (v);
    \draw (md) edge[->,decorate,decoration={snake,amplitude=.2mm,segment length=2pt,post length=1mm}] (t);
    \draw (t) edge[->,decorate,decoration={snake,amplitude=.2mm,segment length=2pt,post length=1mm}] (v2);

\end{scope}
\end{tikzpicture}
   \caption{An example of the contraction of a tight set $S$ and the lift of a tour. Only $y$-values of the sets $R\in \cL: R\subseteq S$ are depicted. On the left, only edges that have one endpoint in $S$ are shown. These are exactly the edges that are incident to $s$ in the contracted instance. In the center, a tour of $\cI/S$ is illustrated, and on the right we depict the lift of that tour.}
  \label{fig:contraction}
\end{figure}
\else
\begin{figure}[t]
  \centering
  \hspace*{-3em}
  \begin{tikzpicture}[scale=0.95]
\tikzset{arrow data/.style 2 args={decoration={markings,
         mark=at position #1 with \arrow{#2}},
         postaction=decorate}
      }

  \begin{scope}[scale=0.6]
    \node at (0, 4.6) {\footnotesize Instance $\cI$};
    \draw[fill=gray!10!white, draw=gray!80!black] (0, 0) ellipse (4cm and 2cm) node[above = 1.20cm] {$S$}node[above left = 1.00cm and 0.9cm] {\scriptsize $5$};
    \begin{scope}
      \draw[fill=gray!30!white, draw=gray!80!black] (-2.5, 0.5) ellipse (0.75cm and 0.5cm) node[above  left= 0.20cm and -0.15cm] {\scriptsize $2$};
    \end{scope}
    \begin{scope}[xshift=1.8cm,rotate=30]
      \draw[fill=gray!30!white, draw=gray!80!black] (0, 0) ellipse (1.3cm and 1.5cm) node[above left= 0.65cm and 0.40cm] {\scriptsize $4$};
\draw[fill=gray!50!white, draw=gray!80!black] (0, -.75) ellipse (0.3cm and 0.4cm) node[above left = .12cm and 0.0cm] {\scriptsize $3$};
       \node[ssssgvertex, fill=black] (v1) at (0.1, 0.85) {};
\node[ssssgvertex, fill=black]  at (-0.0, -0.55) {};
       \node[ssssgvertex, fill=black] (v2) at (-0.0, -0.85) {};
       \node[ssssgvertex, fill=black] (t)at (-0.85, -0.15) {};
    \end{scope}
    \node[ssssgvertex, fill=black] (u1) at (-2.7, 0.4) {};
    \node[ssssgvertex, fill=black]  at (-2.2, 0.7) {};
    \node[ssssgvertex, fill=black] (u2) at (-2.5, -0.7) {};
\node[ssssgvertex, fill=black] (v) at (2.6, 0.2) {};
    \draw[fill=gray!30!white, draw=gray!80!black] (-0.7, .7) ellipse (0.2cm and 0.2cm) node[above = 0.08cm] {\scriptsize $1$};
    \node[ssssgvertex, fill=black] (mu)  at (-0.7, 0.7) {};
    \draw[fill=gray!30!white, draw=gray!80!black] (-0.7, -1.2) ellipse (0.2cm and 0.2cm) node[above = 0.08cm] {\scriptsize $2$};
    \node[ssssgvertex, fill=black] (md) at (-0.7, -1.2) {};
    \begin{scope}[xshift=-1.0cm]
      \draw[fill=gray!10!white, draw=gray!80!black, rotate around={55:(3.2, 3.0)}] (3.2, 3.0) ellipse (0.7cm and 1.2cm); \draw[fill=gray!30!white, draw=gray!80!black] (3.5, 2.8) ellipse (0.4cm and 0.4cm); \node[ssssgvertex, fill=black] (b1p) at (3.65, 2.9) {};
      \node[ssssgvertex, fill=black]  (b1) at (3.35, 2.65) {};
      \node[ssssgvertex, fill=black]  at (2.6, 3.4) {};
    \end{scope}
    \begin{scope}
      \draw[fill=gray!10!white, draw=gray!80!black] (-1.5, 3.2) ellipse (0.2cm and 0.2cm); \node[ssssgvertex, fill=black]  at (-1.5, 3.2) {};
      \draw[fill=gray!10!white, draw=gray!80!black] (-3.9, 2.0) ellipse (0.2cm and 0.2cm); \node[ssssgvertex, fill=black] (a1)  at (-3.9, 2.0) {};
    \end{scope}
    \begin{scope}
      \draw[fill=gray!10!white, draw=gray!80!black] (3.6, -2.3) ellipse (0.2cm and 0.2cm); \node[ssssgvertex, fill=black] (b2) at (3.6, -2.3) {};
    \end{scope}
    \begin{scope}
      \draw[fill=gray!10!white, draw=gray!80!black, rotate around={0:(0.0, 0.0)}] (0.0, -3.0) ellipse (0.6cm and 0.5cm); \node[ssssgvertex, fill=black]  at (0.25, -3.1) {};
      \node[ssssgvertex, fill=black]  at (-0.25, -2.9) {};
    \end{scope}
    \begin{scope}
      \draw[fill=gray!10!white, draw=gray!80!black] (-3.6, -2.3) ellipse (0.2cm and 0.2cm); \node[ssssgvertex, fill=black] (a2) at (-3.6, -2.3) {};
    \end{scope}
    \draw (a1) edge[dashed, ->] (u1);
    \draw (a2) edge[dashed, ->] (u2);
    \draw (v1) edge[dashed, ->] (b1);
    \draw (v) edge[dashed, ->, bend right] (b1p);
    \draw (v2) edge[dashed, ->] (b2);
\end{scope}

  \begin{scope}[scale=0.6, xshift=10cm]
    \node at (0, 4.6) {\footnotesize A tour of the instance $\cI/S$ };
    \node[sgvertex] (s) at (0, 0) {$s$};
    \begin{scope}[xshift=-1.0cm]
      \draw[fill=gray!10!white, draw=gray!80!black, rotate around={55:(3.2, 3.0)}] (3.2, 3.0) ellipse (0.7cm and 1.2cm); \draw[fill=gray!30!white, draw=gray!80!black] (3.5, 2.8) ellipse (0.4cm and 0.4cm); \node[ssssgvertex, fill=black] (b1p) at (3.65, 2.9) {};
      \node[ssssgvertex, fill=black]  (b1) at (3.35, 2.65) {};
      \node[ssssgvertex, fill=black](c)  at (2.6, 3.4) {};
    \end{scope}
    \begin{scope}
      \draw[fill=gray!10!white, draw=gray!80!black] (-1.5, 3.2) ellipse (0.2cm and 0.2cm); \node[ssssgvertex, fill=black] (c2) at (-1.5, 3.2) {};
      \draw[fill=gray!10!white, draw=gray!80!black] (-3.9, 2.0) ellipse (0.2cm and 0.2cm); \node[ssssgvertex, fill=black] (a1)  at (-3.9, 2.0) {};
    \end{scope}
    \begin{scope}
      \draw[fill=gray!10!white, draw=gray!80!black] (3.6, -2.3) ellipse (0.2cm and 0.2cm); \node[ssssgvertex, fill=black] (b2) at (3.6, -2.3) {};
    \end{scope}
    \begin{scope}
      \draw[fill=gray!10!white, draw=gray!80!black, rotate around={0:(0.0, 0.0)}] (0.0, -3.0) ellipse (0.6cm and 0.5cm); \node[ssssgvertex, fill=black] (d1) at (0.25, -3.1) {};
      \node[ssssgvertex, fill=black] (d2) at (-0.25, -2.9) {};
    \end{scope}
    \begin{scope}
      \draw[fill=gray!10!white, draw=gray!80!black] (-3.6, -2.3) ellipse (0.2cm and 0.2cm); \node[ssssgvertex, fill=black] (a2) at (-3.6, -2.3) {};
    \end{scope}
    \draw (a1) edge[->] node[below left = -0.1cm and -0.1cm] {\scriptsize $(u^1_{\textrm{in}}, s)$}(s);
    \draw (a2) edge[->] node[above left= -0.1cm and -0.1cm] {\scriptsize $(u^2_{\textrm{in}}, s)$} (s);
\draw (s) edge[->, bend right=60] node[below right= -0.1cm and -0.1cm] {\scriptsize $(s, v^2_{\textrm{out}})$} (b1p);
    \draw (s) edge[->] node[above right= -0.1cm and -0.1cm] {\scriptsize $(s, v^1_{\textrm{out}})$}(b2);
    \draw (b1p) edge[->] (b1);
    \draw (b1) edge[->] (c);
    \draw (c) edge[->] (c2);
    \draw (c2) edge[->] (a1);
    \draw (b2) edge[->] (d1);
    \draw (d1) edge[->] (d2);
    \draw (d2) edge[->] (a2);

\end{scope}

  \begin{scope}[scale=0.6, xshift=20cm]
    \node at (0, 4.6) {\footnotesize The lift of the tour to a subtour of $\cI$ };
    \draw[fill=gray!10!white, draw=gray!80!black] (0, 0) ellipse (4cm and 2cm) node[above = 1.20cm] {$S$} node[above left = 1.00cm and 0.9cm] {\scriptsize $5$};
    \begin{scope}
      \draw[fill=gray!30!white, draw=gray!80!black] (-2.5, 0.5) ellipse (0.75cm and 0.5cm) node[above  left= 0.20cm and -0.15cm] {\scriptsize $2$};
    \end{scope}
    \begin{scope}[xshift=1.8cm,rotate=30]
      \draw[fill=gray!30!white, draw=gray!80!black] (0, 0) ellipse (1.3cm and 1.5cm) node[above left= 0.65cm and 0.40cm] {\scriptsize $4$};
\draw[fill=gray!50!white, draw=gray!80!black] (0, -.75) ellipse (0.3cm and 0.4cm) node[above left = .12cm and 0.0cm] {\scriptsize $3$};
       \node[ssssgvertex, fill=black] (v1) at (0.1, 0.85) {};
\node[ssssgvertex, fill=black]  at (-0.0, -0.55) {};
       \node[ssssgvertex, fill=black] (v2) at (-0.0, -0.85) {};
       \node[ssssgvertex, fill=black] (t)at (-0.85, -0.15) {};
    \end{scope}
    \node[ssssgvertex, fill=black] (u1) at (-2.7, 0.4) {};
    \node[ssssgvertex, fill=black]  at (-2.2, 0.7) {};
    \node[ssssgvertex, fill=black] (u2) at (-2.5, -0.7) {};
\node[ssssgvertex, fill=black] (v) at (2.6, 0.2) {};
    \draw[fill=gray!30!white, draw=gray!80!black] (-0.7, .7) ellipse (0.2cm and 0.2cm) node[above = 0.08cm] {\scriptsize $1$};
    \node[ssssgvertex, fill=black] (mu)  at (-0.7, 0.7) {};
    \draw[fill=gray!30!white, draw=gray!80!black] (-0.7, -1.2) ellipse (0.2cm and 0.2cm) node[above = 0.08cm] {\scriptsize $2$};
    \node[ssssgvertex, fill=black] (md) at (-0.7, -1.2) {};
    \begin{scope}[xshift=-1.0cm]
      \draw[fill=gray!10!white, draw=gray!80!black, rotate around={55:(3.2, 3.0)}] (3.2, 3.0) ellipse (0.7cm and 1.2cm); \draw[fill=gray!30!white, draw=gray!80!black] (3.5, 2.8) ellipse (0.4cm and 0.4cm); \node[ssssgvertex, fill=black] (b1p) at (3.65, 2.9) {};
      \node[ssssgvertex, fill=black]  (b1) at (3.35, 2.65) {};
      \node[ssssgvertex, fill=black](c)  at (2.6, 3.4) {};
    \end{scope}
    \begin{scope}
      \draw[fill=gray!10!white, draw=gray!80!black] (-1.5, 3.2) ellipse (0.2cm and 0.2cm); \node[ssssgvertex, fill=black] (c2) at (-1.5, 3.2) {};
      \draw[fill=gray!10!white, draw=gray!80!black] (-3.9, 2.0) ellipse (0.2cm and 0.2cm); \node[ssssgvertex, fill=black] (a1)  at (-3.9, 2.0) {};
    \end{scope}
    \begin{scope}
      \draw[fill=gray!10!white, draw=gray!80!black] (3.6, -2.3) ellipse (0.2cm and 0.2cm); \node[ssssgvertex, fill=black] (b2) at (3.6, -2.3) {};
    \end{scope}
    \begin{scope}
      \draw[fill=gray!10!white, draw=gray!80!black, rotate around={0:(0.0, 0.0)}] (0.0, -3.0) ellipse (0.6cm and 0.5cm); \node[ssssgvertex, fill=black] (d1) at (0.25, -3.1) {};
      \node[ssssgvertex, fill=black] (d2) at (-0.25, -2.9) {};
    \end{scope}
    \begin{scope}
      \draw[fill=gray!10!white, draw=gray!80!black] (-3.6, -2.3) ellipse (0.2cm and 0.2cm); \node[ssssgvertex, fill=black] (a2) at (-3.6, -2.3) {};
    \end{scope}
    \draw (a1) edge[->]node[left= 0.0cm] {\scriptsize $(u^1_{\textrm{in}}, v^1_{\textrm{in}})$} (u1);
    \draw (a2) edge[->]node[left=  -0.05cm] {\scriptsize $(u^2_{\textrm{in}}, v^2_{\textrm{in}})$} (u2);
\draw (v) edge[->, bend right]node[right= 0.0cm] {\scriptsize $(u^2_{\textrm{out}}, v^2_{\textrm{out}})$} (b1p);
    \draw (v2) edge[->] node[right= 0.0cm] {\scriptsize $(u^1_{\textrm{out}}, v^1_{\textrm{out}})$}(b2);
    \draw (b1p) edge[->] (b1);
    \draw (b1) edge[->] (c);
    \draw (c) edge[->] (c2);
    \draw (c2) edge[->] (a1);
    \draw (b2) edge[->] (d1);
    \draw (d1) edge[->] (d2);
    \draw (d2) edge[->] (a2);
    \draw (u1) edge[->,decorate,decoration={snake,amplitude=.2mm,segment length=2pt,post length=1mm}] (md);
    \draw (u2) edge[->,decorate,decoration={snake,amplitude=.2mm,segment length=2pt,post length=1mm}] (mu);
    \draw (mu) edge[->,decorate,decoration={snake,amplitude=.2mm,segment length=2pt,post length=1mm}] (v1);
    \draw (v1) edge[->,decorate,decoration={snake,amplitude=.2mm,segment length=2pt,post length=1mm}] (v);
    \draw (md) edge[->,decorate,decoration={snake,amplitude=.2mm,segment length=2pt,post length=1mm}] (t);
    \draw (t) edge[->,decorate,decoration={snake,amplitude=.2mm,segment length=2pt,post length=1mm}] (v2);

\end{scope}
\end{tikzpicture}
   \caption{An example of the contraction of a tight set $S$ and the lift of a tour. Only $y$-values of the sets $R\in \cL: R\subseteq S$ are depicted. On the left, only edges that have one endpoint in $S$ are shown. These are exactly the edges that are incident to $s$ in the contracted instance. In the center, a tour of $\cI/S$ is illustrated, and on the right we depict the lift of that tour.}
  \label{fig:contraction}
\end{figure}
\fi
  The intuition of the definition of $D_{S}$ is as follows. After contracting
  $S$, all sets of the laminar family are still present in the contracted
  instance, except for the sets strictly contained in $S$. Now, after finding
  a tour in the contracted instance, we need to  lift it back to a subtour in the original instance. This is done as depicted in
  Figure~\ref{fig:contraction}: for each visit of the tour to $s$ (the vertex
  corresponding to the contraction of $S$) on the edges $(u^i_{\textrm{in}},
  s), (s, v^i_{\textrm{out}})$, we obtain a subtour of the original 
  instance by replacing $(u^i_{\textrm{in}},
  s), (s, v^i_{\textrm{out}})$  by the corresponding edges (i.e., by their
  preimages) $(u^i_{\textrm{in}}, v^i_{\textrm{in}}), (u^i_{\textrm{out}},
  v^i_{\textrm{out}})$ of $G$ together with the minimum-weight path inside $S$  from
  $v^i_{\textrm{in}}$ to $u^i_{\textrm{out}}$. The value $D_S(v^i_{\textrm{in}},
  u^i_{\textrm{out}})$ is defined to capture the weight increase incurred by this operation.  For example, in Figure~\ref{fig:contraction} we have
  \begin{align*}
 \underbrace{D_S(v^1_{\textrm{in}}, u^1_{\textrm{out}}) }_{= 22} =   \underbrace{\sum_{R\in \cL: \ v^1_{\textrm{in}}\in R \subsetneq S} y_R}_{=2} + \underbrace{d_S(v^1_{\textrm{in}}, u^1_{\textrm{out}})}_{= 2 + 2\cdot 2+ 4 + 3} + \underbrace{\sum_{R\in \cL: \ u^1_{\textrm{out}}\in R \subsetneq S} y_R}_{=3+4}\,. 
  \end{align*}

  Before formally defining the notions of contraction and lift, we state the following useful bound on $D_S(u,v)$.
  \begin{fact}
  	\label{fact:bound_on_dU}
  	For any $u, v \in S$
  	with $u \in S_\textrm{in}$ or $v \in S_\textrm{out}$ we have
  	\[ D_S(u,v) \le \valu(S). \]
  \end{fact}
  \begin{proof}
 Lemma~\ref{lem:path_tight_set}\ref{item:exists_path}  says that  there is a path from $u$ to $v$ inside $S$.
  	Select $P$ to be the path from $u$ to $v$ as guaranteed by Lemma~\ref{lem:short_path}.
  	Since $u \in S_\textrm{in}$ or $v \in S_\textrm{out}$, we have
  	\[ d_S(u,v) \le \cost(P) \le \sum_{R \in \cL  : \ R \subsetneq S} \left( 2 - |R \cap \{u,v\}| \right) \cdot y_R \]
  	and thus
  	\begin{align*}
  	D_{S}(u,v)
  	&= d_S(u,v) + \sum_{R\in \cL: R \subsetneq S} |R \cap \{u,v\}| \cdot y_R \le \sum_{R \in \cL : \ R \subsetneq S} 2 \cdot y_R = \valu(S) \,.
  	\end{align*}
  \end{proof}
  
  We now define the notion of contracting a tight set for an ATSP instance. In short, the \emph{contraction} is the instance obtained by performing the classic graph contraction of $S$, modifying $\cL$ to remove the sets contained in $S$, and increasing the $y$-value of the new singleton $\{s\}$ corresponding to $S$ so as to become $y_S + \nicefrac{1}{2} \max_{u\in S_{\textrm{in}}, v\in S_{\textrm{out}}}D_S(u,v)$. This increase is done in order to pay for the maximum possible weight increase incurred when lifting a tour in the contraction back to a subtour in the original instance (as depicted in Figure~\ref{fig:contraction}, defined in Definition~\ref{def:lift}, and analyzed in Lemma~\ref{lem:lifting}).

  \begin{definition}[Contracting  a tight set]
    \label{def:contraction}
    The instance $(G',  \cL', x', y')$ obtained from $\mathcal{I} = (G,  \cL,\xs, y)$ by \emph{contracting} $S\in \cL$, denoted by
    $\mathcal{I}/S$, is defined as follows:
    \begin{itemize}
    \item The graph $G'$ equals $G/S$, i.e.,  the graph obtained from $G$ by contracting $S$. Let $s$ denote the new vertex of $G'$ that corresponds to the set $S$. 
    \item For each edge $e'\in E(G')$, $x'(e')$ equals $\xs(e)$, where $e\in E(G)$ is the preimage of $e'$ in $G$.\footnote{Recall that for notational convenience we allow parallel edges in $G/S$ and therefore the preimage is uniquely defined.}
    \item The laminar family $\cL'$ contains all remaining sets of $\cL$:
      \begin{align*}
        \cL'=  \{(R\setminus S) \cup \{s\} : R\in \cL, S \subseteq R\} \cup \{R: R\in \cL, S\cap R = \emptyset\}\,.
      \end{align*}
    \item The vector $y'$ equals  $y$ (via the natural mapping) on all sets but $\{s\}$.  For $\{s\}$ we define
      \begin{align*}
        y'_s  = y_S +\frac{1}{2}\max_{u\in S_{\textrm{in}}, v\in S_{\textrm{out}}} D_{S}(u,v)\,.
      \end{align*}      
    \end{itemize}
  \end{definition}
  We remark that $\cI/S$ as defined above is indeed an instance: $\cL'$
  is a laminar family of tight sets (since for each $R \in \cL'$ we have $x'(\delta(R')) = x(\delta(R))$, where $R$ is the preimage of $R'$ in $\cL$ via the natural mapping), $y'_R \ge 0$ is defined only for $R \in
  \cL'$, and $x'$ is a feasible solution to $\LP(G',\mathbf{0})$ that is strictly positive on all edges.

  The way we defined the  new dual weight $y'_{s}$ implies the natural property that the value of the linear programming solution does not increase after contracting a tight set:
  \begin{fact}
    \label{fact:no_increase}

     $ \valu(\mathcal{I}/S)  =   \valu(\cI) - \left(\valu_{\cI}(S)-\max_{u\in S_{\textrm{in}}, v\in S_{\textrm{out}}} D_{S}(u,v)\right) \leq \valu(\cI)$. \end{fact}
  \begin{proof}
    By definition,
	  \begin{align*}
      \valu(\mathcal{I}/S) = 2\cdot \sum_{R \in \cL'} y'_R &= 2 \cdot y'_{s} + 2 \cdot \sum_{R \in \cL : \ R \not \subseteq S} y_R \\
        &= \max_{u\in S_{\textrm{in}}, v\in S_{\textrm{out}}} D_{S}(u,v) + 2 \cdot y_S + 2 \cdot \sum_{R \in \cL : \ R \not \subseteq S} y_R \\
	    &= \max_{u\in S_{\textrm{in}}, v\in S_{\textrm{out}}} D_{S}(u,v) + 2 \cdot \sum_{R \in \cL} y_R - 2 \cdot \sum_{R \in \cL : \ R \subsetneq S} y_R \\ 
	    &= \max_{u\in S_{\textrm{in}}, v\in S_{\textrm{out}}} D_{S}(u,v) + \valu(\cI) - \valu_\cI(S) 
    \end{align*}
    and so the equality of the statement holds. Finally, the inequality of the
    statement follows from Fact~\ref{fact:bound_on_dU}, which implies that
    $\max_{u\in S_{\textrm{in}}, v\in S_{\textrm{out}}} D_{S}(u,v) \leq
    \valu_{\cI}(S)$.
  \end{proof}
  Having defined the contraction of a tight set $S\in \cL$, we define the aforementioned operation of lifting  a tour of the contracted instance $\cI/S$ to a subtour in the original instance $\cI$.  
  When considering a tour (or a subtour), we order the edges according to an arbitrary but fixed Eulerian walk. This allows us to talk about consecutive edges. 
  \begin{definition}
    For a tour $T$ of $\mathcal{I}/S$, we define its \emph{lift}  to be the subtour  of $\cI$ obtained from
    $ T$ by replacing each consecutive pair $(u_{\textrm{in}},s),
    (s,v_{\textrm{out}})$ of incoming and outgoing edges incident to
    $s$ by their preimages $(u_{\textrm{in}}, v_{\textrm{in}}) $ and
    $(u_{\textrm{out}},v_{\textrm{out}})$ in $G$, together with  a
    minimum-weight path from $v_{\textrm{in}}$ to $u_{\textrm{out}}$ inside $S$.\footnote{
    	We remark that it is not crucial that the minimum-weight
    	path from $v_{\textrm{in}}$ to $u_{\textrm{out}}$ is selected to be
    	inside $S$; a minimum-weight path without this restriction would also work. We have chosen this definition as we find it more intuitive and it simplifies some arguments.}
    \label{def:lift}
  \end{definition}
  
  See Figure~\ref{fig:contraction} for an illustration.
  It follows that the lift is a subtour (i.e., an Eulerian multiset of edges that forms a single component), because we added paths between consecutive edges in the tour of $\cI/S$. However, the lift  is usually not  a tour of the instance $\cI$, as  it is not guaranteed to visit all the vertices in $S$. To extend the lift to a tour, we use the concept of inducing on the tight set $S$, which we introduce in Section~\ref{sec:induce}. 
  
  We complete this section  by bounding the weight of the lift of~$T$.    
  \begin{lemma}
      Let $T$ be a tour of the instance $\mathcal{I}/S$. Then the lift $F$ of $T$ satisfies 
$ \cost_{\mathcal{I}}(F) \leq \cost_{\mathcal{I}/S}(T).$
\label{lem:lifting}
  \end{lemma}
  \begin{proof}
    Consider the tour $T$ and let $(u^{(1)}_{\textrm{in}},s),
    (s,v^{(1)}_{\textrm{out}}), \ldots, (u^{(k)}_{\textrm{in}},s),
    (s,v^{(k)}_{\textrm{out}})$ be the edges that $T$ uses to visit the vertex $s$
    (which corresponds to the contracted set $S$). That is,
    $(u^{(i)}_{\textrm{in}},s)$ and $(s,v^{(i)}_{\textrm{out}})$ are the
    incoming and outgoing edge of the $i$-th visit of $T$ to $s$.  By the
    definition of contraction, we can write the weight of $T$ as
    \begin{align*}
      \cost_{\mathcal{I}/S} (T) &= \sum_{R \in \cL: \ R \not\subseteq S} \alpha_R y_R + 2 k \cdot y'_s \\
                                &= \sum_{R \in \cL: \ R \not\subseteq S} \alpha_R y_R + k \cdot \left(2\cdot y_S +  \max_{u\in S_{\textrm{in}}, v\in S_{\textrm{out}}} D_{S}(u,v)\right)\,,
    \end{align*}
    where $\alpha_R  = |\delta(R) \cap  T|$. We now compare this weight to that of the lift $F$. Let $(u^{(i)}_{\textrm{in}},v^{(i)}_{\textrm{in}})$ and $(u^{(i)}_{\textrm{out}},v^{(i)}_{\textrm{out}})$  be the edges of $G$ that are the preimages of $(u^{(i)}_{\textrm{in}},s)$ and $(s,v^{(i)}_{\textrm{out}})$.  The lift $F$ is obtained from $T$ by replacing $(u^{(i)}_{\textrm{in}},s), (s,v^{(i)}_{\textrm{out}})$ by $(u^{(i)}_{\textrm{in}},v^{(i)}_{\textrm{in}}), (u^{(i)}_{\textrm{out}},v^{(i)}_{\textrm{out}})$ and
    adding   a minimum-weight path inside $S$ from $v^{(i)}_{\textrm{in}}$ to $u^{(i)}_{\textrm{out}}$.  So $F$ crosses every $R\in \cL: R \not \subseteq S$ the same number of times $\alpha_R$ as $T$. To bound the weight incurred by crossing the tight sets ``inside'' $S$, note that the $i$-th visit to the set $S$ incurs a weight from crossing sets $R\in \cL: R \subseteq S$ that equals 
    \begin{align*}
      2y_S +\sum_{R\in \cL: \ v^{(i)}_{\textrm{in}}\in R \subsetneq S} y_R  + d_S(v^{(i)}_{\textrm{in}},u^{(i)}_{\textrm{out}}) + \sum_{R\in \cL: \ u^{(i)}_{\textrm{out}}\in R \subsetneq S} y_R  = 2y_S + D_S(v^{(i)}_{\textrm{in}}, u^{(i)}_{\textrm{out}})\,.
    \end{align*}
    Hence
    \begin{align*}
      \cost_{\mathcal{I}}(F) & = \sum_{R \in \cL: \ R \not\subseteq S} \alpha_R y_R  + \sum_{i=1}^k \left( 2\cdot y_S + D_S(v^{(i)}_{\textrm{in}}, u^{(i)}_{\textrm{out}}) \right)  \\
      & \leq \sum_{R \in \cL: \ R \not\subseteq S} \alpha_R y_R  + \sum_{i=1}^k \left( 2\cdot y_S + \max_{u\in S_{\textrm{in}}, v\in S_{\textrm{out}}} D_S(u, v) \right) \\
      & = \cost_{\mathcal{I}/S} (T)\,.
    \end{align*}
  \end{proof}

\subsection{Inducing on a Tight Set}
\label{sec:induce}
In this section we introduce our notion of induced instances. This concept  will be used for completing a lift of a tour of a contracted instance into a tour of the original instance (see Definition~\ref{def:contractible} of ``contractible'' below).
Inducing on a tight set $S$
is similar to contracting its complement $V \setminus S$
into a single vertex $\bar s$ (see Definition~\ref{def:contraction}),
though the resulting laminar family
and dual values
are somewhat different:
namely,
we let $y'_{\bar s} = \valu(S) / 2$ and
we remove $S$ (as well as all supersets of $S$) from $\cL'$.  The intuitive reason for the setting of $y'_{\bar s}$ is that each visit to $\bar s$ should pay for the  most expensive shortest paths in the strongly connected components of $S$  (see Figure~\ref{fig:induce} and the proof of Lemma~\ref{lem:get_contractible}).

We remark that the notion of inducing on $S$ for ATSP instances differs compared to the graph obtained by inducing on  $S$ (in the usual graph-theoretic sense), as here we also
 have the vertex $\bar s$ corresponding to the contraction of the
vertices not in $S$. This is needed to make sure that we
obtain an ATSP instance (in particular, that we obtain a feasible solution
$x'$ to the linear programming relaxation).
\begin{definition}
  \label{def:induced_instance}
  The instance $(G', \cL', x', y')$ obtained from $\cI = (G, \cL,\xs, y)$ by \emph{inducing} on
  a tight set $S \in \cL$, denoted by $\cI[S]$, is defined as follows:
  \begin{itemize}
    \item The graph $G'$ equals $G/\bar S$, i.e.,  the graph obtained from $G$ by contracting $\bar S = V\setminus S$. Let $\bar s$ denote the new vertex of $G'$ that corresponds to the set $\bar S$. 
    \item For each edge $e'\in E(G')$, $x'(e')$ equals $\xs(e)$, where $e\in E(G)$ is the preimage of $e'$ in $G$.\footnote{We again recall that parallel edges are allowed in $G/\bar S$, and thus the preimage of an edge is uniquely defined.}
    \item The laminar family $\cL'$ contains $\{\bar s\}$ and all sets that are strict subsets of $S$: \begin{align*}
        \cL' =  \{ R \in \cL : R \subsetneq S \} \cup \{\{\bar s \}\} \,.
      \end{align*}
    \item The vector $y'$ equals $y$ on the sets common to $\cL'$ and $\cL$. For the new set $\{\bar s\}$ we define  $y'_{\bar s} = \valu(S) / 2$.
  \end{itemize}
\end{definition}
  We remark that $\cI[S]$ in an instance: $\cL'$
  is a laminar family of tight sets, $y'_R \ge 0$ is defined only for $R \in
  \cL'$, and $x'$ is a feasible solution to $\LP(G',\mathbf{0})$ that is strictly positive on all edges.

As for the value of $\cI[S]$, it is comprised of the $y$-values of sets strictly inside $S$,
which contribute $\valu(S)$,
and that of $\{\bar s\}$,
which also contributes $2 y'_{\bar s} = \valu(S)$.
Thus we have
\begin{fact}
  \label{fact:value_of_induced_instance} 
  $\valu(\cI[S]) = 2 \valu(S) .$
\end{fact}

As alluded to above, we will use the instance $\cI[S]$ to find a  collection $F$ of subtours in the original instance $\cI$ such that $F$ plus a lift of a tour in $\cI/S$ form a tour of the instance $\cI$.  We say that such a set $F$ makes $S$ \emph{contractible}:

\begin{definition}
  We say that $S \in \cL$ is \emph{contractible} with respect to a collection  $F \subseteq E$ of subtours (i.e., $F$ is an Eulerian multiset of edges) if the lift of any tour of $\mathcal{I}/S$ plus the edge set $F$ 
  is a tour of $\mathcal{I}$. 
  \label{def:contractible}
\end{definition}
\ifdefined\ACM
\begin{figure}[t]
  \centering
  \begin{tikzpicture}[scale=0.9]
\tikzset{arrow data/.style 2 args={decoration={markings,
         mark=at position #1 with \arrow{#2}},
         postaction=decorate}
      }

  \begin{scope}[scale=0.6]
    \node at (0, 3.5) {\small Tight set $S = S_1 \cup S_2$};

    \draw[fill=gray!10!white, draw=gray!80!black] (0, 0) ellipse (4cm and 2cm) node[below left=1.4cm] {\scriptsize 1000};
    \node at (-1.7, 2.3) {$S_1$};
    \node at (1.8, 2.3) {$S_2$};

	\draw[fill=gray!30!white, draw=gray!80!black] (-2.5, 0.5) ellipse (0.25cm and 0.25cm) node[above = 0.1cm] {\scriptsize 5};
    \node[sssagvertex, fill=black] (u1) at (-2.5, 0.5) {};

	\draw[fill=gray!30!white, draw=gray!80!black] (-2.5, -0.7) ellipse (0.25cm and 0.25cm) node[below = 0.1cm] {\scriptsize 6};
    \node[sssagvertex, fill=black] (u2) at (-2.5, -0.7) {};

	\draw[fill=gray!30!white, draw=gray!80!black] (-0.6, -0.8) ellipse (0.25cm and 0.25cm) node[below = 0.1cm] {\scriptsize 3};
    \node[sssagvertex, fill=black] (u4) at (-0.6, -0.8) {};
    
    \node[sssagvertex, fill=black] (u3) at (-0.6, 0.8) {};
    
	\draw[fill=gray!30!white, draw=gray!80!black] (0.5, 0.8) ellipse (0.25cm and 0.25cm) node[below = 0.1cm] {\scriptsize 2};
    \node[sssagvertex, fill=black] (u5) at (0.5, 0.8) {};
    
    \node[sssagvertex, fill=black] (u6) at (0.5, -0.8) {};
    
	\begin{scope}[xshift=2.0cm,rotate=30]
      \draw[fill=gray!30!white, draw=gray!80!black] (0, 0) ellipse (1cm and 1.5cm) node[below = 0.7cm] {\scriptsize 9};
      \draw[fill=gray!50!white, draw=gray!80!black, rotate=10] (0, .75) ellipse (0.5cm and 0.5cm) node[right = 0.2cm] {\scriptsize 7};
      \draw[fill=gray!50!white, draw=gray!80!black,rotate=-5] (0, -.75) ellipse (0.5cm and 0.6cm) node[above right = 0.12cm] {\scriptsize 4};
       \node[sssagvertex, fill=black] (v1) at (-0.1, 0.85) {};
       \node[sssagvertex, fill=black] (v2) at (0, -0.75) {};
    \end{scope}
    
    \node[sssagvertex, fill=black] (v) at (3.3, 0.4) {};

	\draw (u1) edge[bend left=10,->] (u3);
    \draw (u3) edge[->] (u4);
    \draw (u4) edge[->,bend left=50] (u2);
    \draw (u4) edge[->] (u1);
    \draw (u2) edge[->] (u4);
    \draw (u3) edge[->] (u5);
    \draw (u4) edge[->] (u6);
    \draw (u5) edge[->] (v1);
    \draw (u6) edge[->] (v2);
    \draw (v2) edge[bend right=35,->] (v);
    \draw (v1) edge[->] (u6);
    \draw (v)  edge[bend right=70,->] (u5);

    \begin{scope}
      \clip (0,0) ellipse (4cm and 2cm); 
      \draw[thick] (0.05,-3) -- (0.05,3);  
    \end{scope}
    \draw ($(u1)- (1.7, -1.3)$) edge[densely dotted,->]  (u1);
    \draw ($(u2)- (1.7, 1.3)$) edge[densely dotted,->] (u2);
    \draw (v1) edge[densely dotted,->] ($(v1) + (2.2, 1.5)$);
\draw  (v) edge[densely dotted, bend left = 15, ->] ($(v) + (0.7, -2)$);
  \end{scope}

\begin{scope}[xshift=5cm, scale=0.6]
    \node at (0, 3.5) { \small A tour $F'$ of the instance $\mathcal{I}'$}; 

    \node[sgvertex] (bs) at (0, 2.5) {$\bar s$};
    \node at (0.55, 2.75) {\scriptsize 36}; 

	\draw[fill=gray!30!white, draw=gray!80!black] (-2.5, 0.5) ellipse (0.25cm and 0.25cm) node[above = 0.1cm] {\scriptsize 5};
    \node[sssagvertex, fill=black] (u1) at (-2.5, 0.5) {};

	\draw[fill=gray!30!white, draw=gray!80!black] (-2.5, -0.7) ellipse (0.25cm and 0.25cm) node[below = 0.1cm] {\scriptsize 6};
    \node[sssagvertex, fill=black] (u2) at (-2.5, -0.7) {};

	\draw[fill=gray!30!white, draw=gray!80!black] (-0.6, -0.8) ellipse (0.25cm and 0.25cm) node[below = 0.1cm] {\scriptsize 3};
    \node[sssagvertex, fill=black] (u4) at (-0.6, -0.8) {};
    
    \node[sssagvertex, fill=black] (u3) at (-0.6, 0.8) {};
    
	\draw[fill=gray!30!white, draw=gray!80!black] (0.5, 0.8) ellipse (0.25cm and 0.25cm) node[below = 0.1cm] {\scriptsize 2};
    \node[sssagvertex, fill=black] (u5) at (0.5, 0.8) {};
    
    \node[sssagvertex, fill=black] (u6) at (0.5, -0.8) {};
    
	\begin{scope}[xshift=2.0cm,rotate=30]
      \draw[fill=gray!30!white, draw=gray!80!black] (0, 0) ellipse (1cm and 1.5cm) node[below = 0.7cm] {\scriptsize 9};
      \draw[fill=gray!50!white, draw=gray!80!black, rotate=10] (0, .75) ellipse (0.5cm and 0.5cm) node[right = 0.2cm] {\scriptsize 7};
      \draw[fill=gray!50!white, draw=gray!80!black,rotate=-5] (0, -.75) ellipse (0.5cm and 0.6cm) node[above right = 0.12cm] {\scriptsize 4};
       \node[sssagvertex, fill=black] (v1) at (-0.1, 0.85) {};
       \node[sssagvertex, fill=black] (v2) at (0, -0.75) {};
    \end{scope}
    
    \node[sssagvertex, fill=black] (v) at (3.3, 0.4) {};

	\draw (u1) edge[bend left=10,thick,->] (u3);
    \draw (u3) edge[->,line width=0.15] (u4);
    \draw (u4) edge[->,line width=0.15,bend left=50] (u2);
    \draw (u4) edge[->,line width=0.15] (u1);
    \draw (u2) edge[thick,->] (u4);
    \draw (u3) edge[thick,->] node[above] {\scriptsize $e_2$} (u5);
    \draw (u4) edge[thick,->] node[above] {\scriptsize $e_5$} (u6);
    \draw (u5) edge[thick,->] (v1);
    \draw (u6) edge[thick,->] (v2);
    \draw (v2) edge[bend right=35,thick,->] (v);
    \draw (v1) edge[->,line width=0.15,] (u6);
    \draw (v)  edge[bend right=70,line width=0.15,->] (u5);

    \draw (bs) edge[thick,->,bend right=10] node[above] {\scriptsize $e_1$} (u1) ;
    \draw (bs) edge[thick,->,bend right=85] node[above = 0.07cm] {\scriptsize $e_4$} (u2);
    \draw (v1) edge[thick,->,bend left=10] node[above right = -0.16cm] {\scriptsize $e_3$} (bs);
\draw  (v) edge[thick,->,bend right=20] node[above] {\scriptsize $e_6$} (bs);
  \end{scope}

\begin{scope}[xshift=10cm,scale=0.6]
    \node at (0, 3.5) { \small Collection $F$ of subtours}; 

	\draw[fill=gray!10!white, draw=gray!80!black] (0, 0) ellipse (4cm and 2cm);
    \node at (-1.7, 2.3) {$S_1$};
    \node at (1.8, 2.3) {$S_2$};

	\draw[fill=gray!30!white, draw=gray!80!black] (-2.5, 0.5) ellipse (0.25cm and 0.25cm) node[above = 0.1cm] {\scriptsize 5};
    \node[sssagvertex, fill=black] (u1) at (-2.5, 0.5) {};

	\draw[fill=gray!30!white, draw=gray!80!black] (-2.5, -0.7) ellipse (0.25cm and 0.25cm) node[below = 0.1cm] {\scriptsize 6};
    \node[sssagvertex, fill=black] (u2) at (-2.5, -0.7) {};

	\draw[fill=gray!30!white, draw=gray!80!black] (-0.6, -0.8) ellipse (0.25cm and 0.25cm) node[below = 0.1cm] {\scriptsize 3};
    \node[sssagvertex, fill=black] (u4) at (-0.6, -0.8) {};
    
    \node[sssagvertex, fill=black] (u3) at (-0.6, 0.8) {};
    
	\draw[fill=gray!30!white, draw=gray!80!black] (0.5, 0.8) ellipse (0.25cm and 0.25cm) node[below = 0.1cm] {\scriptsize 2};
    \node[sssagvertex, fill=black] (u5) at (0.5, 0.8) {};
    
    \node[sssagvertex, fill=black] (u6) at (0.5, -0.8) {};
    
	\begin{scope}[xshift=2.0cm,rotate=30]
      \draw[fill=gray!30!white, draw=gray!80!black] (0, 0) ellipse (1cm and 1.5cm) node[below = 0.7cm] {\scriptsize 9};
      \draw[fill=gray!50!white, draw=gray!80!black, rotate=10] (0, .75) ellipse (0.5cm and 0.5cm) node[right = 0.2cm] {\scriptsize 7};
      \draw[fill=gray!50!white, draw=gray!80!black,rotate=-5] (0, -.75) ellipse (0.5cm and 0.6cm) node[above right = 0.12cm] {\scriptsize 4};
       \node[sssagvertex, fill=black] (v1) at (-0.1, 0.85) {};
       \node[sssagvertex, fill=black] (v2) at (0, -0.75) {};
    \end{scope}
    
    \node[sssagvertex, fill=black] (v) at (3.3, 0.4) {};

	\draw (u1) edge[bend left=10,thick,->] (u3);
    \draw (u3) edge[->,dashed,thick] (u4);
    \draw (u4) edge[->,dashed,thick,bend left=50] (u2);
    \draw (u4) edge[dashed,thick,->] (u1);
    \draw (u2) edge[thick,->] (u4);
\draw (u5) edge[thick,->] (v1);
    \draw (u6) edge[thick,->] (v2);
    \draw (v2) edge[bend right=35,thick,->] (v);
    \draw (v1) edge[dashed,thick,->] (u6);
    \draw (v)  edge[bend right=70,dashed,thick,->] (u5);

    \begin{scope}
      \clip (0,0) ellipse (4cm and 2cm); 
      \draw[thick] (0.05,-3) -- (0.05,3);  
    \end{scope}
\end{scope}

\end{tikzpicture}
   \caption{In the left figure we depict a tight set $S\in \cL$ with two strongly connected components $S_1$ and $S_2$. The induced instance (center figure) is obtained by contracting $\bar S = V\setminus S$ into a vertex $\bar s$ and removing the tight set $S$ from $\cL$. The solid edges are paths and edges of a tour in the induced instance. In Lemma~\ref{lem:get_contractible} we obtain a collection of subtours in the original instance (right figure) by adding the dashed paths, resulting in  a tour of each strongly connected component.}
  \label{fig:induce}
\end{figure}
\else
\begin{figure}[t]
  \centering
  \begin{tikzpicture}
\tikzset{arrow data/.style 2 args={decoration={markings,
         mark=at position #1 with \arrow{#2}},
         postaction=decorate}
      }

  \begin{scope}[scale=0.6]
    \node at (0, 3.5) {\small Tight set $S = S_1 \cup S_2$};

    \draw[fill=gray!10!white, draw=gray!80!black] (0, 0) ellipse (4cm and 2cm) node[below left=1.4cm] {\scriptsize 1000};
    \node at (-1.7, 2.3) {$S_1$};
    \node at (1.8, 2.3) {$S_2$};

	\draw[fill=gray!30!white, draw=gray!80!black] (-2.5, 0.5) ellipse (0.25cm and 0.25cm) node[above = 0.1cm] {\scriptsize 5};
    \node[sssagvertex, fill=black] (u1) at (-2.5, 0.5) {};

	\draw[fill=gray!30!white, draw=gray!80!black] (-2.5, -0.7) ellipse (0.25cm and 0.25cm) node[below = 0.1cm] {\scriptsize 6};
    \node[sssagvertex, fill=black] (u2) at (-2.5, -0.7) {};

	\draw[fill=gray!30!white, draw=gray!80!black] (-0.6, -0.8) ellipse (0.25cm and 0.25cm) node[below = 0.1cm] {\scriptsize 3};
    \node[sssagvertex, fill=black] (u4) at (-0.6, -0.8) {};
    
    \node[sssagvertex, fill=black] (u3) at (-0.6, 0.8) {};
    
	\draw[fill=gray!30!white, draw=gray!80!black] (0.5, 0.8) ellipse (0.25cm and 0.25cm) node[below = 0.1cm] {\scriptsize 2};
    \node[sssagvertex, fill=black] (u5) at (0.5, 0.8) {};
    
    \node[sssagvertex, fill=black] (u6) at (0.5, -0.8) {};
    
	\begin{scope}[xshift=2.0cm,rotate=30]
      \draw[fill=gray!30!white, draw=gray!80!black] (0, 0) ellipse (1cm and 1.5cm) node[below = 0.7cm] {\scriptsize 9};
      \draw[fill=gray!50!white, draw=gray!80!black, rotate=10] (0, .75) ellipse (0.5cm and 0.5cm) node[right = 0.2cm] {\scriptsize 7};
      \draw[fill=gray!50!white, draw=gray!80!black,rotate=-5] (0, -.75) ellipse (0.5cm and 0.6cm) node[above right = 0.12cm] {\scriptsize 4};
       \node[sssagvertex, fill=black] (v1) at (-0.1, 0.85) {};
       \node[sssagvertex, fill=black] (v2) at (0, -0.75) {};
    \end{scope}
    
    \node[sssagvertex, fill=black] (v) at (3.3, 0.4) {};

	\draw (u1) edge[bend left=10,->] (u3);
    \draw (u3) edge[->] (u4);
    \draw (u4) edge[->,bend left=50] (u2);
    \draw (u4) edge[->] (u1);
    \draw (u2) edge[->] (u4);
    \draw (u3) edge[->] (u5);
    \draw (u4) edge[->] (u6);
    \draw (u5) edge[->] (v1);
    \draw (u6) edge[->] (v2);
    \draw (v2) edge[bend right=35,->] (v);
    \draw (v1) edge[->] (u6);
    \draw (v)  edge[bend right=70,->] (u5);

    \begin{scope}
      \clip (0,0) ellipse (4cm and 2cm); 
      \draw[thick] (0.05,-3) -- (0.05,3);  
    \end{scope}
    \draw ($(u1)- (1.7, -1.3)$) edge[densely dotted,->]  (u1);
    \draw ($(u2)- (1.7, 1.3)$) edge[densely dotted,->] (u2);
    \draw (v1) edge[densely dotted,->] ($(v1) + (2.2, 1.5)$);
\draw  (v) edge[densely dotted, bend left = 15, ->] ($(v) + (0.7, -2)$);
  \end{scope}

\begin{scope}[xshift=5cm, scale=0.6]
    \node at (0, 3.5) { \small A tour $F'$ of the instance $\mathcal{I}'$}; 

    \node[sgvertex] (bs) at (0, 2.5) {$\bar s$};
    \node at (0.55, 2.75) {\scriptsize 36}; 

	\draw[fill=gray!30!white, draw=gray!80!black] (-2.5, 0.5) ellipse (0.25cm and 0.25cm) node[above = 0.1cm] {\scriptsize 5};
    \node[sssagvertex, fill=black] (u1) at (-2.5, 0.5) {};

	\draw[fill=gray!30!white, draw=gray!80!black] (-2.5, -0.7) ellipse (0.25cm and 0.25cm) node[below = 0.1cm] {\scriptsize 6};
    \node[sssagvertex, fill=black] (u2) at (-2.5, -0.7) {};

	\draw[fill=gray!30!white, draw=gray!80!black] (-0.6, -0.8) ellipse (0.25cm and 0.25cm) node[below = 0.1cm] {\scriptsize 3};
    \node[sssagvertex, fill=black] (u4) at (-0.6, -0.8) {};
    
    \node[sssagvertex, fill=black] (u3) at (-0.6, 0.8) {};
    
	\draw[fill=gray!30!white, draw=gray!80!black] (0.5, 0.8) ellipse (0.25cm and 0.25cm) node[below = 0.1cm] {\scriptsize 2};
    \node[sssagvertex, fill=black] (u5) at (0.5, 0.8) {};
    
    \node[sssagvertex, fill=black] (u6) at (0.5, -0.8) {};
    
	\begin{scope}[xshift=2.0cm,rotate=30]
      \draw[fill=gray!30!white, draw=gray!80!black] (0, 0) ellipse (1cm and 1.5cm) node[below = 0.7cm] {\scriptsize 9};
      \draw[fill=gray!50!white, draw=gray!80!black, rotate=10] (0, .75) ellipse (0.5cm and 0.5cm) node[right = 0.2cm] {\scriptsize 7};
      \draw[fill=gray!50!white, draw=gray!80!black,rotate=-5] (0, -.75) ellipse (0.5cm and 0.6cm) node[above right = 0.12cm] {\scriptsize 4};
       \node[sssagvertex, fill=black] (v1) at (-0.1, 0.85) {};
       \node[sssagvertex, fill=black] (v2) at (0, -0.75) {};
    \end{scope}
    
    \node[sssagvertex, fill=black] (v) at (3.3, 0.4) {};

	\draw (u1) edge[bend left=10,thick,->] (u3);
    \draw (u3) edge[->,line width=0.15] (u4);
    \draw (u4) edge[->,line width=0.15,bend left=50] (u2);
    \draw (u4) edge[->,line width=0.15] (u1);
    \draw (u2) edge[thick,->] (u4);
    \draw (u3) edge[thick,->] node[above] {\scriptsize $e_2$} (u5);
    \draw (u4) edge[thick,->] node[above] {\scriptsize $e_5$} (u6);
    \draw (u5) edge[thick,->] (v1);
    \draw (u6) edge[thick,->] (v2);
    \draw (v2) edge[bend right=35,thick,->] (v);
    \draw (v1) edge[->,line width=0.15,] (u6);
    \draw (v)  edge[bend right=70,line width=0.15,->] (u5);

    \draw (bs) edge[thick,->,bend right=10] node[above] {\scriptsize $e_1$} (u1) ;
    \draw (bs) edge[thick,->,bend right=85] node[above = 0.07cm] {\scriptsize $e_4$} (u2);
    \draw (v1) edge[thick,->,bend left=10] node[above right = -0.16cm] {\scriptsize $e_3$} (bs);
\draw  (v) edge[thick,->,bend right=20] node[above] {\scriptsize $e_6$} (bs);
  \end{scope}

\begin{scope}[xshift=10cm,scale=0.6]
    \node at (0, 3.5) { \small Collection $F$ of subtours}; 

	\draw[fill=gray!10!white, draw=gray!80!black] (0, 0) ellipse (4cm and 2cm);
    \node at (-1.7, 2.3) {$S_1$};
    \node at (1.8, 2.3) {$S_2$};

	\draw[fill=gray!30!white, draw=gray!80!black] (-2.5, 0.5) ellipse (0.25cm and 0.25cm) node[above = 0.1cm] {\scriptsize 5};
    \node[sssagvertex, fill=black] (u1) at (-2.5, 0.5) {};

	\draw[fill=gray!30!white, draw=gray!80!black] (-2.5, -0.7) ellipse (0.25cm and 0.25cm) node[below = 0.1cm] {\scriptsize 6};
    \node[sssagvertex, fill=black] (u2) at (-2.5, -0.7) {};

	\draw[fill=gray!30!white, draw=gray!80!black] (-0.6, -0.8) ellipse (0.25cm and 0.25cm) node[below = 0.1cm] {\scriptsize 3};
    \node[sssagvertex, fill=black] (u4) at (-0.6, -0.8) {};
    
    \node[sssagvertex, fill=black] (u3) at (-0.6, 0.8) {};
    
	\draw[fill=gray!30!white, draw=gray!80!black] (0.5, 0.8) ellipse (0.25cm and 0.25cm) node[below = 0.1cm] {\scriptsize 2};
    \node[sssagvertex, fill=black] (u5) at (0.5, 0.8) {};
    
    \node[sssagvertex, fill=black] (u6) at (0.5, -0.8) {};
    
	\begin{scope}[xshift=2.0cm,rotate=30]
      \draw[fill=gray!30!white, draw=gray!80!black] (0, 0) ellipse (1cm and 1.5cm) node[below = 0.7cm] {\scriptsize 9};
      \draw[fill=gray!50!white, draw=gray!80!black, rotate=10] (0, .75) ellipse (0.5cm and 0.5cm) node[right = 0.2cm] {\scriptsize 7};
      \draw[fill=gray!50!white, draw=gray!80!black,rotate=-5] (0, -.75) ellipse (0.5cm and 0.6cm) node[above right = 0.12cm] {\scriptsize 4};
       \node[sssagvertex, fill=black] (v1) at (-0.1, 0.85) {};
       \node[sssagvertex, fill=black] (v2) at (0, -0.75) {};
    \end{scope}
    
    \node[sssagvertex, fill=black] (v) at (3.3, 0.4) {};

	\draw (u1) edge[bend left=10,thick,->] (u3);
    \draw (u3) edge[->,dashed,thick] (u4);
    \draw (u4) edge[->,dashed,thick,bend left=50] (u2);
    \draw (u4) edge[dashed,thick,->] (u1);
    \draw (u2) edge[thick,->] (u4);
\draw (u5) edge[thick,->] (v1);
    \draw (u6) edge[thick,->] (v2);
    \draw (v2) edge[bend right=35,thick,->] (v);
    \draw (v1) edge[dashed,thick,->] (u6);
    \draw (v)  edge[bend right=70,dashed,thick,->] (u5);

    \begin{scope}
      \clip (0,0) ellipse (4cm and 2cm); 
      \draw[thick] (0.05,-3) -- (0.05,3);  
    \end{scope}
\end{scope}

\end{tikzpicture}
   \caption{In the left figure we depict a tight set $S\in \cL$ with two strongly connected components $S_1$ and $S_2$. The induced instance (center figure) is obtained by contracting $\bar S = V\setminus S$ into a vertex $\bar s$ and removing the tight set $S$ from $\cL$. The solid edges are paths and edges of a tour in the induced instance. In Lemma~\ref{lem:get_contractible} we obtain a collection of subtours in the original instance (right figure) by adding the dashed paths, resulting in  a tour of each strongly connected component.}
  \label{fig:induce}
\end{figure}
\fi

As an example, if $F$ were a subtour visiting every vertex of $S$,
then $S$ would be contractible with respect to $F$.
The following lemma shows that, in general, it is sufficient to find a tour of $\cI[S]$ in order to make $S$ contractible (see also Figure~\ref{fig:induce}). 
\begin{lemma}
  \label{lem:get_contractible}
  Given a tour $T$ of $\cI[S]$, 
  we can in polynomial time find
  a  collection $F \subseteq E$ of subtours
  such that
  $S$ is contractible with respect to $F$
  and
  $\cost_{\cI}(F) \le \cost_{\cI[S]}(T)$.
\end{lemma}

\begin{proof}
  Let $S_1, \ldots, S_\ell$ be the strongly connected components of $S$ indexed using a topological ordering.
  We will use $T$ to obtain a low-weight tour $F_i$ inside each $S_i$,
  and define $F$ to be the union of these tours.
  Then $S$ is contractible with respect to $F$.
  Indeed, the lift of any tour of $\cI / S$ must contain a path traversing $S$,
  and any such path visits every connected component by Lemma~\ref{lem:path_tight_set}\ref{item:path_through_every_scc}.
  
  Let us fix one component $S_i$.
  We obtain the tour $F_i$ of $S_i$
  by reproducing the movements of $T$ inside $S_i$
  (recall that we think of $T$ as a cyclically ordered Eulerian walk).
  More precisely,
  we retain those edges of $T$ that are inside $S_i$
  and, every time $T$ exits $S_i$ on an edge $(u_\textrm{out}, v_\textrm{out}) \in \delta^+(S_i)$
  and then returns to $S_i$ on an edge $(u_\textrm{in}, v_\textrm{in}) \in \delta^-(S_i)$,
  we also insert a minimum-weight path from $u_\textrm{out}$ to $v_\textrm{in}$
  inside $S_i$.
  Such a path exists because $S_i$ is strongly connected. (This step corresponds to adding the dashed paths in Figure~\ref{fig:induce}.)
  Then we set $F = F_1 \cup \ldots \cup F_\ell$.

  It remains to show that $F$ has low weight, i.e., that
  $\cost_{\cI}(F) \leq \cost_{\cI[S]}(T)$.
  For this,
  let $k$ be the number of times the tour $T$ visits the auxiliary vertex $\bar s$.
  The weight incurred by every such visit is at least $2 y'_{\bar s} = \valu(S)$
  (since the set $\{\bar s\}$ is crossed twice in each visit).
  Thus we have
  \[ \cost_{\cI[S]}(T) \ge k \cdot \valu(S) + \sum_{i = 1}^{\ell} \cost_{\cI[S]}(T \cap E(S_i)) = k \cdot \valu(S) + \sum_{i = 1}^{\ell} \cost_{\cI}(T \cap E(S_i))  \,. \]
  On the other hand,
  each tour $F_i$ consists of all those edges of $T$ that are inside $S_i$,
  as well as $k$ shortest paths between some pairs of vertices in $S_i$. Indeed, if $T$ makes $k$ visits to the auxiliary vertex $\bar s$, then we add exactly $k$ paths inside $S_i$ due to the path-like structure of the strongly connected components of  a tight set $S$ (Lemma~\ref{lem:path_tight_set}).

By Lemma~\ref{lem:lam-strongly}, $\cL\cup\{S_i\}$ is a laminar family. Consequently,
 Lemma~\ref{lem:short_path} is applicable to
  obtain that a shortest path between two vertices inside $S_i$ has
  weight  at most  $\valu(S_i)$. 
 Recall that $F_i$ consists of all those edges of $T$ that are inside $S_i$,
  as well as $k$ shortest paths between some pairs of vertices in $S_i$.
  Therefore, the weight of $F$ is
  \begin{align*}
      \cost_{\cI}(F) &= \sum_{i = 1}^{\ell} \cost_{\cI}(F_i) \\
  	               &\le \sum_{i = 1}^{\ell} \Brac{ k \cdot \valu(S_i) + \cost_{\cI}(T \cap E(S_i)) } \\
  	               &\le k \cdot \valu(S) + \sum_{i = 1}^{\ell} \cost_{\cI}(T \cap E(S_i)) \\
                   &\le \cost_{\cI[S]}(T) \,,
  \end{align*}
  as required.
\end{proof}
\section{Reduction to Irreducible Instances}
\label{sec:reducetoirreducible}
  In this section we reduce the problem of approximating ATSP on general (laminarly-weighted) instances
  to that of approximating ATSP on irreducible instances. Specifically,
  Theorem~\ref{thm:reduction_to_irreducible} says that any approximation
  algorithm for irreducible instances can be turned into an algorithm for
  general instances while losing only a constant factor in the approximation
  guarantee.

  We now define the notions of \emph{reducible sets} and \emph{irreducible instances}. The intuition behind them is as follows. The operations of contracting and inducing on a tight set $S$ introduced in the last section naturally lead to the following recursive algorithm:
  \begin{enumerate}
    \item Select a tight set $S\in \cL$.
    \item Find a tour $T_S$ in the induced instance $\cI[S]$. Via Lemma~\ref{lem:get_contractible}, $T_S$ yields a set $F_S$ that makes $S$ contractible.
    \item Recursively find a tour $T$ in the contraction $\cI/S$.
    \item Output $F_S$ plus the lift of $T$.
  \end{enumerate}
  For this scheme to yield a good approximation guarantee, we need to ensure that we can find a good approximate tour $T_S$ in $\cI[S]$ and that contracting the set $S$ results in  a ``significant'' decrease in the value of the LP solution.  If it does, we refer to the set $S$ as \emph{reducible}: 
  \begin{definition}
    We say that a set $S\in \cL$ is \emph{\reducible} if 
    \begin{align*}
      \max_{u\in S_{\textrm{in}}, v\in S_{\textrm{out}}} D_{S}(u,v) < \delta \cdot \valu(S)\,,
    \end{align*}
    otherwise we say that $S$ is \emph{irreducible}.
    We also say that the instance $\mathcal{I}$  is \emph{irreducible}  if no set $S\in \cL$ is \reducible. 
    \label{def:reducible}
  \end{definition}
We will use the value $\delta=\deltaval$; however, we keep it as a
parameter $\delta \in (\nicefrac12, 1)$ to exhibit the dependence of the approximation ratio on this value.

Note that singleton sets are never reducible. Moreover, we have the following observation:
\begin{fact}
  \label{fact:irreducibleinv}
  Consider an instance $\cI = (G, \cL, \xs, y)$ and a set $S\in \cL$. If every set $R\in \cL: R \subsetneq S$ is irreducible, then $\cI[S]$ is irreducible. In particular, if $\cI$ is an irreducible instance, then $\cI[S]$ is irreducible for every $S\in \cL$.
\end{fact}
\begin{proof}
  Let $\cI[S] = (G', \cL', x', y')$. By definition, $\cL' =  \{ R \in \cL : R \subsetneq S \} \cup \{\{\bar s \}\}$. Clearly, the singleton set $\{\bar s \}$ is irreducible. Now consider a set $R \in \cL : R \subsetneq S$; we need to show that $R$ is irreducible in $\cI[S]$. Note that $R$ is also present in $\cI$ and that the sets $\{Q\in \cL:  Q \subsetneq R\}$ and $\{Q\in \cL': Q\subsetneq R\}$ are identical. This implies that the distance function $D_R$  is identical in the instances $\cI$ and $\cI[S]$.  Moreover, the sets $R_{\textrm{in}}$ and $R_{\textrm{out}}$ are also the same in the two instances.  Therefore,  as $R$ is irreducible in $\cI$ by assumption, we have that $R$ is also irreducible in $\cI[S]$. 
\end{proof}

  The above fact implies that if we select $S\in \cL$ to be a \emph{minimal} reducible set, then  the instance $\cI[S]$ is irreducible. Hence, we only need to be able to find an approximate tour $T_S$  for \emph{irreducible} instances (in Step 2 of the above recursive algorithm). This is the idea behind the following theorem, and its proof is based on formally analyzing the aforementioned approach. 
  \begin{theorem}
  Let $\cA$ be a polynomial-time $\rho$-approximation algorithm (with respect to the Held--Karp lower bound) for irreducible instances. Then
  there is a polynomial-time $\frac{2\rho}{1-\delta}$-approximation algorithm (with respect to the Held--Karp lower bound) for general instances.
  \label{thm:reduction_to_irreducible}
\end{theorem}

\begin{proof}
Consider a general instance $\cI = (G,  \cL, \xs, y)$. If it is
irreducible, we can simply return the result of a single call to $\cA$. So
assume that $\cI$ is not irreducible,
i.e., that $\cL$ contains a reducible set.
Let $S\in \cL$ be a minimal
(inclusion-wise) \reducible set, i.e., one such that
all subsets $R\in \cL: R \subsetneq S$ are irreducible.

We will work with the induced instance
$\cI[S]$.
Recall that $\valu(\cI[S]) = 2 \valu(S)$ (Fact~\ref{fact:value_of_induced_instance}).
Moreover, $\cI[S]$ is irreducible by Fact~\ref{fact:irreducibleinv}. We can therefore use $\cA$ to find a tour $T_S$ of
$\cI[S]$. Since $\cA$ is a $\rho$-approximation algorithm, we have 
\begin{align*}
  \cost_{\cI[S]}(T_S) \leq \rho \cdot \valu(\cI[S]) = 2\rho \valu(S) \,. 
\end{align*}

Next, we invoke the algorithm of Lemma~\ref{lem:get_contractible} to obtain
a  collection $F_S \subseteq E$ of subtours 
such that $S$ is contractible with respect to $F_S$ and
\begin{equation}
  \label{eq:cost_F}
  \cost_{\cI}(F_S) \le \cost_{\cI[S]}(T_S) \le 2 \rho \valu(S) \,.
\end{equation}
Now we recursively solve the contraction $\cI / S$.
(This is a smaller instance than $\cI$, because $|\cI / S| = |\cI| - |S| + 1$ and $|S| \ge 2$ since a singleton set $S$ would not have been reducible.)
Let $T$ be the tour obtained from the recursive call,
and let $F$ be the lift of $T$ to $\cI$.
We finally return $F_S \cup F$.
This is a tour of $\cI$,
as $S$ is contractible with respect to $F_S$.

The running time of this algorithm is polynomial
since each recursive call consists at most of:
one call to $\cA$,
the algorithm of Lemma~\ref{lem:get_contractible},
simple graph operations
and one recursive call (for a smaller instance).

Finally, let us show
that this is a $\frac{2\rho}{1-\delta}$-approximation algorithm
by induction on the instance size.
We have
\begin{align*}
	\cost_{\cI}(F \cup F_S) &= \cost_{\cI}(F) + \cost_{\cI}(F_S) \\
	&\le \cost_{\cI / S}(T) + \cost_{\cI}(F_S) \\
	&\le \frac{2 \rho}{1 - \delta} \valu(\cI / S) + 2 \rho \valu(S) \\
	&< \frac{2 \rho}{1 - \delta} \Brac{ \valu(\cI) - (1 - \delta) \valu(S) } + 2 \rho \valu(S) \\
	&= \frac{2 \rho}{1 - \delta} \valu(\cI),
\end{align*}
where the first inequality is by Lemma~\ref{lem:lifting},
the second follows since $T$ is a $\frac{2 \rho}{1 - \delta}$-approximate solution for $\cI / S$
and by \eqref{eq:cost_F},
and the strict inequality is by Fact~\ref{fact:no_increase}
and the reducibility of $S$.
This shows that $F \cup F_S$ is a $\frac{2 \rho}{1 - \delta}$-approximate solution for $\cI$.
\end{proof}

\section{Backbones and Reduction to Vertebrate Pairs}
\label{sec:reducetovertebrate}

In this section we further reduce the task of approximating ATSP to that of finding a tour  in instances with a \emph{backbone}. For an example of such an instance see the right part of Figure~\ref{fig:intro} on page~\pageref{fig:intro}.
\begin{definition}
  We say that an instance $\cI = (G, \cL, \xs, y)$ and a subtour $B$ form a \emph{vertebrate pair} if
  every $S\in \cL$ with $|S| \geq 2$ is visited by $B$, i.e., $S \cap V(B)\neq \emptyset$.
  The set $B$ is referred to as the \emph{backbone} of the vertebrate pair.
\end{definition}

The main result of this section,
 Theorem~\ref{thm:fromSimpletoGeneral}, takes as input 
  an algorithm for
 vertebrate pairs that returns a tour with a weight bound depending on the backbone, and shows that this implies a constant-factor approximation algorithm for irreducible
 instances.
Combining this
 with Theorem~\ref{thm:reduction_to_irreducible} allows us to reduce the
 problem of approximating ATSP on general instances to that of approximating ATSP on
 vertebrate pairs.
 
 The proof of Theorem~\ref{thm:fromSimpletoGeneral} is done in two steps. First, in Section~\ref{sec:backbone}, we give an efficient algorithm for finding a \emph{quasi-backbone} $B$ of an irreducible instance -- a subtour that visits a large (weighted) fraction of the sets in $\cL$. We use the term \emph{quasi-backbone} as $B$ might not visit \emph{all} non-singleton sets, as would be required for a backbone. 
 Then, in Section~\ref{sec:simpleirr}, we give the reduction to vertebrate pairs via a recursive algorithm (similar to the proof of Theorem~\ref{thm:reduction_to_irreducible} in the previous section).

\subsection{Finding a Quasi-Backbone}
\label{sec:backbone}

\ifdefined\ACM
\begin{figure}[t]
  \centering
  \begin{tikzpicture}[scale=0.85,
    my style/.style={append after command={\pgfextra{\node [right=-0.1cm] at (\tikzlastnode.mid east) {\tiny \tikzlastnode};}
         },
       },
  ]
\tikzset{arrow data/.style 2 args={decoration={markings,
         mark=at position #1 with \arrow{#2}},
         postaction=decorate}
      }

  \begin{scope}[yscale=0.7, xscale=0.55]
    \node at (0.5, 3.6) {\begin{minipage}{4.5cm}\scriptsize The rerouting inside $S$ when obtaining the quasi-backbone $B$ from a lift $B'$ of a  tour $T$ of the instance $\cI'$ obtained by contracting maximal sets in $\cL$.\end{minipage}};
    \draw[fill=gray!10!white, draw=gray!80!black] (0.3, 0) ellipse (4.3cm and 2.5cm); \begin{scope}
      \draw[fill=gray!30!white, draw=gray!80!black] (-2.5, 0.7) ellipse (0.95cm and 0.7cm);\end{scope}
    \begin{scope}[xshift=1.0cm, yshift=0.6cm, rotate=90]
      \draw[fill=gray!30!white, draw=gray!80!black] (0, 0) ellipse (1.3cm and 2.0cm);\draw[fill=gray!50!white, draw=gray!80!black] (0, -.85) ellipse (0.6cm and 0.9cm) node[above left = .30cm and -0.2cm] {\scriptsize $R_2$};
      \draw[fill=gray!50!white, draw=gray!80!black] (0.3, 1.1) ellipse (0.5cm and 0.5cm);\node[ssssgvertex, fill=black] (a) at (0.1, 0.85) {};
       \node[ssssgvertex, fill=black] (r) at (-0.85, 0.45) {};
\draw[fill=gray!70!white, draw=gray!80!black] (0.0, -0.7) ellipse (0.4cm and 0.4cm);\node[ssssgvertex, fill=black] (b) at (-0.15, -0.55) {}; \node[ssssgvertex, fill=black] (c) at (0.15, -0.85) {};\node[ssssgvertex, fill=black]  (d) at (-0.2, -1.4) {}; \node[ssssgvertex, fill=black] (e) at (0.2, -1.4) {};\node[ssssgvertex, fill=black] (f)  at (0.4, 1.2) {};
    \end{scope}
    \node[ssssgvertex, fill=black] (g) at (-2.7, 0.4) {}; \node at (-2.85, 0.8) {\scriptsize $u^S_{\textrm{max}}$};
    \node[ssssgvertex, fill=black] (h) at (-2.2, 0.7) {}; \node[ssssgvertex, fill=black] (i) at (-2.8, -0.7) {};
    \node[ssssgvertex, fill=black]  (j) at (-2.4, -1.4)  {};
    \node at (-2.00, -1.2) {\scriptsize $u^S$};

\begin{scope}[xshift=-1cm, yshift=-2cm]
      \draw[fill=gray!30!white, draw=gray!80!black] (0.5, 0.22) ellipse (0.80cm and 0.6cm) node[above  left= 0.30cm and -0.15cm] {\scriptsize $R_1$};
      \node[ssssgvertex, fill=black] (k) at (0.0, 0.0) {};
      \node[ssssgvertex, fill=black] (l) at (0.5, 0.5) {};
      \node[ssssgvertex, fill=black] (m) at (1.0, 0.0) {};
    \end{scope}
    \begin{scope}[xshift=1.8cm, yshift=-1.5cm]
      \draw[fill=gray!30!white, draw=gray!80!black, rotate=30] (0, 0) ellipse (0.70cm and 0.5cm);\node[ssssgvertex, fill=black] (n) at (-0.25, -0.25) {};
      \node[ssssgvertex, fill=black] (o) at (0.25, 0.25) {};
    \end{scope}
      \node[ssssgvertex, fill=black] (p) at (-1.25, -0.5) {};
      \node[ssssgvertex, fill=black] (s) at (3.4, 0.5) {};
      \node at (3.6, 0.9) {\scriptsize $v_{\textrm{max}}^{S}$};
      \node[ssssgvertex, fill=black] (t) at (3.7, -0.25) {};
      \node[ssssgvertex, fill=black] (u) at (3.2, -1.25) {};
      \node at (2.9, -1.1) {\scriptsize $v^{S}$};

    \draw ($(j) - (2, -0.2)$) edge[->] (j);
    \draw (j) edge[->] (i);
    \draw (i) edge[->] (g);
    \draw (g) edge[->] (p);
    \draw (p) edge[->] (a);
    \draw (a) edge[->] (r);
    \draw (r) edge[->] (n);
    \draw (n) edge[->] (o);
    \draw (o) edge[->] (s);
    \draw (s) edge[->] (t);
    \draw (t) edge[->] (u);
    \draw (u) edge[->] ($(u) + (1.5, -1.0)$);
  \end{scope}

  \begin{scope}[yscale=0.7, xscale=0.55, xshift=10cm]
    \node at (0.5, 3.6) {\begin{minipage}{4.5cm}\scriptsize The lift $F$ of the tour $T'$ found in the vertebrate pair $(\cI', B)$, where $\cI'$ was obtained by contracting $R_1$ and $R_2$. \end{minipage}};
    \draw[fill=gray!10!white, draw=gray!80!black] (0.3, 0) ellipse (4.3cm and 2.5cm); \begin{scope}
      \draw[fill=gray!30!white, draw=gray!80!black] (-2.5, 0.7) ellipse (0.95cm and 0.7cm);\end{scope}
    \begin{scope}[xshift=1.0cm, yshift=0.6cm, rotate=90]
      \draw[fill=gray!30!white, draw=gray!80!black] (0, 0) ellipse (1.3cm and 2.0cm);\draw[fill=gray!50!white, draw=gray!80!black] (0, -.85) ellipse (0.6cm and 0.9cm) node[above left = .30cm and -0.2cm] {\scriptsize $R_2$};
      \draw[fill=gray!50!white, draw=gray!80!black] (0.3, 1.1) ellipse (0.5cm and 0.5cm);\node[ssssgvertex, fill=black] (a) at (0.1, 0.85) {};
       \node[ssssgvertex, fill=black] (r) at (-0.85, 0.45) {};
\draw[fill=gray!70!white, draw=gray!80!black] (0.0, -0.7) ellipse (0.4cm and 0.4cm);\node[ssssgvertex, fill=black] (b) at (-0.15, -0.55) {}; \node[ssssgvertex, fill=black] (c) at (0.15, -0.85) {};\node[ssssgvertex, fill=black]  (d) at (-0.2, -1.4) {}; \node[ssssgvertex, fill=black] (e) at (0.2, -1.4) {};\node[ssssgvertex, fill=black] (f)  at (0.4, 1.2) {};
    \end{scope}
    \node[ssssgvertex, fill=black] (g) at (-2.7, 0.4) {}; \node[ssssgvertex, fill=black] (h) at (-2.2, 0.7) {}; \node[ssssgvertex, fill=black] (i) at (-2.8, -0.7) {};
    \node[ssssgvertex, fill=black]  (j) at (-2.4, -1.4)  {};

\begin{scope}[xshift=-1cm, yshift=-2cm]
      \draw[fill=gray!30!white, draw=gray!80!black] (0.5, 0.22) ellipse (0.80cm and 0.6cm) node[above  left= 0.30cm and -0.15cm] {\scriptsize $R_1$};
      \node[ssssgvertex, fill=black] (k) at (0.0, 0.0) {};
      \node[ssssgvertex, fill=black] (l) at (0.5, 0.5) {};
      \node[ssssgvertex, fill=black] (m) at (1.0, 0.0) {};
    \end{scope}
    \begin{scope}[xshift=1.8cm, yshift=-1.5cm]
      \draw[fill=gray!30!white, draw=gray!80!black, rotate=30] (0, 0) ellipse (0.70cm and 0.5cm);\node[ssssgvertex, fill=black] (n) at (-0.25, -0.25) {};
      \node[ssssgvertex, fill=black] (o) at (0.25, 0.25) {};
    \end{scope}
      \node[ssssgvertex, fill=black] (p) at (-1.25, -0.5) {};
      \node[ssssgvertex, fill=black] (s) at (3.4, 0.5) {};
      \node[ssssgvertex, fill=black] (t) at (3.7, -0.25) {};
      \node[ssssgvertex, fill=black] (u) at (3.2, -1.25) {};

    \draw ($(j) - (2, -0.2)$) edge[->] (j);
    \draw (j) edge[->] (i);
    \draw (i) edge[->] (g);
    \draw (g) edge[->] (p);
    \draw (p) edge[->] (a);
    \draw (a) edge[->] (r);
    \draw (r) edge[->] (n);
    \draw (n) edge[->] (o);
    \draw (o) edge[->] (s);
    \draw (s) edge[->] (t);
    \draw (t) edge[->] (u);
    \draw (u) edge[->] ($(u) + (1.5, -1.0)$);

    \draw (h) edge[bend left,->] (p);
    \draw (p) edge[bend right,->] (k);
    \draw (m) edge[->] (r);
    \draw (r) edge[->] (b);
    \draw (c) edge[->, bend right] (f);
    \draw (f) edge[->, bend right]  (h); 
\draw (k) edge[->] (m);
    \draw (b) edge[->] (d);
    \draw (d) edge[->] (c);
  \end{scope}

  \begin{scope}[yscale=0.7, xscale=0.55, xshift=20cm]
    \node at (0.5, 3.6) {\begin{minipage}{4cm}\scriptsize The final tour $T$ obtained by adding results of recursive calls on $R_1$ and $R_2$ (i.e., on $\cI[R_1]$ and $\cI[R_2]$).\end{minipage}};
    \draw[fill=gray!10!white, draw=gray!80!black] (0.3, 0) ellipse (4.3cm and 2.5cm); \begin{scope}
      \draw[fill=gray!30!white, draw=gray!80!black] (-2.5, 0.7) ellipse (0.95cm and 0.7cm);\end{scope}
    \begin{scope}[xshift=1.0cm, yshift=0.6cm, rotate=90]
      \draw[fill=gray!30!white, draw=gray!80!black] (0, 0) ellipse (1.3cm and 2.0cm);\draw[fill=gray!50!white, draw=gray!80!black] (0, -.85) ellipse (0.6cm and 0.9cm) node[above left = .30cm and -0.2cm] {\scriptsize $R_2$};
      \draw[fill=gray!50!white, draw=gray!80!black] (0.3, 1.1) ellipse (0.5cm and 0.5cm);\node[ssssgvertex, fill=black] (a) at (0.1, 0.85) {};
       \node[ssssgvertex, fill=black] (r) at (-0.85, 0.45) {};
\draw[fill=gray!70!white, draw=gray!80!black] (0.0, -0.7) ellipse (0.4cm and 0.4cm);\node[ssssgvertex, fill=black] (b) at (-0.15, -0.55) {}; \node[ssssgvertex, fill=black] (c) at (0.15, -0.85) {};\node[ssssgvertex, fill=black]  (d) at (-0.2, -1.4) {}; \node[ssssgvertex, fill=black] (e) at (0.2, -1.4) {};\node[ssssgvertex, fill=black] (f)  at (0.4, 1.2) {};
    \end{scope}
    \node[ssssgvertex, fill=black] (g) at (-2.7, 0.4) {}; \node[ssssgvertex, fill=black] (h) at (-2.2, 0.7) {}; \node[ssssgvertex, fill=black] (i) at (-2.8, -0.7) {};
    \node[ssssgvertex, fill=black]  (j) at (-2.4, -1.4)  {};

\begin{scope}[xshift=-1cm, yshift=-2cm]
      \draw[fill=gray!30!white, draw=gray!80!black] (0.5, 0.22) ellipse (0.80cm and 0.6cm) node[above  left= 0.30cm and -0.15cm] {\scriptsize $R_1$};
      \node[ssssgvertex, fill=black] (k) at (0.0, 0.0) {};
      \node[ssssgvertex, fill=black] (l) at (0.5, 0.5) {};
      \node[ssssgvertex, fill=black] (m) at (1.0, 0.0) {};
    \end{scope}
    \begin{scope}[xshift=1.8cm, yshift=-1.5cm]
      \draw[fill=gray!30!white, draw=gray!80!black, rotate=30] (0, 0) ellipse (0.70cm and 0.5cm);\node[ssssgvertex, fill=black] (n) at (-0.25, -0.25) {};
      \node[ssssgvertex, fill=black] (o) at (0.25, 0.25) {};
    \end{scope}
      \node[ssssgvertex, fill=black] (p) at (-1.25, -0.5) {};
      \node[ssssgvertex, fill=black] (s) at (3.4, 0.5) {};
      \node[ssssgvertex, fill=black] (t) at (3.7, -0.25) {};
      \node[ssssgvertex, fill=black] (u) at (3.2, -1.25) {};

    \draw ($(j) - (2, -0.2)$) edge[->] (j);
    \draw (j) edge[->] (i);
    \draw (i) edge[->] (g);
    \draw (g) edge[->] (p);
    \draw (p) edge[->] (a);
    \draw (a) edge[->] (r);
    \draw (r) edge[->] (n);
    \draw (n) edge[->] (o);
    \draw (o) edge[->] (s);
    \draw (s) edge[->] (t);
    \draw (t) edge[->] (u);
    \draw (u) edge[->] ($(u) + (1.5, -1.0)$);

    \draw (h) edge[bend left,->] (p);
    \draw (p) edge[bend right,->] (k);
    \draw (m) edge[->] (r);
    \draw (r) edge[->] (b);
    \draw (c) edge[->, bend right] (f);
    \draw (f) edge[->, bend right]  (h); 
    \draw (k) edge[->] (m);
    \draw (b) edge[->] (d);
    \draw (d) edge[->] (c);

    \draw (m) edge[<-, bend right] (l);
    \draw (l) edge[<-, bend right] (k);
    \draw (k) edge[<-, bend right] (m);

    \draw (e) edge[->, bend right] (c);
    \draw (c) edge[->] (b);
    \draw (b) edge[->, bend right] (d);
    \draw (d) edge[->] (e);

\end{scope}

\end{tikzpicture}
   \caption{An illustration of the steps in the proofs of Lemma~\ref{lem:skeleton} (left) and Theorem~\ref{thm:fromSimpletoGeneral} (center and right). Only one maximal set $S\in \cL$ is shown.}
  \label{fig:vertebrate}
\end{figure}
\else
\begin{figure}[t]
  \centering
  \begin{tikzpicture}[scale=0.95,
    my style/.style={append after command={\pgfextra{\node [right=-0.1cm] at (\tikzlastnode.mid east) {\tiny \tikzlastnode};}
         },
       },
  ]
\tikzset{arrow data/.style 2 args={decoration={markings,
         mark=at position #1 with \arrow{#2}},
         postaction=decorate}
      }

  \begin{scope}[yscale=0.7, xscale=0.55]
    \node at (0.5, 3.6) {\begin{minipage}{4.5cm}\scriptsize The rerouting inside $S$ when obtaining the quasi-backbone $B$ from a lift $B'$ of a  tour $T$ of the instance $\cI'$ obtained by contracting maximal sets in $\cL$.\end{minipage}};
    \draw[fill=gray!10!white, draw=gray!80!black] (0.3, 0) ellipse (4.3cm and 2.5cm); \begin{scope}
      \draw[fill=gray!30!white, draw=gray!80!black] (-2.5, 0.7) ellipse (0.95cm and 0.7cm);\end{scope}
    \begin{scope}[xshift=1.0cm, yshift=0.6cm, rotate=90]
      \draw[fill=gray!30!white, draw=gray!80!black] (0, 0) ellipse (1.3cm and 2.0cm);\draw[fill=gray!50!white, draw=gray!80!black] (0, -.85) ellipse (0.6cm and 0.9cm) node[above left = .30cm and -0.2cm] {\scriptsize $R_2$};
      \draw[fill=gray!50!white, draw=gray!80!black] (0.3, 1.1) ellipse (0.5cm and 0.5cm);\node[ssssgvertex, fill=black] (a) at (0.1, 0.85) {};
       \node[ssssgvertex, fill=black] (r) at (-0.85, 0.45) {};
\draw[fill=gray!70!white, draw=gray!80!black] (0.0, -0.7) ellipse (0.4cm and 0.4cm);\node[ssssgvertex, fill=black] (b) at (-0.15, -0.55) {}; \node[ssssgvertex, fill=black] (c) at (0.15, -0.85) {};\node[ssssgvertex, fill=black]  (d) at (-0.2, -1.4) {}; \node[ssssgvertex, fill=black] (e) at (0.2, -1.4) {};\node[ssssgvertex, fill=black] (f)  at (0.4, 1.2) {};
    \end{scope}
    \node[ssssgvertex, fill=black] (g) at (-2.7, 0.4) {}; \node at (-2.85, 0.8) {\scriptsize $u^S_{\textrm{max}}$};
    \node[ssssgvertex, fill=black] (h) at (-2.2, 0.7) {}; \node[ssssgvertex, fill=black] (i) at (-2.8, -0.7) {};
    \node[ssssgvertex, fill=black]  (j) at (-2.4, -1.4)  {};
    \node at (-2.00, -1.2) {\scriptsize $u^S$};

\begin{scope}[xshift=-1cm, yshift=-2cm]
      \draw[fill=gray!30!white, draw=gray!80!black] (0.5, 0.22) ellipse (0.80cm and 0.6cm) node[above  left= 0.30cm and -0.15cm] {\scriptsize $R_1$};
      \node[ssssgvertex, fill=black] (k) at (0.0, 0.0) {};
      \node[ssssgvertex, fill=black] (l) at (0.5, 0.5) {};
      \node[ssssgvertex, fill=black] (m) at (1.0, 0.0) {};
    \end{scope}
    \begin{scope}[xshift=1.8cm, yshift=-1.5cm]
      \draw[fill=gray!30!white, draw=gray!80!black, rotate=30] (0, 0) ellipse (0.70cm and 0.5cm);\node[ssssgvertex, fill=black] (n) at (-0.25, -0.25) {};
      \node[ssssgvertex, fill=black] (o) at (0.25, 0.25) {};
    \end{scope}
      \node[ssssgvertex, fill=black] (p) at (-1.25, -0.5) {};
      \node[ssssgvertex, fill=black] (s) at (3.4, 0.5) {};
      \node at (3.6, 0.9) {\scriptsize $v_{\textrm{max}}^{S}$};
      \node[ssssgvertex, fill=black] (t) at (3.7, -0.25) {};
      \node[ssssgvertex, fill=black] (u) at (3.2, -1.25) {};
      \node at (2.9, -1.1) {\scriptsize $v^{S}$};

    \draw ($(j) - (2, -0.2)$) edge[->] (j);
    \draw (j) edge[->] (i);
    \draw (i) edge[->] (g);
    \draw (g) edge[->] (p);
    \draw (p) edge[->] (a);
    \draw (a) edge[->] (r);
    \draw (r) edge[->] (n);
    \draw (n) edge[->] (o);
    \draw (o) edge[->] (s);
    \draw (s) edge[->] (t);
    \draw (t) edge[->] (u);
    \draw (u) edge[->] ($(u) + (1.5, -1.0)$);
  \end{scope}

  \begin{scope}[yscale=0.7, xscale=0.55, xshift=10cm]
    \node at (0.5, 3.6) {\begin{minipage}{4.5cm}\scriptsize The lift $F$ of the tour $T'$ found in the vertebrate pair $(\cI', B)$, where $\cI'$ was obtained by contracting $R_1$ and $R_2$. \end{minipage}};
    \draw[fill=gray!10!white, draw=gray!80!black] (0.3, 0) ellipse (4.3cm and 2.5cm); \begin{scope}
      \draw[fill=gray!30!white, draw=gray!80!black] (-2.5, 0.7) ellipse (0.95cm and 0.7cm);\end{scope}
    \begin{scope}[xshift=1.0cm, yshift=0.6cm, rotate=90]
      \draw[fill=gray!30!white, draw=gray!80!black] (0, 0) ellipse (1.3cm and 2.0cm);\draw[fill=gray!50!white, draw=gray!80!black] (0, -.85) ellipse (0.6cm and 0.9cm) node[above left = .30cm and -0.2cm] {\scriptsize $R_2$};
      \draw[fill=gray!50!white, draw=gray!80!black] (0.3, 1.1) ellipse (0.5cm and 0.5cm);\node[ssssgvertex, fill=black] (a) at (0.1, 0.85) {};
       \node[ssssgvertex, fill=black] (r) at (-0.85, 0.45) {};
\draw[fill=gray!70!white, draw=gray!80!black] (0.0, -0.7) ellipse (0.4cm and 0.4cm);\node[ssssgvertex, fill=black] (b) at (-0.15, -0.55) {}; \node[ssssgvertex, fill=black] (c) at (0.15, -0.85) {};\node[ssssgvertex, fill=black]  (d) at (-0.2, -1.4) {}; \node[ssssgvertex, fill=black] (e) at (0.2, -1.4) {};\node[ssssgvertex, fill=black] (f)  at (0.4, 1.2) {};
    \end{scope}
    \node[ssssgvertex, fill=black] (g) at (-2.7, 0.4) {}; \node[ssssgvertex, fill=black] (h) at (-2.2, 0.7) {}; \node[ssssgvertex, fill=black] (i) at (-2.8, -0.7) {};
    \node[ssssgvertex, fill=black]  (j) at (-2.4, -1.4)  {};

\begin{scope}[xshift=-1cm, yshift=-2cm]
      \draw[fill=gray!30!white, draw=gray!80!black] (0.5, 0.22) ellipse (0.80cm and 0.6cm) node[above  left= 0.30cm and -0.15cm] {\scriptsize $R_1$};
      \node[ssssgvertex, fill=black] (k) at (0.0, 0.0) {};
      \node[ssssgvertex, fill=black] (l) at (0.5, 0.5) {};
      \node[ssssgvertex, fill=black] (m) at (1.0, 0.0) {};
    \end{scope}
    \begin{scope}[xshift=1.8cm, yshift=-1.5cm]
      \draw[fill=gray!30!white, draw=gray!80!black, rotate=30] (0, 0) ellipse (0.70cm and 0.5cm);\node[ssssgvertex, fill=black] (n) at (-0.25, -0.25) {};
      \node[ssssgvertex, fill=black] (o) at (0.25, 0.25) {};
    \end{scope}
      \node[ssssgvertex, fill=black] (p) at (-1.25, -0.5) {};
      \node[ssssgvertex, fill=black] (s) at (3.4, 0.5) {};
      \node[ssssgvertex, fill=black] (t) at (3.7, -0.25) {};
      \node[ssssgvertex, fill=black] (u) at (3.2, -1.25) {};

    \draw ($(j) - (2, -0.2)$) edge[->] (j);
    \draw (j) edge[->] (i);
    \draw (i) edge[->] (g);
    \draw (g) edge[->] (p);
    \draw (p) edge[->] (a);
    \draw (a) edge[->] (r);
    \draw (r) edge[->] (n);
    \draw (n) edge[->] (o);
    \draw (o) edge[->] (s);
    \draw (s) edge[->] (t);
    \draw (t) edge[->] (u);
    \draw (u) edge[->] ($(u) + (1.5, -1.0)$);

    \draw (h) edge[bend left,->] (p);
    \draw (p) edge[bend right,->] (k);
    \draw (m) edge[->] (r);
    \draw (r) edge[->] (b);
    \draw (c) edge[->, bend right] (f);
    \draw (f) edge[->, bend right]  (h); 
\draw (k) edge[->] (m);
    \draw (b) edge[->] (d);
    \draw (d) edge[->] (c);
  \end{scope}

  \begin{scope}[yscale=0.7, xscale=0.55, xshift=20cm]
    \node at (0.5, 3.6) {\begin{minipage}{4cm}\scriptsize The final tour $T$ obtained by adding results of recursive calls on $R_1$ and $R_2$ (i.e., on $\cI[R_1]$ and $\cI[R_2]$).\end{minipage}};
    \draw[fill=gray!10!white, draw=gray!80!black] (0.3, 0) ellipse (4.3cm and 2.5cm); \begin{scope}
      \draw[fill=gray!30!white, draw=gray!80!black] (-2.5, 0.7) ellipse (0.95cm and 0.7cm);\end{scope}
    \begin{scope}[xshift=1.0cm, yshift=0.6cm, rotate=90]
      \draw[fill=gray!30!white, draw=gray!80!black] (0, 0) ellipse (1.3cm and 2.0cm);\draw[fill=gray!50!white, draw=gray!80!black] (0, -.85) ellipse (0.6cm and 0.9cm) node[above left = .30cm and -0.2cm] {\scriptsize $R_2$};
      \draw[fill=gray!50!white, draw=gray!80!black] (0.3, 1.1) ellipse (0.5cm and 0.5cm);\node[ssssgvertex, fill=black] (a) at (0.1, 0.85) {};
       \node[ssssgvertex, fill=black] (r) at (-0.85, 0.45) {};
\draw[fill=gray!70!white, draw=gray!80!black] (0.0, -0.7) ellipse (0.4cm and 0.4cm);\node[ssssgvertex, fill=black] (b) at (-0.15, -0.55) {}; \node[ssssgvertex, fill=black] (c) at (0.15, -0.85) {};\node[ssssgvertex, fill=black]  (d) at (-0.2, -1.4) {}; \node[ssssgvertex, fill=black] (e) at (0.2, -1.4) {};\node[ssssgvertex, fill=black] (f)  at (0.4, 1.2) {};
    \end{scope}
    \node[ssssgvertex, fill=black] (g) at (-2.7, 0.4) {}; \node[ssssgvertex, fill=black] (h) at (-2.2, 0.7) {}; \node[ssssgvertex, fill=black] (i) at (-2.8, -0.7) {};
    \node[ssssgvertex, fill=black]  (j) at (-2.4, -1.4)  {};

\begin{scope}[xshift=-1cm, yshift=-2cm]
      \draw[fill=gray!30!white, draw=gray!80!black] (0.5, 0.22) ellipse (0.80cm and 0.6cm) node[above  left= 0.30cm and -0.15cm] {\scriptsize $R_1$};
      \node[ssssgvertex, fill=black] (k) at (0.0, 0.0) {};
      \node[ssssgvertex, fill=black] (l) at (0.5, 0.5) {};
      \node[ssssgvertex, fill=black] (m) at (1.0, 0.0) {};
    \end{scope}
    \begin{scope}[xshift=1.8cm, yshift=-1.5cm]
      \draw[fill=gray!30!white, draw=gray!80!black, rotate=30] (0, 0) ellipse (0.70cm and 0.5cm);\node[ssssgvertex, fill=black] (n) at (-0.25, -0.25) {};
      \node[ssssgvertex, fill=black] (o) at (0.25, 0.25) {};
    \end{scope}
      \node[ssssgvertex, fill=black] (p) at (-1.25, -0.5) {};
      \node[ssssgvertex, fill=black] (s) at (3.4, 0.5) {};
      \node[ssssgvertex, fill=black] (t) at (3.7, -0.25) {};
      \node[ssssgvertex, fill=black] (u) at (3.2, -1.25) {};

    \draw ($(j) - (2, -0.2)$) edge[->] (j);
    \draw (j) edge[->] (i);
    \draw (i) edge[->] (g);
    \draw (g) edge[->] (p);
    \draw (p) edge[->] (a);
    \draw (a) edge[->] (r);
    \draw (r) edge[->] (n);
    \draw (n) edge[->] (o);
    \draw (o) edge[->] (s);
    \draw (s) edge[->] (t);
    \draw (t) edge[->] (u);
    \draw (u) edge[->] ($(u) + (1.5, -1.0)$);

    \draw (h) edge[bend left,->] (p);
    \draw (p) edge[bend right,->] (k);
    \draw (m) edge[->] (r);
    \draw (r) edge[->] (b);
    \draw (c) edge[->, bend right] (f);
    \draw (f) edge[->, bend right]  (h); 
    \draw (k) edge[->] (m);
    \draw (b) edge[->] (d);
    \draw (d) edge[->] (c);

    \draw (m) edge[<-, bend right] (l);
    \draw (l) edge[<-, bend right] (k);
    \draw (k) edge[<-, bend right] (m);

    \draw (e) edge[->, bend right] (c);
    \draw (c) edge[->] (b);
    \draw (b) edge[->, bend right] (d);
    \draw (d) edge[->] (e);

\end{scope}

\end{tikzpicture}
   \caption{An illustration of the steps in the proofs of Lemma~\ref{lem:skeleton} (left) and Theorem~\ref{thm:fromSimpletoGeneral} (center and right). Only one maximal set $S\in \cL$ is shown.}
  \label{fig:vertebrate}
\end{figure}
\fi

We give an efficient algorithm for calculating a low-weight \emph{quasi-backbone} of an irreducible instance.
 \begin{definition}
   For an instance $\cI = (G, \cL, \xs, y)$,
  we call a subtour $B$ a \emph{quasi-backbone}    if 
  \begin{align*}
    2\sum_{S\in  \cL^*} y_S \leq (1-\delta) \valu(\cI)\,,
  \end{align*}
	where $\cL^* = \{S\in \cL: S\cap V(B)  = \emptyset\}$ contains those  laminar sets that  $B$  does \emph{not} visit. 
  \label{def:quasiskeleton}
\end{definition}
Recall that $\delta = \deltaval$ is the parameter in Definition~\ref{def:reducible} (of irreducible instances).
Also note that a backbone is not necessarily a quasi-backbone, as it
may not satisfy the above inequality if much $y$-value is on
singleton sets. 
Recall that $\nw = 18+ \epsilon$ denotes
the approximation guarantee for singleton instances 
as in Corollary~\ref{cor:singleton}.
\begin{lemma}
  There is a polynomial-time algorithm that,
  given an irreducible instance $\cI = (G,\cL, \xs, y)$,
  constructs a quasi-backbone $B$ such that $\cost(B) \leq (\nw + 3) \valu(\cI)$.
  \label{lem:skeleton}
\end{lemma}
\begin{proof}
Let $\cLmax$ be the family of all maximal sets in $\cL$.
We define $\cI'$ to be the instance obtained from $\cI$ by contracting all sets in $\cLmax$; recall the definition of the contraction operation from Section~\ref{sec:contraction}. 
By Fact~\ref{fact:no_increase}, the LP value does not increase, i.e.,
$\valu(\cI') \le \valu(\cI)$.
In $\cI'$, all sets in the laminar family are singletons,
therefore the new instance is a singleton instance and we can use the $\nw$-approximation algorithm (Corollary~\ref{cor:singleton})
to find a tour $T$ in $\cI'$
with $\cost_{\cI'}(T) \le \nw \valu(\cI') \le \nw \valu(\cI)$.

Now, to obtain a subtour $B$ of the original instance $\cI$,
we consider the lift $B'$ of $T$ back to $\cI$ (see Definition~\ref{def:lift}).
The lift  $B'$ is a subtour of low weight. Indeed,  $\cost_\cI(B') \leq \cost_{\cI'}(T) \leq \nw \valu(\cI)$ by Lemma~\ref{lem:lifting}. 
It also visits every maximal set $S \in \cLmax$.
However, it might not yet satisfy the inequality of Definition~\ref{def:quasiskeleton}. We therefore slightly modify the subtour $B'$ to obtain $B$ as follows. For each set $S\in \cLmax$:
\begin{enumerate}
  \item Suppose the first visit\footnote{Recall that the edges of  the subtour $B'$ are ordered by an Eulerian walk.} to $S$
in the subtour $B'$
arrives at a vertex $u^S \in S_\textrm{in}$ and departs from a vertex $v^S \in S_\textrm{out}$.

  \item Replace the segment of $B'$ from $u^S$ to $v^S$ by the union of:
    \begin{itemize}
      \item a shortest path from $u^S$ to $u^S_{\textrm{max}}$,
      \item a path from $u^S_{\textrm{max}}$ to $v^S_{\textrm{max}}$ inside $S$ as given by Lemma~\ref{lem:short_path},
      \item and a shortest path from $v^S_{\textrm{max}}$ to $v^S$,
    \end{itemize}
    where $u^S_\textrm{max} \in S_\textrm{in}$ and $v^S_\textrm{max} \in S_\textrm{out}$
are selected to maximize $D_S(u^S_\textrm{max}, v^S_\textrm{max})$.
\end{enumerate}
See the left part of Figure~\ref{fig:vertebrate} for an illustration.
The existence of the second path above is guaranteed by
Lemma~\ref{lem:path_tight_set}\ref{item:exists_path}
since $u^S_{\textrm{max}} \in S_{\textrm{in}}$.
It is clear that the obtained multiset $B$ 
is a subtour  (since $B'$ is a subtour) and that the algorithm for finding $B$ runs in polynomial time. 
It remains to bound the weight of $B$ and to show that $B$ satisfies the property of a quasi-backbone, i.e., the inequality of Definition~\ref{def:quasiskeleton}.

For the former, note that the weight of $B$ is at most the weight of the lift $B'$ plus the weight of the three paths added for each set $S\in \cLmax$. 
For such a set $S\in \cLmax$, the weight of the path from $u^S$ to $u^S_{\textrm{max}}$ is at most  $\valu(S)$ since there is a path from $u^S\in S_{\textrm{in}}$ to $u^S_{\textrm{max}}$ \emph{inside $S$} by Lemma~\ref{lem:path_tight_set}\ref{item:exists_path} and such a path can be selected to have weight at most $\valu(S)$ by Lemma~\ref{lem:short_path}. By the same argument, we have that the weight of the path from $v^S_{\textrm{max}}$ to $v^S\in S_{\textrm{out}}$ is at most $\valu(S)$. Finally, by applying Lemma~\ref{lem:short_path} again, we have that \bll{the weight of} the path added from $u^S_{\textrm{max}}$ to $v^S_{\textrm{max}}$ is also bounded by $\valu(S)$. It follows that
\begin{align*}
  \cost(B) \leq \cost(B') + 3\cdot \sum_{S\in \cLmax} \valu(S) \leq \cost(B') + 3 \valu(\cI) \leq (\nw  + 3) \valu(\cI)\,,
\end{align*}
as required.
(In the second inequality we used that the sets $S \in \cLmax$ are disjoint.)

We proceed to prove that $B$ satisfies  the inequality of Definition~\ref{def:quasiskeleton}. Recall that $\cL^* = \{S\in \cL: S\cap V(B)  = \emptyset\}$ contains those  \bll{sets in $\cL$} that  $B$  does not visit.
As $B$ visits  every $S \in \cLmax$
(i.e., $\cLmax \cap \cL^* = \emptyset$), it
 is enough to show the following:

\begin{claim*}
	For every $S \in \cLmax$
	we have
	\[ \sum_{R \in \cL^* : \ R \subsetneq S} 2 y_R \le (1 - \delta) \valu(S) \,. \]
\end{claim*}
Once we have this claim, the property of a quasi-backbone indeed follows:
\[ 2 \sum_{R \in \cL^*} y_R = \sum_{S \in \cLmax} \sum_{R \in \cL^*: \ R \subsetneq S} 2 y_R
\le \sum_{S \in \cLmax} (1 - \delta) \valu(S) \le (1 - \delta) \valu(\cI) \,. \]
\begin{proof}[Proof of Claim]
  \renewcommand\qedsymbol{$\diamond$}

  The intuition behind the claim is that, when forming $B$, we have added a path $P$ from $u^S_{\textrm{max}}$ to $v^S_{\textrm{max}}$.
  Since $S$ is irreducible, this path $P$ has a large weight.
  However, it is chosen so that it crosses each set in $\cL$ at most twice.
  Thus it must cross most (weighted by value) sets of $\cL$ contained in $S$.

  Now we proceed with the formal proof.
  As  $u^S_{\textrm{max}} \in S_{\textrm{in}}$,  the path $P$ inside $S$
from $u^S_\textrm{max}$ to $v^S_\textrm{max}$ 
that we have obtained from Lemma~\ref{lem:short_path}
crosses every  tight set $R\in \cL$ at most
$2 - |R \cap \{u^S_\textrm{max}, v^S_\textrm{max}\}|$ times.
Moreover, $P$ (a subset of $B$) does not cross any set $R \in \cL^*$.
Therefore
\[ d_S(u^S_{\textrm{max}}, v^S_{\textrm{max}}) \leq  \cost(P) \le \sum_{R \in \cL \setminus \cL^* : \ R \subsetneq S} \left( 2 - |R \cap \{u^S_\textrm{max}, v^S_\textrm{max}\}| \right) \cdot y_R \,. \]
Furthermore, we have that the quasi-backbone $B$ visits all sets in $\cL_{\textrm{max}}$ and visits both vertices $u^S_{\textrm{max}}$ and $v^S_{\textrm{max}}$. Therefore it must visit all sets $R \in \cL$ for which $R\cap \{u^S_{\textrm{max}}, v^S_{\textrm{max}}\}$ is non-empty; i.e., for all $R \in \cL^*$ we have $|R \cap \{u^S_\textrm{max}, v^S_\textrm{max}\}| = 0$. It follows that that the quasi-backbone visits most (weighted by value) laminar sets. If $u^S_{\textrm{max}}=v^S_{\textrm{max}}$, then we  have 
\[
\sum_{R \in \cL \setminus \cL^* : \ R \subsetneq S} 2 y_R=D_S(u^S_\textrm{max},v^S_\textrm{max}),
\]
and if  $u^S_{\textrm{max}}\neq v^S_{\textrm{max}}$, then
\begin{align*}
  \sum_{R \in \cL \setminus \cL^* : \ R \subsetneq S} 2 y_R &= \sum_{R \in \cL \setminus \cL^* : \ R \subsetneq S} (2 - |R \cap \{u^S_\textrm{max}, v^S_\textrm{max}\}|) \cdot y_R + \sum_{R \in \cL \setminus \cL^* : \ R \subsetneq S} |R \cap \{u^S_\textrm{max}, v^S_\textrm{max}\}| \cdot y_R \\[2mm]
&\geq  d_S(u^S_{\textrm{max}},v^S_{\textrm{max}})  + \sum_{R \in \cL : \ R \subsetneq S} |R \cap \{u^S_\textrm{max}, v^S_\textrm{max}\}| \cdot y_R  \\
&= D_S(u^S_\textrm{max},v^S_\textrm{max}) 
\end{align*}
In both cases, we have $D_S(u^S_\textrm{max},v^S_\textrm{max})\ge \delta \valu(S)$ by
the choice of $u^S_{\textrm{max}}, v^S_{\textrm{max}}$ and by the irreducibility of $S$.
The claim now follows:
\[ \sum_{R \in \cL^* : \ R \subsetneq S} 2 y_R = \valu(S) - \sum_{R \in \cL \setminus \cL^* : \ R \subsetneq S} 2 y_R \le \valu(S) - \delta \valu(S) = (1 - \delta) \valu(S) \,. \]
\end{proof}
The proof of the above claim completes the proof of Lemma~\ref{lem:skeleton}.
\end{proof}

\subsection{Obtaining a Vertebrate Pair via Recursive Calls}
\label{sec:simpleirr}
We now prove the main result of
Section~\ref{sec:reducetovertebrate}. Recall the notation
$\lb_\cI(\bar B)$ introduced in \eqref{eq:lb-B} on page~\pageref{eq:lb-B}.
\begin{theorem}
Let $\cA$ be a polynomial-time algorithm that, given a vertebrate pair
$(\cI', B)$,
returns a tour of $\cI'$ with weight at most 
\[ \kappa \valu(\cI') + \eta \lb_{\cI'}(\bar B)+ \cost_{\cI'}(B)\]
for some $\kappa, \eta \ge 0$.
Then there is a polynomial-time $\rho$-approximation algorithm (with respect to the Held--Karp relaxation) for
ATSP for irreducible instances,
where 
\[
\rho = 
\frac{\kappa +\eta(1-\delta)+\nw+3}{2\delta-1}.
\]
\label{thm:fromSimpletoGeneral}
\end{theorem}
The essence of the theorem is that if we have an algorithm for
vertebrate pairs where the approximation factor is bounded by a
constant factor of the value of the instance and the weight of the
backbone, then this translates to a constant-factor approximation for
ATSP in arbitrary irreducible instances (with no backbone given).
The proof of this theorem
is somewhat similar to that of Theorem~\ref{thm:reduction_to_irreducible},
in that
the algorithm presented here
will call itself recursively on smaller instances,
as well as invoking the black-box algorithm $\cA$
(once per recursive call). The complicated dependence on the
parameters is due to the recursive arguments. We will optimize the
parameters $\kappa$ and $\eta$ in
Section~\ref{sec:completepuzzle}. 

\begin{proof}
\newcommand{\cLsmax}{\cL^*_\textrm{max}}
We briefly discuss the intuition first.
Consider an irreducible instance $\cI = (G, \cL, \xs, y)$. By Lemma~\ref{lem:skeleton}, we can find a quasi-backbone $B$ -- a subtour such that 
$2\sum_{S \in \cL^*} y_S \le (1 - \delta) \valu(\cI)$, where as before $\cL^* = \{S \in \cL: S\cap V(B) = \emptyset\}$ contains the laminar sets that the quasi-backbone does \emph{not} visit. 
This is a small fraction of the entire optimum $\valu(\cI)$,
so we can afford to run an expensive procedure (say, a $2 \rho$-approximation)
on the unvisited sets (using recursive calls)
so as to make them contractible.
Once we contract all these sets,
 $B$ will become  a backbone in the contracted instance and we will have thus obtained a vertebrate pair,
 on which the algorithm $\cA$ can be applied.\footnote{
   Note that we never actually find a backbone of the original, uncontracted instance.
 }
See Figure~\ref{fig:vertebrate} for an illustration.

We now formally describe the $\rho$-approximation algorithm $\cA_{\textrm{irr}}$ for irreducible instances. Given an irreducible instance $\cI = (G, \cL, \xs, y)$,  it proceeds as follows:
\begin{enumerate}
  \item   Invoke  the algorithm of Lemma~\ref{lem:skeleton}
to obtain a quasi-backbone $B$
with $\cost_{\cI}(B) \le (\nw + 3) \valu(\cI)$.
Denote by $\cLsmax$ the family of maximal
(inclusion-wise) \emph{non-singleton}
sets in $\cL^*  = \{ S \in \cL : S \cap V(B) = \emptyset \}$.
(For example, in Figure~\ref{fig:vertebrate}, $R_1$ and $R_2$ are two such sets.)
  \item For each $S\in \cLsmax$, recursively call $\cA_{\textrm{irr}}$ to find a tour $T_S$ in the instance $\cI[S]$ (which is irreducible by Fact~\ref{fact:irreducibleinv}). Then use $T_S$ and the algorithm of Lemma~\ref{lem:get_contractible} to find  a collection $F_S$ of subtours such that $S$ is contractible with respect to $F_S$ and $\cost_\cI(F_S) \leq \cost_{\cI[S]} (T_S)$. 
  \item Let $\cI'=(G',\cL',x',y')$ be the instance obtained
from $\cI$ by contracting all the maximal sets $S \in \cLsmax$; let
$V'$ denote the contracted ground set. 
We have that $(\cI', B)$ is a vertebrate pair by construction:
we have contracted all \bll{sets in $\cL$} that were not visited by $B$
into single vertices, and so  $B$ is  a backbone of $\cI'$. 
Note that
\[
\lb_{\cI'}(\bar B)=2\sum_{v\in V'\setminus V(B)}y'_v\le 2\sum_{S \in \cL^*} y_S
\le (1 - \delta) \valu(\cI).
\]
The first inequality follows by the definition of contraction, using
Fact~\ref{fact:bound_on_dU}.
We can
invoke the algorithm $\cA$ on the vertebrate pair $(\cI', B)$; by the hypothesis
of the theorem, it
returns a tour $T'$ of $\cI'$ such that 
\begin{equation}\label{eq:w-T}
w_{\cI'}(T') \leq \kappa \valu(\cI') + \eta (1 - \delta) \valu(\cI)
+w_{\cI'}(B).
\end{equation}
 \item Finally, return the tour $T$ consisting of the lift $F$ of $T'$ to $\cI$  together with $\bigcup_{S\in \cLsmax} F_S$. (See the center and right parts of Figure~\ref{fig:vertebrate} for an illustration.)
\end{enumerate}
We remark that $T$ is indeed a tour of $\cI$ since all sets $S\in \cLsmax$ are contractible with respect to $\bigcup_{S\in \cLsmax} F_S$. 

Having described the algorithm, it remains to show that $\cA_{\textrm{irr}}$ runs in polynomial time and that it has an approximation guarantee of $\rho$. 

For the former, we bound the total number of recursive calls that $\cA_{\textrm{irr}}$ makes. We claim that the total number of recursive calls on input $\cI = (G, \cL, x, y)$ is at most the cardinality of $\cL_{\geq 2} = \{S\in \cL: |S| \geq 2\}$. The proof is by induction on $|\cL_{\geq 2}|$. For the base case, i.e.,  when $|\cL_{\geq 2}| = 0$, there are no recursive calls since there are no non-singleton sets in $\cL^* \subseteq \cL$ and so $\cLsmax = \emptyset$. For the inductive step, suppose that $\cLsmax = \{S_1, S_2, \ldots, S_\ell\} \subseteq \cL^*$  and so there are $\ell$ recursive calls in this iteration -- on the instances $\cI[S_1], \cI[S_2], \ldots, \cI[S_\ell]$.  If we let $\cL^i_{\geq 2}$ denote the non-singleton laminar sets of $\cI[S_i]$ then, by the definition of inducing on a tight set, for every $R \in \cL^i_{\geq 2}$ we have $R \subsetneq S_i$ and $R\in \cL_{\geq 2}$. It follows by the induction hypothesis that  the total number of recursive calls that $\cA_{\textrm{irr}}$ makes is
\begin{align*}
  \ell + \sum_{i=1}^\ell |\cL_{\geq 2}^i| \leq \ell + |\cL_{\geq 2}| - \ell = |\cL_{\geq 2}| \,,
\end{align*}
where the inequality holds because the sets $\cL^i_{\geq 2}$ are disjoint and
$\cL^1_{\geq 2} \cup \cL^2_{\geq 2} \cup \cdots \cup \cL^\ell_{\geq 2}
\subseteq \cL_{\geq 2} \setminus \{S_1, S_2, \ldots, S_\ell\}$.  Hence, the
total number of recursive calls $\cA_{\textrm{irr}}$ makes is $|\cL_{\geq 2}| \le |\cL|$,
which is at most linear in $|V|$.
The fact that $\cA_{\textrm{irr}}$ runs in polynomial time now follows because each call runs in polynomial time. Indeed, the algorithm of Lemma~\ref{lem:skeleton},
the algorithm of Lemma~\ref{lem:get_contractible}, and $\cA$ all run in polynomial time.

We now complete the proof of the theorem by showing that
$\cA_{\textrm{irr}}$ is a $\rho$-approximation algorithm. From
\eqref{eq:w-T} and by  Lemma~\ref{lem:lifting} we have that the weight $w_{\cI}(F)$ of the lift $F$ of $T'$ is at most
  \[\cost_{\cI'}(T') \le \kappa \valu(\cI') + \eta (1 - \delta) \valu(\cI)
+w_{\cI'}(B) \le 
(\kappa +\eta(1-\delta)+\nw+3)\valu(\cI),
  \]
where the second inequality follows by Fact~\ref{fact:no_increase}
and since $w_{\cI'}(B) = w_{\cI}(B)$
($\cI'$ arises by contracting only sets not visited by $B$, which preserves the weight of $B$)
and $w_{\cI}(B) \leq (\nw + 3) \valu(\cI)$.

Now, to show that $\cost(T) = \cost(F) + \cost\left( \bigcup_{S\in \cLsmax} F_S \right)\leq \rho \valu(\cI)$,  we proceed by induction on the total number of recursive calls. In the base case, when no recursive calls are made, we have $\cost(T) = \cost(F) \leq  \cost_{I'}(T') \leq (\kappa +\eta(1-\delta)+\nw+3) \valu(\cI) \leq \rho \valu(\cI)$. For the inductive step, the induction hypothesis yields that for each $S\in \cLsmax$ we have
\[ \cost(F_S) \leq \cost_{\cI[S]}(T_S) \leq \rho \valu(\cI[S]) = 2\rho\valu(S) \, , \]
where the equality is by Fact~\ref{fact:value_of_induced_instance}. Hence
\begin{align*}
  \cost\left( \bigcup_{S\in \cLsmax} F_S \right) &= \sum_{S \in \cLsmax} \cost(F_S) \le \sum_{S \in \cLsmax} 2 \rho \valu(S)  \\
&= \sum_{S \in \cLsmax} 2 \rho \sum_{R \in \cL^* : \ R \subsetneq S} 2 y_R
\le 2 \rho \sum_{R \in \cL^*} 2 y_R \le 2 \rho (1 - \delta) \valu(\cI)\,.
\end{align*}
\bll{The second equality uses that $\valu(S) = \sum_{R \in \cL : \ R \subsetneq S} 2 y_R = \sum_{R \in \cL^* : \ R \subsetneq S} 2 y_R$ for all $S \in \cLsmax$:
as $S \cap V(B) = \emptyset$, any $R \in \cL$ with $R \subsetneq S$ also has $R \cap V(B) = \emptyset$ and thus $R \in \cL^*$.}
The last inequality holds because $B$ is a quasi-backbone of $\cI$ (see Definition~\ref{def:quasiskeleton}).  Summing up the weight of the lift $F$ of $T'$ and of $\bigcup_{S\in \cLsmax} F_S$ we get
\begin{align*}
  \cost(T) \leq \left(\kappa +\eta(1-\delta)+\nw+3 + 2\rho(1-\delta) \right) \valu(\cI) = \rho \valu(\cI)\,,
\end{align*}
by the selection of  $\rho$ to equal
$ (\kappa +\eta(1-\delta)+\nw+3)/(2\delta-1)$.
This concludes the inductive step and the proof of the theorem.
\end{proof}

\part{Solving \EPC}
\label{part:solvingvert}
In this part we solve \EPC{} on vertebrate pairs. \bl{Recall the definition of \EPC{} given in Section~\ref{sec:LCdef}: we are given an instance $\cI = (G, \cL, x, y)$ with 
a subtour $B$ in $G$,
and a partition $(V_1, V_2, \ldots, V_k)$
	of $V\setminus V(B)$, where each $V_i$ is strongly connected. Our goal is to find a collection $F$ of  subtours   of $E$ crossing each of the sets $V_i$, and satisfying certain ``local'' and ``global'' cost bounds.}
The main
technical concept in the argument
is that of \emph{witness flows}.
On a high level,
we want
every subtour $T$ in
our solution to \EPC{}
to be forced to intersect the backbone $B$
if $T$ crosses a non-singleton set in the laminar family $\cL$.
Every subtour that does not
cross any such set locally
behaves  as in a singleton instance with regard to its cost,
and it is easy to account for those subtours.
On the other hand, we are able to take care of the cost of all subtours that
do cross some such set (and thus also intersect $B$), together with $B$,
using a global cost argument.
The witness flow is a tool
that allows us to enforce this crucial property
in our solution to \EPC{}.
It is inspired by a general method of ensuring connectivity
in integer/linear programming formulations for graph problems,
which requires the existence of a flow
(supported on the LP solution)
between the pairs of vertices
that should be connected.

We note that witness flows used in this paper can be seen as a more
concise variant of a previous argument using split graphs in the
conference version \cite{SvenssonTV18}. Split graphs have been first
used in a similar role in \cite{SvenssonTV16}.

By the reductions in  the previous parts, this is sufficient for obtaining a constant-factor approximation algorithm. We combine all the ingredients and calculate the obtained ratio in Section~\ref{sec:completepuzzle}.

\section{Algorithm for Vertebrate Pairs}
\label{sec:solvingvertebrate}

In this section we consider a vertebrate pair $(\mathcal{I},B)$ and
prove the following theorem and corollary. This provides the algorithm required in Theorem~\ref{thm:fromSimpletoGeneral}.
\begin{theorem}
There exists a $(4,2\valu(\cI)+\lb_\cI(\bar B))$-light algorithm for \EPC{} for
  vertebrate pairs
$(\mathcal{I},B)$.
  \label{thm:finish_skeleton}
\end{theorem}
(Refer to Definition~\ref{def:light_algorithm} of a light algorithm on page~\pageref{def:light_algorithm}.)
Combined with Theorem~\ref{thm:LocalToGlobal}, we immediately obtain
the following.
\begin{corollary}  \label{cor:finish_skeleton}
For every $\epsilon > 0$ there is a polynomial-time algorithm that,
given a vertebrate pair $(\mathcal{I},B)$,
returns a  tour $T$ of $\mathcal{I}$  with $w(T) \leq 2 \valu(\cI)+ (37+36\varepsilon)
\lb_\cI(\bar B)
+ w(B)$.
\end{corollary}
Throughout this section we will assume that $B \ne \emptyset$.
  In the special case when $B = \emptyset$, it must be the case that
  $\cL_{\geq 2} = \emptyset$ and thus the instance is singleton; in that case, 
  we simply apply the strictly better $(2,0)$-light algorithm of
  Theorem~\ref{thm:lcapprox}.

\medskip

We now formulate our main technical lemma. Let $\cL_{\geq 2}$ denote the family of
non-singleton sets in 
$\cL$.

\NewDocumentCommand{\statementoftechnicallemma}{s}{There is a polynomial-time algorithm that solves the following problem.
  Let $(\cI,B)$ be a vertebrate pair,
  and let $U_1, \ldots, U_{\ell} \subseteq V\setminus V(B)$ be disjoint non-empty vertex sets  such
  that the subgraphs $G[U_1], \ldots, G[U_{\ell}]$ are strongly connected and for every $S\in \cL_{\geq 2}$ and $i = 1, \ldots, \ell$ we have either
  $U_i \cap S = \emptyset$ or $U_i \subseteq S$.
  Then the algorithm
  finds a collection of subtours $F\subseteq E$ such that:
  \begin{enumerate}[label=\emph{(\alph*)}]\itemsep0mm
      \item \IfBooleanF{#1}{\label{lc:cost}} $w(F) \leq 2
        \valu(\cI)+\lb_{\cI}(\bar B)$, 
      \item \IfBooleanF{#1}{\label{lc:cross}} $|\delta^-_F(U_i) | \geq 1$ for every $i
        = 1,\ldots, \ell$, 
      \item \IfBooleanF{#1}{\label{lc:degree}} $|\delta^-_F(v)| \leq
        4$  whenever  $x(\delta^-(v))=1$,
   \item  \IfBooleanF{#1}{\label{lc:nobadcycle}} any subtour in $F$ that crosses a 
     set in $\cL_{\geq 2}$ visits a vertex of the backbone.
 \end{enumerate}
}
\begin{lemma}
 \statementoftechnicallemma
 \label{lem:mainLCATSP}
 \end{lemma}

Notice that the requirements
on the disjoint sets $U_1, \ldots, U_\ell$ imply that
$\cL_{\geq 2} \cup \{U_1, \ldots, U_\ell\}$ is a laminar family in which the sets $U_1 , \ldots, U_\ell$ are minimal (see the left part of Figure~\ref{fig:lemma}).
We also remark that property~\ref{lc:nobadcycle} will be important for analyzing the lightness of our tour.
Indeed, it will imply that any subtour in our solution $F^\star$  to \EPC{} that is disjoint from the backbone does not cross a set in $\cL_{\geq 2}$.  
Thus any edge $(u,v)$ in such a subtour will have weight equal to $y_u + y_v$. Intuitively, this
(almost) reduces the problem to the singleton case.

The proof will be given in
Section~\ref{sec:witness-flow}, using the concept of {\em witness flows} that allow us to enforce the crucial property~\ref{lc:nobadcycle}.

\begin{figure}[t]
  \centering
  \begin{tikzpicture}
\tikzset{arrow data/.style 2 args={decoration={markings,
         mark=at position #1 with \arrow{#2}},
         postaction=decorate}
      }

  \begin{scope}[scale=0.8]
    \draw[fill=gray!15!white, draw=gray!80!black] (-2.5, -0.5) ellipse (1.5cm and 2cm); \draw[fill=none,thick, densely dotted, draw=gray!80!black,rotate around={20:(-3.2,-1.5)}] (-3.2, -1.5) ellipse (0.40cm and 0.7cm); \begin{scope}
      \draw[fill=gray!40!white, draw=gray!80!black] (-2.5, 0) ellipse (0.75cm and 1cm); \end{scope}
    \begin{scope}[xshift=2.0cm,rotate=30]
      \draw[fill=gray!15!white, draw=gray!80!black] (0, 0) ellipse (1cm and 1.5cm); \draw[fill=gray!40!white, draw=gray!80!black,rotate=-5] (0, -.75) ellipse (0.5cm and 0.6cm); \draw[fill=none,thick, densely dotted, draw=gray!80!black,rotate=0] (-0.1, -.6) ellipse (0.25cm and 0.2cm); \draw[fill=none,thick, densely dotted, draw=gray!80!black,rotate=0] (-0.3, .45) ellipse (0.30cm and 0.5cm); \end{scope}
      \draw[fill=gray!15!white, draw=gray!80!black,rotate=-5] (0.5, -1.70) ellipse (0.5cm and 0.6cm); \draw[fill=none,thick, densely dotted, draw=gray!80!black,rotate=0] (0.4, -1.90) ellipse (0.30cm and 0.25cm); \node[ssssgvertex, minimum size = 0pt, fill=none, draw=none] (u) at (-2.4, 0.5) {};
    \node[ssssgvertex, minimum size = 0pt, fill=none, draw=none] (a) at (-2.7, 0.0) {};
    \node[ssssgvertex, minimum size = 0pt, fill=none, draw=none] (b) at (-2.4, -0.5) {};
    \node[ssssgvertex, minimum size = 0pt, fill=none, draw=none] (c) at (-3.0, -1.5) {};
    \node[ssssgvertex, minimum size = 0pt, fill=none, draw=none] (d) at (-2.0, -1.8) {};

    \node[ssssgvertex, minimum size = 0pt, fill=none, draw=none] (e) at (0.2, -2.0) {};
    \node[ssssgvertex, minimum size = 0pt, fill=none, draw=none] (f) at (0.35, -1.5) {};
    \node[ssssgvertex, minimum size = 0pt, fill=none, draw=none] (g) at (0.5, -2.0) {};

    \node[ssssgvertex, minimum size = 0pt, fill=none, draw=none] (h) at (1.6, 0.8) {};
    \node[ssssgvertex, minimum size = 0pt, fill=none, draw=none] (i) at (1.4, 0.4) {};

    \node[ssssgvertex, minimum size = 0pt, fill=none, draw=none] (j) at (2.3, 0.3) {};
    \node[ssssgvertex, minimum size = 0pt, fill=none, draw=none] (k) at (2.3, -0.4) {};
    \node[ssssgvertex, minimum size = 0pt, fill=none, draw=none] (l) at (2.5, -1.0) {};

    \draw [black, thick] plot [smooth cycle,tension=1.0] coordinates { (h) (u) (d)  (f) (l)  };
\end{scope}

  \begin{scope}[xshift=8.0cm,scale=0.8]
    \draw[fill=gray!15!white, draw=gray!80!black] (-2.5, -0.5) ellipse (1.8cm and 2cm); \draw[fill=gray!40!white, draw=gray!80!black,rotate around={-60:(-2.7,-1.5)}] (-2.7, -1.5) ellipse (0.40cm and 0.7cm); \begin{scope}
      \draw[fill=gray!40!white, draw=gray!80!black] (-2.5, 0) ellipse (1.25cm and 0.8cm) node[above left = 0cm and 0.3cm] {\small $S$};
    \end{scope}
    \begin{scope}
      \clip(-2.5, 0) ellipse (1.25cm and 0.8cm);
      \draw [pattern=north west lines, pattern color=gray!70!black, draw=none]  (-1,-0.5) ellipse (2cm and 1.5cm);
\end{scope}
    \draw [draw=black, dashed]  (-1,-0.5) ellipse (2cm and 1.5cm) node [above right = 0.2cm and 0.5cm] {\small $V_i$};\draw[fill=white,thick, densely dotted, draw=gray!80!black] (-2, 0) ellipse (0.45cm and 0.35cm) node {\small $U_i$}; \end{scope}

\end{tikzpicture}
   \caption{On the left, the ``dotted'' sets $U_1, \ldots, U_\ell$ of
    Lemma~\ref{lem:mainLCATSP} are depicted. \\ On the right, we show
    how the set $U_i$ is obtained by the algorithm for \EPC{} in the
    proof of Theorem~\ref{thm:finish_skeleton}: $V_i$ is intersected with a minimal non-singleton set $S$ to obtain $V_i'$ (the striped area). Then, $U_i$ is a source component in the decomposition of $V_i'$ into strongly connected components. This implies that there are no edges from $V_i' \setminus U_i$ to $U_i$ and so any edge in $\delta(V_i \setminus U_i, U_i)$ must come from outside of $V_i'$ and thus cross the tight set $S$.}
  \label{fig:lemma}
\end{figure}

\begin{proof}[Proof of Theorem~\ref{thm:finish_skeleton}]
Let $(V_1,V_2,\ldots,V_k)$ be the input partition of $V\setminus
V(B)$ in the  \EPC{} problem. We will apply
Lemma~\ref{lem:mainLCATSP} for a collection $(U_1,U_2,\ldots,U_k)$ of disjoint subsets
with $U_i\subseteq V_i$, defined as follows.
For $i=1,\ldots, k$,
  let $V'_i$ be the intersection of $V_i$ with a minimal set
  $S\in \cL_{\geq 2}\cup \{V\}$ with $S \cap V_i \neq \emptyset$.  Then  consider a decomposition of $V_i'$ into
strongly connected components (with respect to $G[V_i']$).  Let
$U_i \subseteq V_i'$ be the vertex set of a source component in this
decomposition. That is, there is no edge from $V_i'\setminus
U_i$ to $U_i$ in $G$ (see also the right part of Figure~\ref{fig:lemma}). By construction, the sets $U_1,\ldots, U_k$ satisfy the conditions of Lemma~\ref{lem:mainLCATSP};
in particular,  since $V_i'$ (and thus $U_i$) is a subset of the minimal set $S$ chosen above, it follows that $\cL_{\geq 2}\cup \{U_i\}$ is a laminar family, \bll{
 and there are no subsets $S' \subsetneq U_i$, $S' \in \cL_{\ge 2}$.
}
We let $F$ be the Eulerian multiset guaranteed by
Lemma~\ref{lem:mainLCATSP}. 

The rest of the proof is dedicated to showing that $F$ satisfies the
requirement of an $(4,2\valu(\cI)+\lb_\cI(\bar B))$-light  algorithm. Let us start with the connectivity requirement.

\begin{claim}
  \label{cl:lcsol}
  We have  $|\delta^-_{F}(V_i)| \geq 1$  for  $i=1, 2, \ldots, k$. 
\end{claim}
\begin{proof}
 By property~\ref{lc:cross} of Lemma~\ref{lem:mainLCATSP},
there exists an edge $e\in\delta^-_{F}(U_i)$.
Then either $e\in \delta_{F}^-(V_i)$ (in which case we are done), or $e\in
\delta_{F}(V_i\setminus U_i,U_i)$. Assume the latter case.

 Using that $U_i$ was a source component in the decomposition of $V_i'$ into strongly connected components, $e$ must enter 
a set in $\cL_{\geq 2}$.  Indeed, recall that $U_i$ was selected so that there is no edge from $V_i' \setminus U_i$ to $U_i$. 
Since $e\in
\delta_{F}(V_i\setminus U_i,U_i)$
and
$\delta(V_i' \setminus U_i, U_i) = \emptyset$,
we must have
$e \in \delta(V_i \setminus V_i', U_i) \subseteq \delta(V_i \setminus V_i', V_i')$.
However, $V_i'$ was obtained by intersecting $V_i$ with a minimal set $S\in \cL_{\geq 2} \cup \{V\}$ with $S\cap V_i \neq \emptyset$.
Thus we must have
$e\in \delta^-_{F}(S)$ (and $S \ne V$). Now, property
\ref{lc:nobadcycle} of Lemma~\ref{lem:mainLCATSP} guarantees that the
connected component (i.e., the subtour) of $F$ containing $e$ must visit
$V(B)$. This subtour thus visits both $V_i$ (the head of $e$ is in $U_i \subseteq V_i$) and $V(B)$, which is disjoint from $V_i$. As such, the subtour must cross $V_i$, i.e., we have 
$|\delta_{F}^-(V_i)| \geq 1$ as required.
\end{proof}

Next, let us consider subtours in $F$ that are disjoint from $B$.

\begin{claim}
Let $T$ be a subtour in $F$ with $V(T) \cap V(B) = \emptyset$.  Then $w(T)\le 4\lb(T)$.
\end{claim}
\begin{proof}
Recall that the lower bound is
  \[ \lb(T) = 2\sum_{v\in V(T)} y_{v} \,. \]
  To bound the weight of $T$, note that by property~\ref{lc:nobadcycle} of Lemma~\ref{lem:mainLCATSP}, the edges of $T$ do not cross any tight set in $\cL_{\geq 2}$.   Therefore any edge $(u,v)$ in $T$ has weight $y_u + y_v$ and so
  \[ w(T) = \sum_{e\in T} w(e) = \sum_{v\in V(T)}  |\delta_{F}(v)| y_{v} \le 8 \sum_{v\in V(T)}   y_{v} = 4 \lb(T) \,, \]
  where for the inequality we used that $y_v$ is only strictly
  positive if $\xs(\delta^-(v)) = 1$ (see Definition~\ref{def:lam}), in which case $|\delta_{F}(v)| = 2 |\delta^-_{F}(v)|\leq 8$ using
  property~\ref{lc:degree} of Lemma~\ref{lem:mainLCATSP}.
\end{proof}

Finally, let $F_B \subseteq F$ be the collection of subtours in $F$ that intersect $B$.   Then, $w(F_B)\le w(F)\le 2\valu(\cI)+\lb_\cI(\bar B)$ by
property~\ref{lc:cost} of Lemma~\ref{lem:mainLCATSP}. This completes
the proof that $F$ is a $(4,2\valu(\cI)+\lb_\cI(\bar B))$-light edge set.
\end{proof}

The rest of this section is devoted to the proof of the
main technical
Lemma~\ref{lem:mainLCATSP}.
\subsection{Witness flows}
\label{sec:witness-flow}

Recall that  $\cL_{\geq 2}$ denotes the family of non-singleton sets in
$\cL$. Let us use an indexing $\cL_{\geq 2}\cup\{V\} = \{S_1, S_2, \ldots, S_\ell \}$ such that $2 \leq |S_1| \leq |S_2| \leq \cdots \leq |S_\ell| = |V|$.
For a vertex $v\in V$
let  $$\level(v) = \min \{i: v\in S_i\}$$ be the index of the first (smallest)
set that contains $v$.
We use these levels to define a partial order $\prec$ on the vertices:
let  $v\prec v'$ if
$\level(v) < \level(v')$.  
This partial order is used to classify the edges as follows.
An edge $(u,v) \in E$ is a 
\begin{itemize}
  \item \emph{forward} edge if $v \prec u$,
  \item \emph{backward} edge if $u \prec v$,
\end{itemize}
and otherwise  it is a \emph{neutral} edge. 
Let $E_f$, $E_b$ and $E_n$ denote the sets of forward, backward, and
neutral edges respectively.

\begin{definition}
Let $z:E\to \R$ be a circulation. We say that $f:E\to \R$ is a \emph{witness
flow} for $z$ if
\begin{enumerate}[label=\emph{(\alph*)}]\itemsep0mm
    \item \label{flow:capacities} $0\leq f\leq z$, 
    \item \label{flow:conservation} $f(\delta^+(v)) \geq f(\delta^-(v))$ for every $v\in V \setminus V(B)$,
  \item \label{flow:backward}  $f(e) = 0$  for each backward edge $e\in E_b$, 
    \item \label{flow:forward} $f(e) = z(e)$  for each forward edge $e\in E_f$.
  \end{enumerate}
We say that a circulation $z$ is {\em witnessed} if there exists a witness
flow for $z$.
\end{definition}

The following lemma reveals the importance of witness flows. By a
component of a circulation $z$ we mean a connected component of
the edge set $\supp(z) = \{ e \in E : z_e > 0 \}$.
\begin{lemma}\label{lem:witness-cross}
Let $z$ be a witnessed circulation.
Any component $C$ of $z$ that crosses a set in $\cL_{\geq 2}$ must intersect $B$.
\end{lemma}
\begin{proof}
Let $f$ be a witness flow for $z$.
Take $i$ to be the smallest value such that $S_i \cap V(C) \ne \emptyset$.
Then we must also have $V(C) \setminus S_i \ne \emptyset$, as otherwise $\cL_{\geq 2}\cup\{V(C)\}$ would be laminar, contradicting the choice of $C$.
So, $C$ must have an edge entering $S_i$;
moreover,
all edges of $C$ entering $S_i$ are forward edges,
and
all edges of $C$ exiting $S_i$ are backward edges.
Let $U = S_i \cap V(C)$.
Then $f(\delta^-(U)) > 0$ and $f(\delta^+(U)) = 0$.
If we had $U \cap V(B) = \emptyset$,
then \ref{flow:conservation} would imply that
\[ 0 > f(\delta^+(U)) - f(\delta^-(U)) = \sum_{v \in U} f(\delta^+(v)) - f(\delta^-(v)) \ge 0 \,, \]
a contradiction.
\end{proof}

The edge set $F$ in Lemma~\ref{lem:mainLCATSP} will be obtained from a
witnessed integer circulation; in particular, we will use the witness
flow to derive property \ref{lc:nobadcycle} in Lemma~\ref{lem:mainLCATSP}. We start by showing that
the Held-Karp solution $x$ is a witnessed circulation.

\begin{lemma}\label{lem:x-witness}
The Held-Karp solution $x$ is a witnessed circulation.
\end{lemma}

Before giving the full proof, let us motivate the existence of 
a witness flow $f$
in a simple example scenario
where there is only one non-singleton set $S \in \cL_{\geq 2}$.
Then we have $E_f = \delta^-(S)$ and $E_b = \delta^+(S)$,
i.e., the forward/backward edges are exactly the incoming/outgoing edges of $S$.
The subtour elimination constraints imply
(via the min-cut max-flow theorem)
that $x$ supports a unit flow between any pair of vertices.
Let $f$ be such a flow from any vertex outside $S$ to a vertex $v \in S \cap V(B)$
(such a $v$ exists by the backbone property).
It is easy to see that $f$ satisfies the conditions of the claim.
Indeed, since $S$ is a tight set,
$f$ saturates all incoming (forward) edges.
It also does not leave $S$, i.e., use any backward edges.
The proof of the general case uses an argument based on LP duality
to argue the existence of $f$.

\begin{proof}[Proof of Lemma~\ref{lem:x-witness}]
We find a witness flow  $f$ in polynomial time by solving the following linear program:
  \begin{align*}
  \arraycolsep=1.4pt\def\arraystretch{1.2}
  \begin{array}{lrlr}
    \mbox{maximize} \qquad & \displaystyle \sum_{e\in E_f} f(e) \\[7mm]
  \mbox{subject to} \qquad  & \displaystyle f(\delta^+(v)) \geq &  \displaystyle f(\delta^-(v)) & \text{ for  } v\in V\setminus V(B), \\  
  & \displaystyle f(e) = & 0 & \text{ for }e\in E_b, \\
  &0 \leq  f \leq & \xs.
  \end{array}
  \end{align*}
  By the constraints of the linear program, we have that $f$
  satisfies~\ref{flow:capacities},~\ref{flow:conservation},
  and~\ref{flow:backward}. It remains to
  verify~\ref{flow:forward}, or equivalently, to show that the
  optimum value of this program equals $x(E_f)$.

This will be shown using the dual linear program. The variables
  $(\pi_v)_{v\in V}$ correspond to the first set of constraints, and
  $(z(e))_{e\in E_f\cup E_n}$ to the capacity constraints on forward
  and neutral edges. No such variables are needed for backward
  edges. For notational simplicity, we introduce $\pi_v$ also for
  $v\in V(B)$, and set $\pi_v=0$ in this case. The dual program can be
  written as follows.
 \begin{align*}
  \arraycolsep=1.4pt\def\arraystretch{1.2}
  \begin{array}{lrlr}
    \mbox{minimize} \qquad & \displaystyle \sum_{e\in E_f\cup E_n} x(e)z(e) \\[7mm]
  \mbox{subject to} \qquad  & \displaystyle  \pi_v-\pi_u+z(u,v)\ge&  1   &  \text{ for } (u,v)\in  E_f, \\  
  & \displaystyle  \pi_v-\pi_u+z(u,v)\ge&  0 & \text{ for } (u,v)\in
                                               E_n, \\  
 & \displaystyle \pi_v=& 0 & \text{ for } v\in V(B),\\
& \pi, z \ge& 0. &
  \end{array}
  \end{align*}
  Note that setting $\pi =\vec{0}$, $z(e)=1$ if $e\in E_f$, and $z(e)=0$
if $e\in E_n$ is a feasible solution with objective value $x(E_f)$. We complete the proof by showing that this is an optimal dual
solution.

Let us select a dual optimal
solution $(\pi,z)$ that minimizes $\pi(V)$. We show
that this minimum value is 0, that is, there exists a dual optimal
solution with $\pi= \vec{0}$. This immediately implies that the above solution is
a dual optimal one,
because given $\pi = \vec{0}$ we get a constraint $z(u,v) \ge 1$ for all $(u,v) \in E_f$,
making the objective value at least $x(E_f)$.

Towards a contradiction, assume that $\pi$ is not everywhere zero. Let us select a vertex $t \in V$ such that $\pi_t>0$, and $\level(t)=i^*$
is minimal among all such vertices. Let $S=S_{i^*}$, and define
$T=\{u\in S: \pi_u>0\}$. Since $S\cap V(B)\neq
\emptyset$, and $\pi_u=0$ for all $u\in V(B)$, we see that $T$ is a
proper subset of $S$. Let $F= \delta(V\setminus S,T)$ and
$F'=\delta(T,S\setminus T)$. A depiction of the sets is as follows:
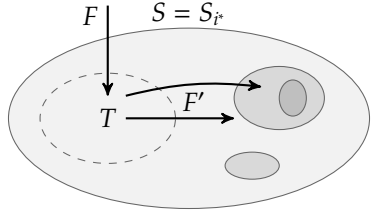
\begin{center}
  \begin{tikzpicture}
    \begin{scope}[scale=0.6]
      \draw[fill=gray!10!white, draw=gray!80!black] (0, 0) ellipse (4cm and 2cm) node[above = 1.1cm ] {\small $S = S_{i^*}$};
      \draw[fill=none, draw=gray!80!black, dashed] (-1.8, 0) ellipse (1.5cm and 1cm) node {\small $T$};
      \begin{scope}[xshift=0.2cm]
        \draw[fill=gray!30!white, draw=gray!80!black] (1.8, 0.45) ellipse (1.cm and 0.7cm);
        \draw[fill=gray!50!white, draw=gray!80!black] (2.1, 0.45) ellipse (0.3cm and 0.4cm);
      \end{scope}
      \draw[fill=gray!30!white, draw=gray!80!black] (1.4, -1.05) ellipse (0.6cm and 0.3cm);
      \draw (-1.8, 2.5) edge[->, thick] node[above left = 0.2cm and -0.01cm] {\small $F$} (-1.8, 0.5);
      \draw (-1.4, 0) edge[->, thick] node[above right = -0.05cm and -0.1cm] {\small $F'$} (1, 0);
      \draw (-1.4, 0.5) edge[->, thick, bend left=10] (1.6, 0.7);
    \end{scope}

  \end{tikzpicture}
\end{center}
Let us show that the edges between
$T$ and $S\setminus T$ can be only of certain types.
\begin{claim*}
$F'=\delta(T,S\setminus T)\subseteq E_f\cup E_n$ and $\delta(S\setminus T,T)\subseteq E_b\cup E_n$.
\end{claim*}
\begin{proof}[Proof of Claim]
    \renewcommand\qedsymbol{$\diamond$}
By the choice of $i^*$, for any $S_i\subsetneq S$ we have that
$\pi_v=0$ for all $v\in S_i$; thus $S_i\cap T=\emptyset$, and so for every $u \in T$ we have $\level(u) = i^*$. Therefore,
for every $(u,v)\in F'$, we must have $\level(u)=i^*\ge \level(v)$, and
for every $(u,v)\in \delta(S\setminus T,T)$, $\level(v)=i^*\ge \level(u)$.
\end{proof}

Let us construct another dual solution $(\pi',z')$ as follows. Let
$\varepsilon=\min\{\pi_u:\ u\in T\}$, and let
\begin{align*}
\pi'_u =  \begin{cases}
   \pi_u-\varepsilon &
    \mbox{if }u\in T, \\
\pi_u &  \mbox{otherwise,}
  \end{cases}
\quad \mbox{and} \quad
  z'(e) =  \begin{cases}
   z(e)+\varepsilon &
    \mbox{if }e\in F\cap (E_f\cup E_n), \\
  z(e)-\varepsilon &
    \mbox{if }e\in F', \\
z(e) &  \mbox{otherwise.}
  \end{cases}
\end{align*}
We show that $(\pi',z')$ is another optimal solution. Then,
$\pi'(V)<\pi(V)$ gives a contradiction to the choice of $(\pi,z)$.
The proof proceeds in two steps: first we show feasibility and then optimality.
\smallskip

\noindent {\bf Feasibility.} We have $\pi' \ge 0$ by the choice of $\varepsilon$. To show $z'\ge
0$, note that for $e=(u,v)\in F'$, by the Claim and the dual constraints for $e \in E_f \cup E_n$ we have $0\le \pi_v-\pi_u+z(u,v)\le
z(u,v)-\varepsilon$, where for the second inequality we used that $\pi_u\ge \varepsilon$ and $\pi_v=0$. Thus,
$z(e)\ge \varepsilon$ and $z'(e) = z(e) - \varepsilon \geq 0$ for every $e\in F'$.  For an edge $e\not \in F'$, we have $z'(e) \geq z(e) \geq 0$ and so we can conclude that $z' \ge {0}$.
We also have $\pi'_v = 0$ for $v \in V(B)$ since $T \cap V(B) = \emptyset$.

It remains to verify the constraints for $(u,v) \in E_f \cup E_n$.
For every $(u,v)\in (F\cup F') \cap (E_f \cup E_n)$, as well as for edges not in $\delta(T)$,
we have $\pi'_v-\pi'_u+z'(u,v)=\pi_v-\pi_u+z(u,v)$, and therefore the
constraint remains valid.
Thus we may have $\pi'_v-\pi'_u+z'(u,v)\neq
\pi_v-\pi_u+z(u,v)$ in only two cases: either  {\bf (i)} if $(u,v)\in
\delta(T,V\setminus S)$, or {\bf (ii)} if $(u,v)\in
\delta(S\setminus T,T)$.

In case {\bf (i)}, the constraint on $(u,v)$ remains valid since  $\pi'_v-\pi'_u+z'(u,v)=\pi_v
-\pi_u+z(u,v)+\varepsilon$. In case {\bf (ii)}, $\pi'_u=0$,
$\pi'_v\ge 0$, and $z'(u,v)\ge 0$. Further, we have shown in the above
Claim that $(u,v)\in E_b\cup E_n$. There is no constraint for
$(u,v)\in E_b$, and $\pi'_v-\pi'_u+z'(u,v)\ge 0$ holds if $(u,v)\in E_n$.

\smallskip 

\noindent {\bf Optimality.}
When changing $(\pi,z)$ to $(\pi',z')$, the objective value
increases by $\varepsilon(x(F \cap (E_f \cup E_n)) - x(F' \cap (E_f \cup E_n))) = \varepsilon(x(F\setminus E_b)-x(F'))$;
we will show that this is non-positive.
Recall that $S$ is either a tight set or $V$, so that $x(\delta^-(S)) \le 1$. Since $T \subsetneq S$, we have 
$1\le x(\delta^-(S\setminus T))=x(\delta^-(S))-x(F)+x(F') \le 1 - x(F) + x(F')$, and
therefore $x(F')\ge x(F)\ge x(F\setminus E_b)$. Thus, the objective value does not increase,
therefore $(\pi',z')$ must be optimal.

The existence of the optimal solution $(\pi',z')$ with $\pi'(V) < \pi(V)$ contradicts
the choice of $(\pi,z)$, which completes the proof of
Lemma~\ref{lem:x-witness}.
\end{proof}

Let us state one more lemma that enables rounding fractional witness
flows to integer ones. Recall that in the proof of
Theorem~\ref{thm:lcapprox} for singleton instances a key step was to
round a fractional circulation to an integer one, a simple corollary
of the integrality of the network flow polyhedron. We will now need a
stronger statement as we need to round a circulation $z$ along with 
a witness flow $f$ consistently. We formulate the
following general statement (which does not
assume that we have a vertebrate pair or that $f$ is a witness flow).

\begin{lemma}\label{lem:TU}
For a digraph $G=(V,E)$ and edge weights $w:E\to \R$,
consider $z,f:E\to\R_+$ such that $z$ is a circulation and $f\le z$.
Then
there exist integer-valued vectors $\bar z, \bar f: E\to\Z_+$ with
$w^\top \bar z\le w^\top z$ satisfying the following properties:
\begin{enumerate}[(a)]
\item\label{round:circ} $\bar z$ is a circulation,
\item \label{round:balance} $\bar f(\delta^+(v))\ge \bar f(\delta^-(v))$ whenever $f(\delta^+(v))\ge  f(\delta^-(v))$,
\item\label{round:node-cap-f} $\lfloor  f(\delta^-(v))\rfloor\le \bar f(\delta^-(v))\le \lceil
  f(\delta^-(v))\rceil$ for every node $v\in V$,
\item \label{round:node-cap-g} $\lfloor  g(\delta^-(v))\rfloor\le \bar g(\delta^-(v))\le \lceil
  g(\delta^-(v))\rceil$ for every node $v\in V$, where $g=z-f$ and $\bar
  g=\bar z-\bar f$.
\item \label{round:cap}  
  $\bar f\le \bar z$; $\bar f_e=\bar z_e$ whenever $f_e=z_e$, and $\bar f_e=0$
  whenever $f_e=0$.
\end{enumerate}
\end{lemma}

The proof relies on the total unimodularity of network matrices. Let
$(V,E)$ be a directed graph and $T = (V,E_T)$ a directed tree on the
same node set. These define a {\em network matrix} $B\in
\mathbb{Z}^{|E_T|\times |E|}$ as follows. For every $e=(u,v)\in E$, let $P_e$
be the unique undirected $u-v$ path in the tree $T$. Then we set
$B_{e_T,e}=1$ if $e_T$ occurs in  forward direction in $P_{e}$,
$B_{e_T,e}=-1$ if $e_T$ occurs in the backward direction in $P_{e}$, and
$B_{e_T,e}=0$ if $e_T$ does not occur in $P_e$. We will use the
following result (see e.g.~\cite[(34) in Sec 19.3]{SchrijverLPIP}):

\begin{theorem}[\cite{Tutte65}] \label{thm:TU}
Every network matrix is totally unimodular.
\end{theorem}
\newcommand{\blower}{b^{\mathrm{lower}}}
\newcommand{\bupper}{b^{\mathrm{upper}}}
\newcommand{\dlower}{d^{\mathrm{lower}}}
\newcommand{\dupper}{d^{\mathrm{upper}}}

Together with \cite[Theorem 19.1 and (4) in Sec 19.1]{SchrijverLPIP},  it follows that 
any LP of the form $\min \{ c^\top x : \blower\le Bx\le \bupper, \ell\le x\le u\}$ has an
integer optimal solution for any integer vectors $\blower,\bupper\in \mathbb{Z}^{|E_T|}$ and $\ell,u\in \mathbb{Z}^{|E|}$.
A fundamental
example of network matrices is the incidence matrix of a directed graph
$G=(V,E)$. Consider the digraph
$(V\cup\{r\},E)$ obtained by adding a new vertex $r$ and the tree $T=(V\cup
\{r\}, \{(r,u): u\in V\})$ that forms a star. The associated network
matrix is identical to the incidence matrix of $G$. The proof below
extends this simple construction to capture
all degree requirements.

\begin{proof}[Proof of Lemma~\ref{lem:TU}]
Let us introduce new variables $\bar g=\bar z-\bar f$ (the notation already
used in property~\ref{round:node-cap-g}). The integrality of $\bar z$ is equivalent to the
integrality of $\bar f$ and $\bar g$.

Let us construct a network matrix $B\in \mathbb{Z}^{4|V|\times 2|E|}$
as follows. We define a tree
$T=(V',E_T)$, where $V'$ contains a
root node $r$, and for each $v\in V$, we include four nodes $v^g$,
$v^f$, $v^{g'}$, $v^{f'}$ and four edges $(r,v^g)$, $(v^g,v^f)$,
$(v^f,v^{f'})$ and $(v^{g},v^{g'})$ in $E_T$. Let us define
the directed graph $(V',E')$, where $E'=E_f'\cup E_g'$ is obtained as follows: for each
$(u,v)\in E$, we add an edge $(u^f, v^{f'})$ to $E_f'$ and an edge $(u^g ,v^{g'})$ to $E_g'$.
Let $B$ be the network matrix corresponding to $(V',E')$ and $T$ (see Figure~\ref{fig:network_matrix} for an example). That is,
$B$ is an $|E_T|\times |E'|$ matrix such that 
\begin{itemize}
\item the column corresponding to $(u^f ,v^{f'})\in E_f'$ has a $-1$ entry
  in the rows corresponding to the edges $(r,u^g)$ and $(u^g,u^f)$, and $+1$ entries
  corresponding to the edges $(r,v^g)$, $(v^g,v^{f})$, and $(v^f,v^{f'})$; all other
  entries are $0$;
\item the column corresponding to $(u^g, v^{g'})\in E_g'$ has a $-1$ entry
  in the row corresponding to the edge $(r,u^g)$, and $+1$ entries
  corresponding to the edges $(r,v^g)$ and $(v^g,v^{g'})$; all other
  entries are $0$.
\end{itemize}

\begin{figure}[t!]
\centering

\subfloat {
\begin{tikzpicture}[scale=0.85]

	\node[sgvertex,minimum size=17pt] (r) at (0,0) {$r$};
	
	\begin{scope}
		\node[sgvertex,minimum size=17pt] (vg) at (0,-2) {$v^g$};
		\node[sgvertex,minimum size=17pt] (vf) at (-1,-4) {$v^f$};
		\node[sgvertex,minimum size=17pt] (vfp) at (0,-6) {$v^{f'}$};
		\node[sgvertex,minimum size=17pt] (vgp) at (1,-4) {$v^{g'}$};
		\draw (r) edge[->] node[] {} (vg);
		\draw (vg) edge[->] node[] {} (vf);
		\draw (vf) edge[->] node[] {} (vfp);
		\draw (vg) edge[->] node[] {} (vgp);
	\end{scope}
	
	\begin{scope}[xshift=4.5cm]
		\node[sgvertex,minimum size=17pt] (ug) at (0,-2) {$u^g$};
		\node[sgvertex,minimum size=17pt] (uf) at (-1,-4) {$u^f$};
		\node[sgvertex,minimum size=17pt] (ufp) at (0,-6) {$u^{f'}$};
		\node[sgvertex,minimum size=17pt] (ugp) at (1,-4) {$u^{g'}$};
		\draw (r) edge[->] node[] {} (ug);
		\draw (ug) edge[->] node[] {} (uf);
		\draw (uf) edge[->] node[] {} (ufp);
		\draw (ug) edge[->] node[] {} (ugp);
	\end{scope}
	
	\begin{scope}[xshift=-4.5cm]
		\node[sgvertex,minimum size=17pt] (wg) at (0,-2) {$w^g$};
		\node[sgvertex,minimum size=17pt] (wf) at (-1,-4) {$w^f$};
		\node[sgvertex,minimum size=17pt] (wfp) at (0,-6) {$w^{f'}$};
		\node[sgvertex,minimum size=17pt] (wgp) at (1,-4) {$w^{g'}$};
		\draw (r) edge[->] node[] {} (wg);
		\draw (wg) edge[->] node[] {} (wf);
		\draw (wf) edge[->] node[] {} (wfp);
		\draw (wg) edge[->] node[] {} (wgp);
	\end{scope}
	
	\draw (vg) edge[->, bend right=20, dashed] node[] {} (wgp);
	\draw (ug) edge[->, bend right=20, dashed] node[] {} (vgp);
	\draw (vf) edge[->, bend left=20, dashed] node[] {} (wfp);
	\draw (uf) edge[->, bend left=20, dashed] node[] {} (vfp);

\end{tikzpicture}
}

\subfloat {
\begin{tabular}{l|ll:ll|l}
	part of $v$'s: & inc.~$f$-flow & inc.~$g$-flow & out.~$f$-flow & out.~$g$-flow & corresponds to \\
	& $(u^f, v^{f'})$ & $(u^g, v^{g'})$ & $(v^f, w^{f'})$ & $(v^g, w^{g'})$ & $v$'s constraint: \\
	\hline
	$(r, v^g)$ & 1 & 1 & -1 & -1 & \ref{round:circ} \\
	$(v^g, v^f)$ & 1 & 0 & -1 & 0 & \ref{round:balance} \\
	$(v^f, v^{f'})$ & 1 & 0 & 0 & 0 & \ref{round:node-cap-f} \\
	$(v^g, v^{g'})$ & 0 & 1 & 0 & 0 & \ref{round:node-cap-g} \\
	\hdashline
	$(r,u^g)$ & -1 & -1 & 0 & 0 & \\
	$(u^g, u^f)$ & -1 & 0 & 0 & 0 & \\
	$(u^f, u^{f'})$ & 0 & 0 & 0 & 0 & \\
	$(u^g, u^{g'})$ & 0 & 0 & 0 & 0 & \\
	\hline
\end{tabular}
}

\caption{
	The network matrix in the proof of Lemma~\ref{lem:TU}.\\
	\textbf{Top:} fragment of the tree $T = (V',E_T)$ corresponding to three vertices $u$, $v$, $w$ and two edges $(u,v)$, $(v,w)$ (whose images are dashed).\\
	\textbf{Bottom:} fragment of the network matrix $B$ corresponding to vertices $u$, $v$ and edges $(u,v)$, $(v,w)$,
illustrating how rows of $B$ encode the degree requirements on vertex $v$.
}
\label{fig:network_matrix}
\end{figure}
 
Then $B$ is a totally unimodular matrix according to Theorem~\ref{thm:TU}.
Using the variables $\bar f$ and $\bar g$, the system
\ref{round:circ}-\ref{round:cap} can be equivalently written in the form
\[
\begin{aligned}
\blower \le B (\bar f,\bar g) &\le \bupper \\
\dlower \le (\bar f,\bar g) &\le \dupper,
\end{aligned}
\]
where $\blower,\bupper,\dlower,\dupper$ are vectors with integer or infinite entries.
(Here, $(\bar f, \bar g) \in \R^{|E|+|E|}$ is a column vector, and $B(\bar f, \bar g)$ the matrix-vector product.)
Namely, the variable $\bar f(u,v)$ corresponds to the column for
$(u^f,v^{f'})\in E_f'$, and the variable $\bar g(u,v)$ corresponds to the column for
$(u^g,v^{g'})\in E_g'$. See also Figure~\ref{fig:network_matrix}.
\begin{itemize}
\item The rows for $(r,v^g)\in E_T$ describe the
flow conservation constraints as in \ref{round:circ} with
$\blower(r,v^g)=\bupper(r,v^{g})=0$.
\item The rows for $(v^g,v^f)\in E_T$ describe the constraints as in
\ref{round:balance}, with $\blower(r,v^f)=-\infty$ and $\bupper(r,v^f)=0$ for those $v \in V$ with
$f(\delta^+(v))\ge f(\delta^-(v))$, and $\bupper(r,v^f)=+\infty$ for other $v \in V$.
\item The rows for $(v^f,v^{f'})\in E_T$ describe the constraints as
  in \ref{round:node-cap-f} with $\blower(v^f,v^{f'})=\lfloor
  f(\delta^-(v))\rfloor$ and $\bupper(v^f,v^{f'})=\lceil
  f(\delta^-(v))\rceil$.
\item The rows for $(v^g,v^{g'})\in E_T$ describe the constraints as
  in \ref{round:node-cap-g} with $\blower(v^g,v^{g'})=\lfloor
  g(\delta^-(v))\rfloor$ and $\bupper(v^f,v^{f'})=\lceil
  g(\delta^-(v))\rceil$ (recall that $g=z-f$).
\item The capacity constraints $\dlower \le (\bar f,\bar g) \le \dupper$ are used to
describe \ref{round:cap}.
\end{itemize}

With $g=z-f$, the vector $(f,g)$ is a feasible solution to
the system. Consider the cost function $(w,w)\in \R^{2|E|}$. 
Using total unimodularity, there must be an integer solution $(\bar
f,\bar g)$ with $(w,w)^\top (\bar f,\bar g)\le (w,w)^\top (f,g)=w^\top
z$.
Thus, $\bar z=\bar f+\bar g$ and $\bar f$ satisfy the statement of the lemma.
\end{proof}

 Equipped with 
the above lemmas, we are ready to prove
Lemma~\ref{lem:mainLCATSP}.

\begin{proof}[Proof of Lemma~\ref{lem:mainLCATSP}]
We are given a vertebrate pair $(\cI,B)$ and disjoint non-empty vertex
sets 
$U_1, \ldots, U_{\ell} \subseteq V\setminus V(B)$ such
  that the subgraphs $G[U_1], \ldots, G[U_{\ell}]$ are strongly connected and for every $S\in \cL_{\geq 2}$ and $i = 1, \ldots, \ell$, we have either
  $U_i \cap S = \emptyset$ or $U_i \subseteq S$.
For the Held-Karp solution $x$, let $f$ be the witness flow guaranteed
by Lemma~\ref{lem:x-witness}.

We can assume that each edge $e$ has either $f(e)=0$ or $f(e)=x(e)$
without loss of generality, by breaking $e$ up into two parallel
copies and dividing its $x$ and $f$ values between them
appropriately.
We say that an edge $e$ is {\em marked} if $f(e)=x(e)$ and
  {\em unmarked} otherwise.

Let us now introduce a convenient decomposition of $x$
that we will obtain and utilize in our algorithm.
 By a {\em 2-cycle} we mean a closed walk
  that visits every vertex at most twice and contains every edge at
  most once.
A 2-cycle $C\subseteq E$ is {\em consistent} if, for any two
consecutive edges $(u,v), (v,v')\in C$ with $v\notin V(B)$, if 
$(u,v)$ is marked then $(v,v')$ is also marked.
\begin{claim}\label{claim:cycle-decomp}
We can in polynomial time decompose $x$  into consistent 2-cycles. That
is, there is a polynomial-time algorithm that outputs consistent 2-cycles $C_1,C_2,\ldots,C_\ell$
and multipliers $\lambda_1,\lambda_2,\ldots,\lambda_\ell\ge 0$ such
that $x=\sum_{i=1}^\ell \lambda_i \one_{C_i}$, where $\one_{C}\in \{0,1\}^E$ denotes the indicator vector of a 2-cycle $C$.
\end{claim}
\begin{proof}
The proof uses a variant of the standard cycle decomposition argument. We identify
2-cycles by constructing walks on edges in the support of
$x$. The algorithm identifies and removes 2-cycles one by one. After
the removal of every cycle, we maintain property \ref{flow:conservation} of
witness flows: for every $v\in V \setminus V(B)$ the outgoing flow
amount on
marked edges is greater or equal to the incoming flow amount on marked
edges.

We find walks using the following procedure. We start on an arbitrary
(marked or unmarked) edge. 
If the walk uses a marked edge then let us select the next edge as a
marked edge whenever
possible, and after an unmarked edge let us continue on an
unmarked edge if possible. We terminate according to the following
rules:
\begin{enumerate}
\item If we visit a node $v\in V(B)$ the second time, then we terminate
  the current walk. We select $C$ as the segment of the walk between the
  first and second visits to $v$.
  (If we started the walk from $v$, we also count this as a visit.)
\item When visiting a node $v\in V\setminus V(B)$ for the second time, \label{cond:rule2}
  we terminate if
  the outgoing edge of the first visit
  and
  the incoming edge of the second visit
  are of the same type
  (marked or unmarked).
  We also terminate if every outgoing edge of $v$ is marked.
  We select $C$ as the segment of
  the walk between the first and second visits to $v$.
\item If we visit a node in $v\in V\setminus V(B)$ the third time,
  then we always terminate. Let $(v,u_1)$ and $(v,u_2)$ be the edges where
  the walk left $v$ for the first and second time, and let $(z_1,v)$
  and $(z_2,v)$ be the edges where we arrived to $v$ the second and
  third times. If $(v,u_2)$ and $(z_2,v)$ are both unmarked or both
  marked, then we let $C$ be the segment of the walk between
  $(v,u_2)$ and $(z_2,v)$. Otherwise, we let $C$ be the segment of the
  walk between $(v,u_1)$ and $(z_2,v)$ (visiting $v$ twice).
  (These edges are of the same type,
  and so are $(z_1,v)$ and $(v,u_2)$.)
\end{enumerate}
Given $C$, we let
$\lambda$ denote the minimum flow value on any edge of $C$. We
decrease the value of $x$ on every edge of $C$ by $\lambda$
(deleting an edge if its $x$-value becomes zero), and also
the value of $f$ on every marked edge of $C$. We add $C$ to the
decomposition with coefficient $\lambda$
and proceed to finding the
next 2-cycle.

Assume property \ref{flow:conservation} of witnessed flows holds when
we start the walk. Thus, whenever the walk enters a node $v\in
V\setminus V(B)$ on a marked edge, it can continue on a marked edge \bll{(or terminate, if it is the second visit, and the first outgoing edge was marked)}.
The above rules guarantee that $C$ will be a consistent 2-cycle: if we
close $C$ in a node $v\in V\setminus V(B)$, then if the incoming edge
is marked, the outgoing edge will also be marked.
Moreover, if the walk is not terminated, then it can continue on a yet-unused edge.

It remains to show that property
\ref{flow:conservation}  is maintained after removing $C$. Assume we
decrease the outgoing flow at a node $v\in V\setminus V(B)$ on a
marked edge. First, notice that there can be at most one outgoing
marked edge from $v$ on $C$. If there was also an incoming marked edge on $C$, then
the marked flow decreases by the same amount on these two
edges. If all incoming edges were unmarked, then our rules
implies that all outgoing edges at $v$ are marked. In this case,
\ref{flow:conservation}  is trivially maintained
(as $x$ is, and remains, a circulation).

The proof can be immediately turned into a polynomial-time
algorithm. Notice that the number of edges decreases by at least one
at the removal of every cycle.
\end{proof}

\begin{figure}[t]
	\centering
	\begin{tabular}{lllll}
		 & graph & integral & obtained from the previous by \\
		\hline
		$x$      & $G$ & no & input to Lemma~\ref{lem:mainLCATSP} \\
		$x, f$   & $G$ & no & finding a witness flow (Lemma~\ref{lem:x-witness}) \\
		$x', f'$ & $G'$ & no & introducing $a_i$, redirecting edges, subtracting $x_i$ and $f_i$ \\
		$\bar z', \bar f'$ & $G'$ & yes & rounding (Lemma~\ref{lem:TU} applied to $z=2x'$ and $2f'$) \\
		$\bar z, \bar f$ & $G$ & yes & un-redirecting edges, mapping back to $G$ \\
		$z^*, f^*$ & $G$ & yes & adding walks $P_i$ (to $\bar f$ only if $\bar f'(\delta^-(a_i)) = 1$)
\end{tabular}
	\caption{This table summarizes the various circulations and flows that appear in our algorithm, in order.}
	\label{tab:flows}
\end{figure}

We will construct an auxiliary graph $G'=(V',E')$ similarly as in the proof of
Theorem~\ref{thm:lcapprox}. 
For convenience, Figure~\ref{tab:flows} gives an overview of the different steps, graphs and flows used by our algorithm.
We select edge sets $X_i^-$, $X_i^+$ and
flows $x_i$ and $ f_i$ as follows:
\begin{itemize}
\item For every $U_i$ we can select
(possibly by subdividing edges as above)
a subset of incoming edges
$X_i^- \subseteq \delta^-(U_i)$ with
  $x(X_i^-)= \nicefrac12$ such that either all edges $e\in X_i^-$ are marked
  or all are unmarked.
  This is possible since $x(\delta^-(U_i)) \ge 1$ (and all edges are either marked or unmarked).

\item Take the decomposition of $x$ into 2-cycles as guaranteed by Claim~\ref{claim:cycle-decomp},
and follow the incoming edges in $X^-_i$ in the decomposition.
We let $X^+_i$ be the set of edges on which these walks first leave $U_i$
after entering on an edge in $X^-_i$.
We let $x_i$ denote the respective $x$-flow on the 
  segments of these walks connecting  the heads of edges in $X_i^-$ and the tails
  of edges in $X_i^+$. On the same edges, we define $f_i$ by
  $f_i(e)= x_i(e)$ for every marked edge and $f_i(e)=0$
  for every unmarked edge.
\end{itemize}

Note that we have $0 \le f - f_i \le x - x_i$. \bl{We further claim that 
$(f - f_i)(\delta^+(v)) \ge (f - f_i)(\delta^-(v) \setminus X_i^-)$ for every $v \in V
\setminus V(B)$.
To see this, consider the consistent 2-cycles in the decomposition. In each of these $2$-cycles an incoming marked edge must be followed by an outgoing marked edge when considering a vertex $v\in V \setminus V(B)$. This gives a (fractional) pairing of incoming and outgoing edges of $v$ such that each incoming marked edge is paired with an outgoing marked edge.   By our selection of $X_i^-$ and $x_i$ using the consistent $2$-cycles, we have that the marked incoming edges that were not used by $x_i$ or $X_i^-$ are still paired with marked outgoing edges that were not used by $x_i$. We thus have $(f - f_i)(\delta^+(v)) \ge (f - f_i)(\delta^-(v) \setminus X_i^-)$ for every $v \in V
\setminus V(B)$ as claimed. }

We now transform $G$ into a new graph $G'$, $x$ into a new circulation
$x'$, and $f$ into a new $f'$, as follows. For every $i=1,\ldots, \ell$, we introduce a
new auxiliary vertex $a_i$ and redirect
all edges in $X_i^-$ to point to $a_i$, and those in $X_i^+$ to point from
$a_i$.  We  subtract the flow $x_i$ from $x$ and $f_i$ from
$f$ inside $U_i$; hence the resulting
vector $x'$ will be a circulation in $G'$, and $f'\le x'$.
We now have
$f'(\delta^+(v)) \ge f'(\delta^-(v))$ for all $v \in V \setminus V(B)$ \bl{(using that $(f - f_i)(\delta^+(v)) \ge (f - f_i)(\delta^-(v) \setminus X_i^-)$ for every $v \in V
\setminus V(B)$ as explained above)}.
We also have 
$x'(\delta^-(a_i)) = x'(\delta^+(a_i)) =\nicefrac12$, and either
$f'(\delta^-(a_i))=0$ (in the case when $f(X^-_i) = 0$)
or
$f'(\delta^-(a_i))=f'(\delta^+(a_i))=\nicefrac12$ (in the case when
$f(X^-_i) = \nicefrac12$: since in this case all edges in $X_i^+$ must
also be marked \bl{due to the facts that the $2$-cycles are consistent and $U_i \cap V(B) = \emptyset$)}.
 We define the
weights $w'(e)=w(e)$ if $e$ was not modified, and for every
redirected edge $e$, we set $w'(e)$ as the weight of the edge it was redirected from.
Thus, the total weight may only decrease: $\sum_{e\in
  E'}w'(e)x'(e)\leq \sum_{e\in E}w(e)x(e)=\valu(\cI)$ (it decreases if
the flows $x_i$ are nonzero on edges with positive weight). 

Let us now apply Lemma~\ref{lem:TU} in the graph $G'$ with weights
$w'$ and $z=2x'$ and
$2f'$ in place of $f$. We thus obtain the integer vectors $\bar z'$
and $\bar f'$ with ${w'}^\top \bar z'\le {w'}^\top z$.  Note in
particular that properties \ref{round:node-cap-f},~\ref{round:node-cap-g} 
imply that
$\bar z'(\delta^-(a_i))=1$ for every auxiliary vertex $a_i$;
this is because
$z(\delta^-(a_i))=1$
and
$2f'(\delta^-(a_i)) \in \{0,1\}$.

Now we map $\bar z'$ and $\bar f'$ back to the vectors $\bar z$ and
$\bar f$ in the original graph $G$. Namely, if $e$ is incident to an
auxiliary vertex $a_i$, then we reverse the redirection of this edge,
and move $\bar z$ and $\bar f$ to the original graph. 

The vector $\bar z$ may not be a circulation. Specifically, since the in- and
out-degree of $a_i$ were exactly $1$ in $\bar z'$, in each component $U_i$ there
is a pair of vertices $u_i$, $v_i$ which are the head and tail,
respectively, of the mapped-back edges adjacent to $a_i$. These are the only
vertices whose in-degree may differ from their out-degree. (They differ unless
$u_i = v_i$.) To repair this, for each $i = 1,\ldots, \ell$ with
$u_i\neq v_i$, we route a path $P_i$ from
$u_i$ to $v_i$ in $U_i$; this is always possible as we assumed that
$U_i$ is strongly connected.

 We obtain a circulation $z^*$ from $\bar
z$ by increasing the value by 1 on every edge of every such path $P_i$. Further, we
obtain $f^*$ from $\bar f$ as follows. If $\bar f'(\delta^-(a_i))=0$,
then we let $f^*$ be identical to $\bar f$ inside $U_i$. If $\bar
f'(\delta^-(a_i))=1$, then we increase the value of $\bar f$ by 1 on
every edge of $P_i$. 

\begin{claim} \label{claim:bound_on_wtz}
We have $w^\top z^*\le 2\valu(\cI)+\lb_{\cI}(\bar B)$, and $f^*$ is a
witness flow for $z^*$.
\end{claim}
\begin{proof}
By the construction, $w^\top \bar z = w'^\top \bar z' \le w'^\top z = 2 w'^\top x' \le 2 w^\top x = 2\valu(\cI)$. We obtained
$z^*$ from $\bar z$ by increasing it on a set of disjoint paths $P_i$ in
$V\setminus V(B)$, and these paths do not cross any set in $\cL_{\ge
  2}$. Thus, their total cost is bounded by $\lb_{\cI}(\bar B) =2\sum_{v\in V\setminus
  V(B)} y_v$.

\bl{We next verify that $f^*$ is a witness flow for $z^*$. Property 
  \ref{flow:capacities}, i.e., that $f^* \leq z^*$ is clear from the construction. We proceed to verify properties \ref{flow:backward} and \ref{flow:forward}, which state that $f^*(e) = 0$ for each backward edge $e$ and $f^*(e) = z^*(e)$ for each forward edge, respectively. This holds for the initial witness flow $f$. We are now going to use that all edges in $G[U_i]$ are neutral since for every $S\in \cL_{\geq 2}$ either $U_i \cap S = \emptyset$ or $U_i \subseteq S$.  It follows that $f'$ and $f$ only differ in neutral edges and hence $f'$ also satisfies properties \ref{flow:backward} and \ref{flow:forward}. By Lemma~\ref{lem:TU}, we thus have that $\bar f'$ also satisfies these properties. Finally, $f^*$ is obtained from $\bar f'$  by reversing the redirection of the edges incident to the auxiliary vertices and potentially increasing the flow value on neutral edges. As these operations maintain the properties, we have that $f^*$ satisfies properties \ref{flow:backward} and \ref{flow:forward}.}

Finally, for  \ref{flow:conservation}, recall that we had
$f'(\delta^+(v)) \geq  f'(\delta^-(v))$ for every $v\in V \setminus
V(B)$, and hence the same holds for $\bar f'$ by
Lemma~\ref{lem:TU}. For $\bar f$, this condition can be violated only
at some nodes $u_i$ for some values of $i=1,2,\ldots,\ell$ \bl{for which $\bar f' (\delta^-(a_i)) = \bar f'(\delta^+(a_i))  = 1$ and $u_i \neq v_i$}. We
obtain $f^*$ from $\bar f$ by increasing the flow value  on the  path
$P_i$ from $u_i$ to $v_i$ for such $i$; this increases $f^*(\delta^+(u_i))$ so that $f^*(\delta^+(u_i))
\geq  f^*(\delta^-(u_i))$, while not violating $f^*(\delta^+(v_i))
\geq  f^*(\delta^-(v_i))$ \bl{since $\bar f' (\delta^-(a_i)) = \bar f'(\delta^+(a_i))  = 1$ implies} $\bar f(\delta^+(v_i))
\geq  \bar f(\delta^-(v_i))+1$.
\end{proof}

Let $F$ be the Eulerian edge multiset obtained by taking $z^*_e$
copies of edge $e$. The proof concludes by showing that $F$ satisfies
all requirements of Lemma~\ref{lem:mainLCATSP}. The cost bound~\ref{lc:cost} follows by
Claim~\ref{claim:bound_on_wtz} since $w(F)=w^\top z^*$. The connectivity requirement~\ref{lc:cross},
namely that $|\delta^-_F(U_i) | \geq 1$ for every $i
        = 1,\ldots, \ell$, is immediate by the construction.

 Let us now
        show \ref{lc:degree}, that is, $|\delta^-_F(v)| =z^*(\delta^-(v))\leq 4$  for every
        $v\in V$ such that $x(\delta^-(v))=1$.
        Note that we have $x'(\delta^-(v)) \le 1$.
        Denote
        $g' = x' - f'$
        and
        $\bar g' = \bar x' - \bar f'$.
        By properties
        \ref{round:node-cap-f}, \ref{round:node-cap-g}
        of
        Lemma~\ref{lem:TU}
        we have
        \[
        \bar z'(\delta^-(v)) = \bar f'(\delta^-(v)) + \bar g'(\delta^-(v))
        \le
        \ceil{2f'(\delta^-(v))} + \ceil{2g'(\delta^-(v))}
        \le 3 \,, \]
        as the maximum value of the function
        $\ceil{2p} + \ceil{2q}$
        subject to $p + q \le 1$
        is $3$.
        Adding the paths $P_i$ \bl{and reversing the redirection of edges incident to auxiliary vertices $a_i$} may further
increase $z^*(\delta^-(v))$ over $\bar z'(\delta^-(v))$ by $1$.

Property
   \ref{lc:nobadcycle} requires that every subtour in $F$ that crosses a tight
     set in $\cL_{\geq 2}$ visit a vertex of the backbone. This 
     follows from Lemma~\ref{lem:witness-cross} since $z^*$ is a witnessed circulation.
\end{proof}

 \section{Completing the Puzzle: Proof of Theorem~\ref{thm:constantATSP}}
\label{sec:completepuzzle}
We now combine the techniques and algorithms of the previous sections
to obtain a constant-factor approximation algorithm for ATSP. In
multiple steps, we have reduced ATSP to finding
tours for vertebrate pairs. Every reduction step was polynomial-time and increased the
approximation ratio by a constant factor.  Hence, together they give a
constant-factor approximation algorithm for ATSP.

We now give an overview of these reductions and set the parameters. 
 Throughout,
$\varepsilon>0$ will be a fixed small value. We set $\delta=\deltaval$.
All approximation guarantees are with respect to the optimum value of
the Held-Karp relaxation $\LP(G,w)$.
The reduction proceeds using the following algorithmic subroutines.
\begin{itemize}
\item Corollary~\ref{cor:singleton} provides a polynomial-time
  $\nw$-approximation $\cA_{\mathrm{S}}$
  for singleton instances, with $\nw=18+\epsilon$.
 This will be used to find a (quasi-)backbone for
  irreducible instances (Lemma~\ref{lem:skeleton}).
\item Algorithm $\cA_{\mathrm{ver}}$, which, for a vertebrate pair
  $(\cI,B)$, finds a tour of cost $\kappa \valu(\cI)+\eta \lb_B(V)+w(B)$, where
  $\kappa= 2$ and $\eta=37+36\varepsilon$ (Corollary~\ref{cor:finish_skeleton}).
Algorithm $\cA_{\mathrm{ver}}$ uses the reduction of
  Theorem~\ref{thm:LoctoGlo} from \EPC{} to  ATSP.
\item Algorithm $\cA_{\mathrm{irr}}$ which, provided
  $\cA_{\mathrm{ver}}$ as above, obtains a
    polynomial-time $\rho$-approximation algorithm for irreducible instances, where 
$\rho=(\kappa +\eta(1-\delta)+\nw+3)/(2\delta-1)<55.61$ for
sufficiently small $\varepsilon>0$ (Theorem~\ref{thm:fromSimpletoGeneral}).
\item Algorithm $\cA_{\mathrm{lam}}$, which converts the
  $\rho$-approximation algorithm $\cA_{\mathrm{irr}}$ to a
  $2\rho/(1-\delta)$-approximation algorithm for an arbitrary
  laminarly-weighted instance $\cI$. Here,  $2\rho/(1-\delta)<\finalval$ (Theorem~\ref{thm:reduction_to_irreducible}).
\item Our final Algorithm $\cA_{\mathrm{ATSP}}$, which reduces an arbitrary
  input   weighted digraph $(G,w)$ to a laminarly-weighted
  instance, keeping the same approximation ratio (Theorem~\ref{thm:laminar}).
\end{itemize}

All in all we have thus obtained a polynomial-time  algorithm for  ATSP that returns a tour of value at most $\finalval$ times the Held-Karp lower bound.

\paragraph{Integrality gap}
\cref{thm:constantATSP} of course implies an upper bound of $\finalval$
on the integrality gap of the Held-Karp relaxation.
However, if we do not require a polynomial-time algorithm,
then the loss factor of $9(1+\varepsilon)$ in the reduction of \cref{thm:LocalToGlobal}
can be decreased to $5$.
Therefore non-constructively we can have
$\nw' = 10$
and
$\eta' = 1 + 5 \cdot 4$
(instead of $\eta = 1 + 9 \cdot (1 + \varepsilon) \cdot 4$),
which yields $\rho' < 35.04$
and a final integrality gap of at most $\finalintegralitygap$.
\begin{theorem} \label{thm:constantATSPintegralitygap}
  The integrality gap of the asymmetric Held-Karp relaxation is at most $\finalintegralitygap$.
\end{theorem}

\paragraph{Asymmetric Traveling Salesman \emph{Path} Problem}

In the Asymmetric Traveling Salesman Path Problem (ATSPP),
in addition to the usual ATSP input,
we are also given two special vertices $s, t \in V$
and wish to find a walk that visits all vertices,
starts from $s$, and ends at $t$.
Feige and Singh~\cite{FeigeS07}
proved that if there is a $\beta$-approximation algorithm for ATSP,
then there is a $((2+\epsilon)\beta)$-approximation algorithm for ATSPP
for any $\epsilon > 0$. Together with \cref{thm:constantATSP}, this implies:
\begin{corollary}
  There is a polynomial-time algorithm for ATSPP that returns a tour of
  value at most $\finalvalFS$ times the integral optimum.
\end{corollary}
Note that the approximation ratio here is not bounded in terms of the Held-Karp lower bound\footnote{
For ATSPP, this is defined as the optimal value of a relaxation that is similar to $\LP(G,w)$, with the differences that $x(\delta^+(s)) - x(\delta^-(s)) = 1$, $x(\delta^-(t)) - x(\delta^+(t)) = 1$, and the cuts $S$ with $s\in S$ are only required to have at least one outgoing edge (but possibly no incoming edge).}.
Köhne, Traub and Vygen~\cite{KohneTV18} give a similar reduction as Feige and Singh,
but for integrality gaps,
with a loss of $\beta \mapsto 4 \beta - 3$.
Together with \cref{thm:constantATSPintegralitygap}, this implies:
\begin{corollary}
  The integrality gap of the asymmetric Held-Karp relaxation for ATSPP is at most $\finalintegralitygapKTV$.
\end{corollary}
Finally, since their reduction is constructive, together with \cref{thm:constantATSP} we get:
\begin{corollary}
  There is a polynomial-time algorithm for ATSPP that returns a tour of
  value at most $\finalvalKTV$ times the Held-Karp lower bound.
\end{corollary}

\section{Conclusion}
\label{sec:conclusion}

In this paper we gave the first constant-factor approximation algorithm for ATSP. The result was obtained in three steps. First, we gave a generic reduction from ATSP to \EPC{}. Then, we showed how to simplify general ATSP instances to very structured instances (vertebrate pairs) by only incurring a constant-factor loss in the approximation guarantee. Finally,  these instances were solved using the connection to \EPC{}.

Subsequent to our work, Traub and Vygen \cite{Traub2020} have attained a substantially better approximation guarantee $22+\varepsilon$, with a matching integrality gap $22$. The main improvement in the approximation factor is due to eliminating the reduction to irreducible instances, and using a variant of Subtour Partition Cover for arbitrary instances in a careful recursive framework. 
It remains open whether one can get an approximation algorithm (nearly) matching the integrality gap of the  Held--Karp relaxation, and whether the current lower bound on the gap is tight.

\begin{openquestion}
 Is the integrality gap of the standard LP relaxation upper-bounded by $2$?
\end{openquestion}
We remark that the above question is also open for the simpler case of unweighted instances, where the best known upper bound is $13$~\cite{Svensson15}.

As mentioned in the introduction, Asadpour et al.~\cite{AsadpourGMGS10}
introduced a different approach for ATSP based on so-called thin spanning
trees. Our algorithm does not imply a better construction of such trees and the
$O(\textrm{poly}\log\log n)$-thin trees of~\cite{AnariG15} remain the best such
(non-constructive) result. Whether trees
of better thinness exist is an interesting question. Also, as shown in~\cite{An2010}, the construction of
$O(1)$-thin trees would lead to a constant-factor approximation algorithm for the
bottleneck ATSP problem. There, we are given a
complete digraph with edge weights satisfying the triangle inequality,
and we wish to find
a Hamiltonian cycle that minimizes the maximum edge weight. A tight
$2$-approximation algorithm for bottleneck \emph{symmetric} TSP was given
already in~\cite{FLEISCHNER197429,Lau1981,PARKER1984269}, but no constant-factor approximation is known for
bottleneck ATSP. 
\begin{openquestion}
 Is there a $O(1)$-approximation algorithm  for bottleneck ATSP? 
\end{openquestion}
We believe that this is an interesting open question in
itself, and progress on it may shed light on the existence of $O(1)$-thin trees.

\section*{Acknowledgments}

The authors are grateful to Johan H\aa stad and Hyung-Chan An for inspiring
discussions and valuable comments that influenced this work, and to
Jens Vygen for useful feedback on the  manuscript.
We are grateful to Andr\'as Frank for pointing us to the reference \cite{Karzanov}.
We are also grateful to the Simons Foundation which has sponsored several workshops and programs where numerous inspiring discussions have taken place.

\bibliographystyle{alpha}
{\small \bibliography{../atsp}}

\end{document}